\documentclass[a4paper,10pt]{article}
\usepackage[bbgreekl]{mathbbol}
\usepackage{soul}
\usepackage[small]{titlesec}
\usepackage[utf8]{inputenc}
\usepackage[T1]{fontenc}  
\usepackage[]{mdframed}
\usepackage{mathrsfs}
\usepackage{caption}
\usepackage[noerroretextools,backend=biber,style=alphabetic,natbib=false,giveninits=true,doi=true,isbn=false,url=false,date=year,maxbibnames=99,sorting=nty,defernumbers=true]{biblatex}
\DeclareNameAlias{default}{family-given}

\DeclareFieldFormat[article]{title}{#1}
\renewbibmacro{in:}{}
\DeclareFieldFormat[article]{pages}{#1}
\DeclareFieldFormat{journal}{#1}
\renewcommand\bf\bfseries

\renewbibmacro*{volume+number+eid}{%
	\printfield{volume}%
	\setunit*{\addnbspace}% NEW (optional); there's also \addnbthinspace
	\printfield{number}%
	\setunit{\addcomma\space}%
	\printfield{eid}}
\DeclareFieldFormat[article]{number}{\mkbibparens{#1}}
\DeclareFieldFormat[article]{volume}{\textbf{#1}}
\DeclareFieldFormat{year}{\mkbibparens{#1}}
\DeclareBibliographyDriver{article}{%
	\printnames{author}:%
	\newunit\newblock
	\printfield{title}%
	\newunit\newblock
	\printfield{journaltitle}
	\newunit
	\iffieldundef{number}{\printfield{volume}}{\printfield{volume}\addspace\printfield{number}}
	\addcomma\addspace\printfield{pages}\addspace
	\printfield{year}
}
\AtEveryBibitem{\clearfield{month}}
\AtEveryBibitem{\clearfield{day}}
\usepackage[english]{babel}
\usepackage[T1]{fontenc}
\usepackage{csquotes}
\usepackage{bbm}
\usepackage[leqno]{amsmath}
\usepackage{amsfonts,amsthm,amsbsy,amssymb,dsfont,stmaryrd}
\usepackage{aliascnt}
\usepackage[dvipsnames]{xcolor}
\usepackage{braket}

\usepackage{caption}
\usepackage{subcaption}

\usepackage{enumitem}
\usepackage{environ}

\makeatletter
\newcommand{\leqnomode}{\tagsleft@true\let\veqno\@@leqno}
\newcommand{\reqnomode}{\tagsleft@false\let\veqno\@@eqno}
\makeatother

\usepackage{mathtools}
\usepackage[makeroom]{cancel}
\numberwithin{equation}{section}

\newcommand\myshade{85}
\colorlet{mylinkcolor}{violet}
\colorlet{mycitecolor}{YellowOrange}
\colorlet{myurlcolor}{Aquamarine}

\usepackage[left=2.5cm,right=2.5cm,top=2.5cm,bottom=2.5cm]{geometry}
\usepackage[unicode=true,pdfusetitle,
bookmarks=true,bookmarksnumbered=false,bookmarksopen=false,
breaklinks=false,pdfborder={0 0 1},backref=false,
linkcolor  = mylinkcolor!\myshade!black,
citecolor  = mycitecolor!\myshade!black,
urlcolor   = myurlcolor!\myshade!black,
colorlinks = true,
]
{hyperref}

\usepackage{ifthen}

\usepackage{graphicx}
\usepackage{caption}
\usepackage{slashed}
\definecolor{ct_black}{HTML}{000000}
\definecolor{ct_orange}{HTML}{ED872D}
\definecolor{ct_purple}{HTML}{7A68A6}
\definecolor{ct_blue}{HTML}{348ABD}
\definecolor{ct_turquoise}{HTML}{188487}
\definecolor{ct_red}{HTML}{E32636}
\definecolor{ct_pink}{HTML}{CF4457}
\definecolor{ct_green}{HTML}{467821}

\definecolor{ct2_green}{HTML}{9FF781}
\definecolor{ct2_green_dark}{HTML}{088A08}

\theoremstyle{plain}
\newtheorem{thm}{\protect\theoremname}[section]

\newaliascnt{lem}{thm}
\theoremstyle{plain}
\newtheorem{lem}[lem]{\protect\lemmaname}
\aliascntresetthe{lem}

\newaliascnt{cor}{thm}
\theoremstyle{plain}
\newtheorem{cor}[cor]{\protect\corollaryname}
\aliascntresetthe{cor}

\newaliascnt{prop}{thm}
\theoremstyle{plain}

\aliascntresetthe{prop}

\newaliascnt{claim}{thm}
\theoremstyle{plain}
\newtheorem{claim}[claim]{\protect\claimname}
\aliascntresetthe{claim}

\newaliascnt{assumption}{thm}

\aliascntresetthe{assumption}

\theoremstyle{remark}
\newaliascnt{rem}{thm}
\newtheorem{rem}[rem]{\protect\remarkname}
\aliascntresetthe{rem}

\theoremstyle{definition}
\newaliascnt{defn}{thm}
\newtheorem{defn}[defn]{\protect\definitionname}
\aliascntresetthe{defn}

\newaliascnt{example}{thm}
\newtheorem{example}[example]{\protect\examplename}
\aliascntresetthe{example}
\providecommand{\assumptionname}{Assumption}
\providecommand{\claimname}{Claim}
\providecommand{\corollaryname}{Corollary}
\providecommand{\definitionname}{Definition}
\providecommand{\lemmaname}{Lemma}
\providecommand{\propositionname}{Proposition}
\providecommand{\remarkname}{Remark}
\providecommand{\theoremname}{Theorem}
\providecommand{\examplename}{Example}

\usepackage[nameinlink]{cleveref}

\crefname{section}{Section}{Sections}
\crefname{example}{Example}{Examples}
\crefname{appendix}{Appendix}{Appendices}
\crefname{figure}{Figure}{Figures}
\crefname{assumption}{Assumption}{Assumptions}
\crefname{thm}{Theorem}{Theorems}
\crefname{lem}{Lemma}{Lemmas}
\crefname{cor}{Corollary}{Corollaries}
\crefname{prop}{Proposition}{Propositions}
\crefname{claim}{Claim}{Claims}
\crefname{rem}{Remark}{Remarks}
\crefname{defn}{Definition}{Definitions}
\Crefname{assumption}{Assumption}{Assumptions}
\Crefname{thm}{Theorem}{Theorems}
\Crefname{lem}{Lemma}{Lemmas}
\Crefname{cor}{Corollary}{Corollaries}
\Crefname{prop}{Proposition}{Propositions}
\Crefname{claim}{Claim}{Claims}
\Crefname{rem}{Remark}{Remarks}
\Crefname{defn}{Definition}{Definitions}
\Crefname{example}{Example}{Examples}
\crefformat{equation}{(#2#1#3)}
\crefname{table}{Table}{Tables}

\crefrangelabelformat{equation}{(#3#1#4--#5#2#6)}

\crefmultiformat{equation}{(#2#1#3}{, #2#1#3)}{#2#1#3}{#2#1#3}
\Crefmultiformat{equation}{(#2#1#3}{, #2#1#3)}{#2#1#3}{#2#1#3}

\crefname{appendix}{Appendix}{Appendices}
\Crefname{appendix}{Appendix}{Appendices}

\newtheorem*{lem*}{\protect\lemmaname}

\newcommand{\ee}{\operatorname{e}}
\newcommand{\ii}{\operatorname{i}}

\newcommand{\ZZ}{\mathbb{Z}}

\newcommand{\bS}{\mathbb{S}}

\newcommand{\NN}{\mathbb{N}}
\newcommand{\RR}{\mathbb{R}}

\newcommand{\CC}{\mathbb{C}}

\newcommand{\calA}{\mathcal{A}}
\newcommand{\calB}{\mathcal{B}}
\newcommand{\calC}{\mathcal{C}}
\newcommand{\calF}{\mathcal{F}}
\newcommand{\calN}{\mathcal{N}}
\newcommand{\calG}{\mathcal{G}}

\newcommand{\calU}{\mathcal{U}}

\newcommand{\calH}{\mathcal{H}}

\newcommand{\calK}{\mathcal{K}}

\newcommand{\calJ}{\mathcal{J}}
\newcommand{\calP}{\mathcal{P}}

\newcommand{\Compacts}[1]{\mathcal{K}(#1)}

\newcommand{\Automorphisms}[1]{\operatorname{Aut}(#1)}

\newcommand\norm[1]{\left\lVert#1\right\rVert}
\newcommand{\ip}[2]{\langle #1, #2 \rangle}

\newcommand{\dif}{\operatorname{d}\!} % shorten the space afterwards
\newcommand{\tr}{\operatorname{tr}}

\newcommand{\szpan}{\operatorname{span}}
\newcommand{\Ad}[1]{\operatorname{Ad}_{#1}}

\newcommand{\ve}{\varepsilon}
\newcommand{\vf}{\varphi}
\newcommand{\Id}{\mathds{1}}
\newcommand{\dist}{\mathrm{dist}}

\newcommand{\sgn}{\operatorname{sgn}}
\newcommand{\findex}{\operatorname{ind}}

\newcommand{\rank}{\operatorname{rank}}

\newcommand{\coker}{\operatorname{coker}}
\newcommand{\supp}{\operatorname{supp}}

\DeclarePairedDelimiter\floor{\lfloor}{\rfloor}

\newcommand{\im}{\operatorname{im}}

\newcommand{\SO}{\bS^1}

\newcommand{\diam}{\operatorname{diam}}

\NewEnviron{malign}{%
	\begin{align}\begin{split}
			\BODY
	\end{split}\end{align}
}

\newcommand{\eq}[1]{\begin{align*}#1\end{align*}}

\newcommand{\eql}[1]{\begin{align}#1\end{align}}

\newcommand{\onehalf}{\frac{1}{2}}

\newcommand{\br}[1]{\left(#1\right)}

\newcommand{\Cstar}{C*}

\newcommand{\Mod}[1]{\ (\mathrm{mod}\ #1)}

\title{The $\ZZ_2$-Index of a Pair of Pure States and the\\Topology of Interacting 1D Superconductors}
\author{\href{mailto:am1864@princeton.edu}{Anna Mazhar}\\
	{\footnotesize Program in Applied and Computational Mathematics, Princeton University }\\
	\href{mailto:jacobshapiro@princeton.edu}{Jacob Shapiro}\\
	{\footnotesize Department of Mathematics, Princeton University}
}

\begin{document}
\reqnomode
\maketitle
\begin{abstract}
We define a relative $\ZZ_2$-index $\calN(\omega_1,\omega_2)$ for locally-comparable, parity-invariant pure states on a unital $\Cstar$-algebra. We prove that the index is well-defined, multiplicative, invariant under parity-preserving automorphisms, and locally constant in the norm topology. For the one-dimensional self-dual CAR algebra, we apply this construction to the half-chain automorphism $\sigma$ and define the many-body Majorana number
$\calN_\sigma(\omega):=\calN(\omega,\omega\circ\sigma)$
for parity-invariant $\sigma$-local pure states. In the quasi-free Hilbert--Schmidt regime, this index agrees with the usual single-body Majorana number. We then introduce symmetric local automorphism paths and prove that $\calN_\sigma$ completely classifies the resulting automorphic-path-components. This automorphic-path equivalence retains the bulk Majorana number but may forget a relative zero-dimensional parity obstruction.
\end{abstract}
\tableofcontents
\section{Introduction}\label{sec:introduction}

One-dimensional topological superconductors provide a basic setting in which a $\ZZ_2$-valued invariant distinguishes two phases. In the quasi-free Bogoliubov--de Gennes description, this invariant is the Majorana number of the Kitaev chain \cite{Kitaev2001Majorana}. For non-interacting systems, it may be expressed as a mod-$2$ Fredholm index associated with a half-space cut, and its nontrivial value predicts unpaired Majorana boundary modes.

The one-body index has several operator-algebraic formulations. Araki and Evans introduced a relative $\ZZ_2$-index for pairs of basis projections in the self-dual CAR algebra \cite{ArakiEvans1983Ising}. Related quasi-free and state-theoretic formulations were developed in \cite{AzaReyesLegaSequera2022QuasiFree,AzaMussnichReyesLega2022StateIndex}, while the norm-topological classification of one-dimensional particle-hole symmetric Fermi projections was obtained in \cite{ChungShapiro2023}. These constructions are intrinsically quasi-free: their input consists of single-body projections.

There is also a substantial operator-algebraic literature on interacting Fermionic chains. Bourne and Schulz-Baldes introduced $\ZZ_2$-indices for parity-invariant ground states of Fermionic chains \cite{BourneSchulzBaldes2020Z2}. Bourne and Ogata classified one-dimensional Fermionic symmetry-protected topological phases with finite on-site symmetry \cite{BourneOgata2021FermionSPT}, building on the split-state approach to one-dimensional phases; see also \cite{Matsui2013Split,Ogata2021PureStates}. Other recent state-based constructions include the many-body Fu--Kane--Mele index of \cite{BachmannBolsRahnama2024ManyBodyFKM}. Our aim is complementary: first introduce an abstract index capturing the topological obstruction associated with a pair of pure states (as in \cite{BachmannShapiroTauber2026}) and then study how it applies to a physical system. Our notion of locality ($\sigma$-locality below) is morally similar to the split property but is meaningfully weaker than it. It allows for a rather abstract classification, echoing the fact that without locality and symmetry, all pure states are automorphic-path-connected for separable simple C-star-algebras  \cite{KishimotoOzawaSakai2003Homogeneity}. It is moreover similar in spirit to \cite{ChungShapiro2023} in the sense that our definition of locality may seem unphysical initially, but the point is precisely that in order to understand the \emph{topological} structure of the problem we exhibit the minimal requirements necessary to make the indices well-defined. 

Let $\calA$ be a unital $\Cstar$-algebra equipped with an involutive automorphism $\theta$. Two pure states $\omega_1,\omega_2$ are \emph{locally-comparable} if
$\omega_2=\omega_1\circ\Ad{u}$
for some $u\in\calU(\calA)$. Following the relative-state perspective of \cite{BachmannShapiroTauber2026}, for two locally-comparable $\theta$-invariant pure states we define
\eq{
    \calN(\omega_1,\omega_2)
    :=
    \omega_1\br{u^\ast\theta(u)}
    \in\Set{\pm1}\,.
}
We prove in \cref{thm:main theorem for properties of the Z2 index of a pair} that this index is well-defined, multiplicative, invariant under automorphisms commuting with $\theta$, and locally constant in the norm topology. If $\omega_P,\omega_Q$ are the quasi-free states associated with particle-hole symmetric basis projections satisfying
$P-Q\in\calJ_2(\calH)$, then \cref{thm:Hilbert-space CAR correspondence} gives
\eq{
    \calN(\omega_P,\omega_Q)
    =
    \br{-1}^{\dim\im P\cap\ker Q}\,.
}
Thus the relative pure-state index extends the usual mod-$2$ index of a pair of projections. 

%We also record in \cref{sec:time reversal pair index} a time-reversal analogue of the relative index. Let $\tau$ be a conjugate-linear $*$-automorphism satisfying $\tau^2=\theta$. For two $\theta$-invariant pure states which are exchanged by $\tau$ and related by an even inner automorphism, we define a $\ZZ_2$-valued time-reversal pair index. In the quasi-free Hilbert--Schmidt regime, this index agrees with the Katsura--Koma index of a pair of projections \cite{KatsuraKoma2016}. We further explain how the half-line flux-insertion construction of \cite{BachmannShapiroTauber2026} produces the pair of $\pm\pi$-flux states and recovers the many-body Fu--Kane--Mele index of \cite{BachmannBolsRahnama2024ManyBodyFKM}.

We next specialize to the one-dimensional self-dual CAR algebra
\eq{
    \calA
    :=
    \operatorname{CAR}_{\rm sd}(\calH,\Xi),
    \qquad
    \calH
    :=
    \ell^2(\ZZ)\otimes\CC^{2N}
}
Here particle-hole symmetry is built into the self-dual CAR relations, whereas the involution $\theta$ describes Fermionic parity. Let $\sigma$ be the Bogoliubov automorphism induced by
$\Sigma=\sgn(X)$. We call a pure state $\omega$ \emph{$\sigma$-local} if $\omega$ and $\omega\circ\sigma$ are locally-comparable. For a $\theta$-invariant, $\sigma$-local pure state we define
\eq{
    \calN_\sigma(\omega)
    :=
    \calN(\omega,\omega\circ\sigma)\,.
}
Every parity-invariant pure split state is $\sigma$-local, although the converse does not hold. The invariant is unchanged by a finite displacement of the cut, is multiplicative under stacking, and is invariant under parity-preserving locally generated automorphisms. For the Kitaev $k$-chain appearing in \cref{example:Kitaev chain},
\eq{
    \calN_\sigma(\omega_k)
    =
    \br{-1}^k\,.
}

For a particle-hole symmetric basis projection $P$ satisfying
$[P,\Sigma]\in\calJ_2(\calH)$, the corresponding state $\omega_P$ is $\sigma$-local and
\eq{
    \calN_\sigma(\omega_P)
    =
    \br{-1}^{\dim\ker\br{P\Sigma P+P^\perp}}
}
see \cref{eq:single-body many-body correspondence}. The Hilbert--Schmidt condition is precisely what permits the single-body cut to be implemented as a local comparison of the associated pure states; compactness of $[P,\Sigma]$ alone is not sufficient.

The topology of the resulting state space is subtler. The index is locally constant in the norm topology, but it is not weak-$*$ continuous, even on quasi-free states. On the other hand, ordinary norm-path connectivity retains the full GNS sector and is therefore too restrictive for the desired classification. This motivates replacing paths of states by suitably constrained paths of automorphisms.

A symmetric local automorphism path is a point-norm continuous path
\eq{
    [0,1]\ni t
    \longmapsto
    \alpha_t\in\Automorphisms{\calA}
}
starting at the identity, commuting with $\theta$, and satisfying
\eq{
    \partial_\sigma(\alpha_t)
    :=
    \alpha_t\circ\sigma\circ\alpha_t^{-1}\circ\sigma^{-1}
    =
    \Ad{v_t}
}
for a norm-continuous path of $\theta$-even unitaries $v_t$. This condition preserves both $\sigma$-locality and $\calN_\sigma$. The converse is the main classification result. Following the pure-state homogeneity theorem of Kishimoto--Ozawa--Sakai \cite{KishimotoOzawaSakai2003Homogeneity} and Ogata's equivariant adaptation of its fixed-point-algebra argument \cite{Ogata2021PureStates}, we prove
\eq{
    \omega_1
    \text{ and }
    \omega_2
    \text{ are }
    \text{automorphic-path-connected in a local symmetric way}
    \quad\Longleftrightarrow\quad
    \calN_\sigma(\omega_1)
    =
    \calN_\sigma(\omega_2)\,;
}
see \cref{cor:index classifies automorphism path components}. The auxiliary finite-group action used in the proof is generated by parity and the locality involution, or its Majorana-twisted counterpart in the nontrivial sector.

This equivalence is coarser than the relative pair index. In particular, the $k=0$ Kitaev state and its image under a local Majorana flip away from the cut have the same $\sigma$-index and are automorphic-path-connected, although their \emph{relative} pair index is $-1$. Thus $\calN_\sigma$ classifies the bulk obstruction for symmetric local automorphism paths, while the relative index may additionally detect a zero-dimensional parity obstruction. In particular, \emph{just because two pure symmetric local states $\omega_1,\omega_2$ have the same $\calN_\sigma$ index does not imply that they have the trivial relative index}.

The relative index and its basic properties are developed in
\cref{sec:pair index}. The one-dimensional construction for class D superconductors and the
correspondence with the single-body Majorana number are treated in
\cref{sec:Majorana index}. Finally, we prove the
automorphic-path classification in
\cref{sec:automorphism path connectedness}. To complete the discussion on $\ZZ_2$ indices of pairs, in \cref{sec:time reversal pair index}, we discuss the time-reversal pair index, its correspondence with the Katsura--Koma index, and its relation to the many-body Fu--Kane--Mele invariant or class DIII superconducting chains. \cref{sec:a childs garden of locality} compares our notion of $\sigma$-locality with other notions of one-dimensional locality appearing in the literature; ours is the weakest. \cref{sec:LGAs SREs} derives some properties of our index when specialized to short-range entangled states and locally-generated automorphisms.

\paragraph{Acknowledgements}
JS was supported in part by NSF grant DMS-2510207 and the "ChatGPT for Academic Researchers program" of OpenAI. The authors wish to thank Christopher Bourne, Nikita Sopenko, Clement Tauber, Sven Bachmann, and Gian Michele Graf for stimulating discussions.

\section{The $\ZZ_2$-index of a pair of pure states}
\label{sec:pair index}
Let $\calA$ be a unital C*-algebra; its space of pure states is denoted $\calP(\calA)$. Let $\theta:\calA\to\calA$ be an involutive automorphism on it. 
We recall the notion of locally-comparable pure states from \cite{BachmannShapiroTauber2026}.
               
\begin{defn}[Locally-comparable pure states]
    Let $\omega_1,\omega_2$ be two pure states on $\calA$. We say that they are \emph{locally-comparable} iff there exists some unitary $u\in\calA$ such that $\omega_2 = \omega_1\circ\operatorname{Ad}_u$
\end{defn}

Clearly that means $\omega_1,\omega_2$ are in the same super selection sector. Be that as it may, they may still be separated by a topological obstruction.

\begin{defn} Let $\omega_1,\omega_2$ be a pair of locally-comparable pure states such that both are $\theta$-invariant ($\omega_i\circ\theta=\omega_i$ for $i=1,2$). Then the $\ZZ_2$-index of the pair is given by
\eql{
    \calN(\omega_1,\omega_2)\equiv \calN_\theta(\omega_1,\omega_2)  := \omega_1(u^\ast\theta u) \,.
}    Here and in the sequel, we may drop the dependence of the index on the involution $\theta$ since it is fixed once and for all. It is however understood to enter implicitly into all formulas.
\end{defn}
\begin{thm}\label{thm:main theorem for properties of the Z2 index of a pair} The index $\calN$ has the following properties:
\begin{enumerate}
    \item It is well-defined (i.e. it does not depend on the choice of the implementing unitary $u$).
    \item It is $\ZZ_2$-valued: $\calN(\omega_1,\omega_2) \in \mathbb{S}^1\cap\RR=\Set{\pm 1}$.
    \item It is multiplicative: Let $\omega_1, \omega_2, \text{ and }\omega_3$ be pure, $\theta$-invariant states which are all pairwise locally comparable. Then, 
    \[\calN(\omega_1, \omega_2)\calN(\omega_2, \omega_3) = \calN(\omega_1, \omega_3)\,.\]

    In particular, \eql{\calN(\omega,\omega)=+1} and \eql{\calN(\omega_1,\omega_2)\calN(\omega_2,\omega_1)=+1\,.}
    \item If $\omega_1,\omega_2$ are pure and $\theta$ invariant, then if $||\omega_1 - \omega_2||<2 $ then $\omega_1,\omega_2$ are locally comparable and \eql{\calN(\omega_1, \omega_2) = +1\,.}
    \item 
    If $\alpha$ is an automorphism such that $\alpha \circ \theta = \theta \circ \alpha$ then, assuming $\omega_1,\omega_2$ are pure $\theta$-invariant states which are locally-comparable, so are $\omega_1\circ\alpha,\omega_2\circ\alpha$ and \eql{\calN(\omega_1, \omega_2 ) = \calN(\omega_1 \circ \alpha, \omega_2 \circ \alpha)\,.}
    \item Multiplicativity with respect to stacking: if $\omega_1,\omega_2$ are $\theta$-invariant pure states on $\calA$ which are locally-comparable and $\widetilde{\omega_1},\widetilde{\omega_2}$ are $\widetilde{\theta}$-invariant pure states on $\widetilde{\calA}$ which are locally-comparable then $\omega_1\hat{\otimes}\widetilde{\omega_1},\omega_2\hat{\otimes}\widetilde{\omega_2}$ are locally-comparable and \eq{
    \calN(\omega_1\hat{\otimes}\widetilde{\omega_1},\omega_2\hat{\otimes}\widetilde{\omega_2}) = \calN(\omega_1,\omega_2)\calN(\widetilde{\omega_1},\widetilde{\omega_2})\,.
    } 
    % \item Relation to other index if there exists full $U(1)$ invariance. \JS{Perhaps we drop this for now since this is the abstract setting.}
\end{enumerate}
\end{thm}
In particular, we learn that if $\calN(\omega_1,\omega_2)=-1$ then although the states are inner-equivalent, every parity-preserving norm-continuous path trying to identify them must encounter distance 2 at the pair level.
\begin{proof}
We demonstrate the various points in the order which they were stated.
\begin{enumerate}
    \item We begin with well-definedness. Let $u,v$ be two unitaries such that $\omega_2 = \omega_1\circ\Ad{u}=\omega_1\circ\Ad{v}$.  With $\pi_1$ being the GNS representation associated to $\omega_1$, this implies that 
    \eq{\pi_1(u)\Omega_1 = \lambda \pi_1(v)\Omega_1 } for some $\lambda\in\CC$ such that $|\lambda| =1$ since $\omega_1$ is a pure state. Then
    \eq{
        \omega_1(v^* \theta(v)) &= \langle \Omega_1, \pi_1 \left(v^* \theta(v) \right)\Omega_1\rangle  \\ 
        &= |\lambda|^2\langle \Omega_1, \pi_1 \left(u^* \theta(u) \right)\Omega_1\rangle \\
        &= \omega_1(u^* \theta(u))\,.
    }
%%%%%%%%%%%%%%%%%%%%%%%%%%%%%%%%%%%%%%%%%%%%%%%%%%%%%%%%%%%%%%%%%
    \item We now tend to the $\ZZ_2$-valuedness. Note that\eq{
        \omega_1 \circ \text{Ad}_{\theta(u)} &= \omega_1 \circ \theta \circ \text{Ad}_u \circ \theta \\ 
        &= \omega_1 \circ\text{Ad}_u \circ \theta  \\ 
        &= \omega_2 \circ \theta \\ 
        &= \omega_1 \circ\text{Ad}_u\,.
    }
    This implies \[\omega_1 = \omega_1 \circ \text{Ad}_{\theta(u)} \circ\text{Ad}_{u^*}  = \omega_1 \circ \text{Ad}_{u^*\theta(u)}\,.\]
    Thus, by \cite[Lemma 2.7]{BachmannShapiroTauber2026}, we know that 
    $|\omega_1(u^* \theta(u))| = 1$. 
    In fact, this quantity is real by the following calculation: 
    \begin{align*}
        \overline{\omega_1( u^* \theta(u))} &= \omega((u^* \theta(u))^*) \\ 
        &= \omega_1(\theta(u^*)u) \\ 
        &= \omega_1(\theta (\theta(u^*)u)) \\
        &= \omega_1(u^*\theta(u))\,.
    \end{align*}
    Therefore, $\calN(\omega_1, \omega_2) = \pm 1$.
%%%%%%%%%%%%%%%%%%%%%%%%%%%%%%%%%%%%%%%%%%%%%%%%%%%%%%%%%%%%%%%%%
    \item For multiplicativeness: let $\omega_2 = \omega_1 \circ \text{Ad}_u$, $\omega_3 = \omega_2 \circ \text{Ad}_v$, so $\omega_3 = \omega_1 \circ \text{Ad}_u \circ\text{Ad}_v = \omega_1 \circ \text{Ad}_{vu} $ and so
    \begin{align*}
        \calN(\omega_1, \omega_3) &= \omega_1((vu)^* \theta (vu)) \\ 
        &= \langle \Omega_1, \pi(u^*) \pi(v^*) \Theta \pi(v) \Theta \Theta \pi(u) \Theta \Omega_1\rangle \\ 
        &= \langle \Omega_2,  \pi(v^*) \Theta \pi(v) \Theta \Theta \Omega_2 \rangle \\ 
        &= \calN(\omega_1, \omega_2)\langle \Omega_2,  \pi(v^*) \Theta \pi(v) \Theta \Omega_2 \rangle \\  
        &= \calN(\omega_1, \omega_2) \omega_2 (v^* \theta(v)) \\
         &= \calN(\omega_1, \omega_2)\calN(\omega_2, \omega_3)\,.
    \end{align*}
%%%%%%%%%%%%%%%%%%%%%%%%%%%%%%%%%%%%%%%%%%%%%%%%%%%%%%%%%%%%%%%%%
    \item We now tend to continuity. By
    \cite[Corollaries~8 and~9]{GlimmKadison1960}, if
    \(\norm{\omega_1-\omega_2}<2\), then the GNS representations of
    \(\omega_1\) and \(\omega_2\) are unitarily equivalent, and hence
    there exists a unitary \(u\in\calU(\calA)\) such that
    \eq{
        \omega_2=\omega_1\circ\Ad{u}\,.
    }
    Thus the two states are locally comparable.

    Let \((\pi,\calH,\Omega_1)\) be the GNS representation of
    \(\omega_1\), and set
    \eq{
        \Omega_2:=\pi(u)\Omega_1\,.
    }
    Then \(\Omega_2\) implements \(\omega_2\) in the same representation.
    Since \(\omega_1\) is \(\theta\)-invariant, there is a unitary
    involution \(\Theta\) on \(\calH\) defined by
    \eq{
        \Theta\pi(a)\Omega_1=\pi(\theta(a))\Omega_1,
        \qquad a\in\calA\,.
    }
    In particular, \(\Theta\Omega_1=\Omega_1\). Moreover,
    \eq{
        \Theta\Omega_2
        =\Theta\pi(u)\Omega_1
        =\pi(\theta(u))\Omega_1\, ,
    }
    and therefore
    \eq{
        \calN(\omega_1,\omega_2)
        =
        \omega_1(u^*\theta(u))
        =
        \ip{\pi(u)\Omega_1}{\pi(\theta(u))\Omega_1}
        =
        \ip{\Omega_2}{\Theta\Omega_2}\,.
    }
    Suppose, for contradiction, that
    \(\calN(\omega_1,\omega_2)=-1\). Since \(\Theta\) is unitary and
    \(\Omega_2\) is a unit vector, the equality
    \(\ip{\Omega_2}{\Theta\Omega_2}=-1\) implies
    \eq{
        \Theta\Omega_2=-\Omega_2\,.
    }
    Thus \(\Omega_1\) and \(\Omega_2\) lie in the \(+1\) and \(-1\)
    eigenspaces of the self-adjoint unitary \(\Theta\), respectively, and
    hence
    \eq{
        \ip{\Omega_1}{\Omega_2}=0\,.
    }
    By the Powers--St{\o}rmer vector-state distance formula
    \cite[Lemma~2.4]{PowersStormer1970FreeStates},
     \eql{
        \norm{\omega_1-\omega_2}
        =
        2\sqrt{1-\left|\ip{\Omega_1}{\Omega_2}\right|^2}
        =
        2\, ,
    }
    contradicting the hypothesis \(\norm{\omega_1-\omega_2}<2\). Hence
    \(\calN(\omega_1,\omega_2)\neq -1\). Since the previous item shows
    that \(\calN(\omega_1,\omega_2)\in\Set{\pm1}\), we conclude that
    \eq{
        \calN(\omega_1,\omega_2)=+1\,.
    }
%%%%%%%%%%%%%%%%%%%%%%%%%%%%%%%%%%%%%%%%%%%%%%%%%%%%%%%%%%%%%%%%%
\item $\omega_1, \, \omega_2$ are locally comparable via unitary $u \in \calU(\calA)$, so 
\eq{\omega_2 \circ \alpha = (\omega_1 \circ \alpha) \circ \text{Ad}_{\alpha^{-1}(u)} } 
and 
\eq{
\calN(\omega_1 \circ \alpha, \omega_2 \circ \alpha) = \omega_1 \circ \alpha (\alpha^{-1}(u)^* \theta(\alpha^{-1}(u)))  = \omega_1 (u^*\theta(u)) = \calN(\omega_1, \omega_2)\,.
}

\item Let \(u\in\calU(\calA)\) and
\(\widetilde u\in\calU(\widetilde{\calA})\) implement the respective
local-comparability relations. Set
\eq{
    \varepsilon
    &:=
    \calN(\omega_1,\omega_2),&
    \widetilde\varepsilon
    &:=
    \calN(\widetilde{\omega}_1,\widetilde{\omega}_2)\,.
}
Let \((\pi,\calH,\Omega_1)\) and
\((\widetilde\pi,\widetilde{\calH},\widetilde\Omega_1)\) be the GNS
representations of \(\omega_1\) and \(\widetilde{\omega}_1\),
respectively, and let \(\Theta\) and \(\widetilde\Theta\) be their
canonical parity implementers. Thus
\eq{
    \Theta\Omega_1=\Omega_1,
    \qquad
    \widetilde\Theta\widetilde\Omega_1=\widetilde\Omega_1\,.
}
Define
\eq{
    \Omega_2:=\pi(u)\Omega_1,
    \qquad
    \widetilde\Omega_2
    :=
    \widetilde\pi(\widetilde u)\widetilde\Omega_1\,.
}
As in the calculation above,
\eq{
    \ip{\Omega_2}{\Theta\Omega_2}
    =\varepsilon,
    \qquad
    \ip{\widetilde\Omega_2}
       {\widetilde\Theta\widetilde\Omega_2}
    =\widetilde\varepsilon\,.
}
Since \(\varepsilon,\widetilde\varepsilon\in\Set{\pm1}\), equality in
Cauchy--Schwarz gives
\eq{
    \Theta\Omega_2
    =\varepsilon\Omega_2,
    \qquad
    \widetilde\Theta\widetilde\Omega_2
    =\widetilde\varepsilon\widetilde\Omega_2\,.
}

Set
\eq{
    \calC:=\calA\hat{\otimes}\widetilde{\calA},
    \qquad
    \vartheta:=\theta\hat{\otimes}\widetilde\theta\,.
}
The graded tensor-product representation \(\Pi\) on
\(\calH\otimes\widetilde{\calH}\) is given, for homogeneous
\(a\in\calA\) and \(\widetilde a\in\widetilde{\calA}\), by
\eq{
    \Pi(a\hat{\otimes}\widetilde a)
    :=
    \pi(a)\Theta^{|\widetilde a|}
    \otimes\widetilde\pi(\widetilde a)\,.
}
This representation is irreducible. Indeed, its bicommutant contains
\(\calB(\calH)\otimes\Id_{\widetilde{\calH}}\), and hence also
\(\Theta\otimes\Id_{\widetilde{\calH}}\); it therefore contains
\(\Id_{\calH}\otimes\calB(\widetilde{\calH})\) as well. Consequently,
\(\Pi(\calC)''=\calB(\calH\otimes\widetilde{\calH})\)

The vectors
\eq{
    \Psi_1
    &:=
    \Omega_1\otimes\widetilde\Omega_1,&
    \Psi_2
    &:=
    \Omega_2\otimes\widetilde\Omega_2
}
represent, respectively,
\(\omega_1\hat{\otimes}\widetilde{\omega}_1\) and
\(\omega_2\hat{\otimes}\widetilde{\omega}_2\). For the second assertion,
it suffices to observe that
\eq{
    \ip{\Psi_2}
       {\Pi(a\hat{\otimes}\widetilde a)\Psi_2}
    &=
    \varepsilon^{|\widetilde a|}
    \omega_2(a)\widetilde{\omega}_2(\widetilde a)\\
    &=
    \omega_2(a)\widetilde{\omega}_2(\widetilde a)\,.
}
Indeed, when \(\widetilde a\) is odd,
\(\widetilde{\omega}_2(\widetilde a)=0\) by
\(\widetilde\theta\)-invariance. In particular, both graded product
states are pure. Since they are vector states in the same irreducible
representation, Kadison transitivity provides a unitary
\(w\in\calU(\calC)\) such that
\eq{
    \Pi(w)\Psi_1=\Psi_2\,.
}
It follows that
\eq{
    \omega_2\hat{\otimes}\widetilde{\omega}_2
    =
    \left(
    \omega_1\hat{\otimes}\widetilde{\omega}_1
    \right)\circ\Ad{w}\, ,
}
so the two product states are locally comparable.

Finally, \(\vartheta\) is implemented in this representation by
\(\Theta\otimes\widetilde\Theta\), which fixes \(\Psi_1\). Therefore,
\eq{
    \calN\left(
        \omega_1\hat{\otimes}\widetilde{\omega}_1,
        \omega_2\hat{\otimes}\widetilde{\omega}_2
    \right)
    &=
    \left(
        \omega_1\hat{\otimes}\widetilde{\omega}_1
    \right)\left(w^*\vartheta(w)\right)\\
    &=
    \ip{\Psi_2}
       {(\Theta\otimes\widetilde\Theta)\Psi_2}\\
    &=
    \varepsilon\widetilde\varepsilon\\
    &=
    \calN(\omega_1,\omega_2)
    \calN(\widetilde{\omega}_1,\widetilde{\omega}_2)\,.
}

\end{enumerate}
\end{proof}

\begin{example}[$\theta$ but not $U(1)$ symmetric state]
    Let $\calA = \operatorname{CAR(\ell^2(\NN)})$ and define the $U(1)$ gauge automorphism as $\rho_t a(\vf) \equiv \ee^{-\ii t}a(\vf)$ for any generator, with $\vf\in\ell^2(\NN)$. Then the involution can be taken as the parity automorphism, $\theta:=\rho_{\pi}$. Then, a pure state which is $\theta$ invariant but not fully $\rho_t$ invariant is described as follows. Let $\Set{e_j}_j$ be the standard basis for $\ell^2(\NN)$ and $\omega_0$ the quasi-free state associated with the zero-projection on $\ell^2(\NN)$. Define then an automorphism $\beta$ by $\beta(a(e_1)):=\frac{1}{\sqrt{2}}\br{a(e_1)+a(e_2)^\ast}$ and $\beta(a(e_2)):=\frac{1}{\sqrt{2}}\br{a(e_2)-a(e_1)^\ast}$. Finally $\beta(a(\vf)):=a(\vf)$ for all $\vf\perp\szpan(\Set{e_1,e_2})$. Then $\omega:=\omega_0\circ\beta$ is the desired state. Indeed, $\beta$ commutes with $\theta$ because it maps generators to (odd linear combination of) generators. However, one may calculate
    \eq{
    \omega(\rho_t(a(e_1)a(e_2))) = -\frac12\ee^{-2\ii t}\, ,
    } and taking $t=\pi/2$ yields non invariance.
\end{example}
\begin{example}[Non-trivial index]
Let $\calA = \operatorname{CAR}(\ell^2(\NN))$ and $\theta$ be the parity automorphism, i.e., $\theta a(\vf)\equiv -a(\vf)$ on generators. Consider then the pure state $\omega_\Id$ which corresponds to the quasi-free state associated with the identity projection, and also $\omega_o:=\omega_\Id\circ\Ad{u}$ where $u := a(\vf)+a(\vf)^\ast$ for any normalized $\vf\in\ell^2(\NN)$. Then $\omega_\Id,\omega_o$ are locally comparable, both $\theta$-invariant, and 
\eq{
\calN(\omega_\Id,\omega_o) = -1\,.
}
\end{example}

\subsection{Relation to the $\ZZ_2$-index of a pair of projections}

To relate the above index of a pair notion on an abstract C-star algebra with a Hilbert space notion, we specify to the case that the unital C-star algebra $\calA$ arises as the CAR algebra associated to a given Hilbert space $\calH$. One then needs slightly more structure on $\calH$: an anti-unitary $\Xi$ which square to $+\Id$.

Let then $P,Q$ be two projections on a given Hilbert space $\calH$, and $\Xi$ be an anti-unitary such that $\Xi^2=+\Id$. Assume further that $P-Q\in\calK$ and that $P,Q$ are both $\Xi$-symmetric, in the sense that $\Xi P \Xi = P^\perp$ and $\Xi Q \Xi = Q^\perp$. Then, in this context, it is well-known \cite{KatsuraKoma2016} (or the original \cite{AtiyahSinger1969} as well as \cite{SchulzBaldes2015Z2Indices,FonsecaEtAl2020,ChungShapiro2023}) that there is a $\ZZ_2$-index associated with the pair of projections, given by 
\eql{
    \findex(P,Q) \equiv \dim \im P \cap\ker Q\mod 2 = \dim \chi_{\Set{+1}}\br{P-Q}\mod 2\,.
} One notes that the usual ($\ZZ$-valued) index of a pair of projections, as in e.g. \cite{Avron1994220}, is always zero in this context. Moreover, \cite{SchulzBaldes2015Z2Indices,KatsuraKoma2016,FonsecaEtAl2020} deal with a somewhat different anti-unitary structure: an anti-unitary that squares to $-\Id$; see \cref{sec:time reversal pair index} below.

We can also lift these projections to make them quasi-free $\theta$-invariant pure states $\omega_P,\omega_Q$ on the self-dual CAR algebra $\calA := \operatorname{CAR}_{\rm sd}(\calH,\Xi)$ \cite{Araki1970QuasifreeCAR} via the usual Gaussian law. Then 
\begin{thm}\label{thm:Hilbert-space CAR correspondence}
    If additionally $P-Q\in \calJ_2$ then $\omega_P$ and $\omega_Q$ are locally-comparable and 
    \eql{
    \calN(\omega_P,\omega_Q) = \br{-1}^{\dim \im P \cap\ker Q}\,.
    }
\end{thm}
Note that as in \cite{BachmannShapiroTauber2026}, currently we cannot cover the full generality of the Hilbert space theorem: we can only manage with the assumption that $P-Q\in\calJ_2$ rather than $P-Q\in\calK$.
\begin{proof}
Since
\eq{
    P-Q\in\calJ_2(\calH),
}
the Shale--Stinespring criterion implies that the GNS representations
of the pure quasi-free states $\omega_P$ and $\omega_Q$ are unitarily
equivalent. By the Kadison transitivity theorem, there is therefore a
unitary $u\in\calU(\calA)$ such that
\eql{
    \omega_Q=\omega_P\circ\Ad{u}\,.
}
Thus $\omega_P$ and $\omega_Q$ are locally-comparable.

We now identify their many-body index. Let
$(\calH_P,\pi_P,\Omega_P)$ and $(\calH_Q,\pi_Q,\Omega_Q)$ be the GNS
representations of $\omega_P$ and $\omega_Q$, respectively, and let
$\Theta_P$ and $\Theta_Q$ be the self-adjoint unitaries implementing
$\theta$ in these representations. Thus
\eq{
    \Theta_P\pi_P(a)\Theta_P&=\pi_P(\theta a),&
    \Theta_P\Omega_P&=\Omega_P,\\
    \Theta_Q\pi_Q(a)\Theta_Q&=\pi_Q(\theta a),&
    \Theta_Q\Omega_Q&=\Omega_Q
}
for all $a\in\calA$.

Set
\eql{
    \widetilde{\Omega_Q}:=\pi_P(u)\Omega_P\,.
}
Then $\widetilde{\Omega_Q}$ implements $\omega_Q$ in the GNS
representation of $\omega_P$:
\eql{
    \omega_Q(a)
    =
    \ip{\widetilde{\Omega_Q}}
        {\pi_P(a)\widetilde{\Omega_Q}}
    \qquad(a\in\calA)\,.
}
Since $\omega_P$ and $\omega_Q$ are both $\theta$-invariant, $\theta u$
is also an implementing unitary for $\omega_Q$. Consequently,
$\pi_P(u)\Omega_P$ and $\pi_P(\theta u)\Omega_P$ implement the same
pure state in the irreducible representation $\pi_P$, and hence they
differ by a phase:
\eql{
    \pi_P(\theta u)\Omega_P
    =
    \lambda\pi_P(u)\Omega_P
}
for some $\lambda\in\mathbb{S}^1$. Multiplying by $\pi_P(u^\ast)$ and
taking the inner product with $\Omega_P$ gives
\eq{
    \lambda
    &=
    \ip{\Omega_P}{\pi_P(u^\ast\theta u)\Omega_P}\\
    &=
    \omega_P(u^\ast\theta u)\\
    &=
    \calN(\omega_P,\omega_Q)\,.
}
In particular, $\lambda\in\Set{\pm1}$. Moreover,
\eq{
    \Theta_P\widetilde{\Omega_Q}
    &=
    \Theta_P\pi_P(u)\Omega_P\\
    &=
    \pi_P(\theta u)\Omega_P\\
    &=
    \lambda\pi_P(u)\Omega_P\\
    &=
    \lambda\widetilde{\Omega_Q}\,.
}

Let $W:\calH_Q\to\calH_P$ be the unitary intertwiner determined by
\eql{
    W\pi_Q(a)\Omega_Q
    :=
    \pi_P(a)\widetilde{\Omega_Q},
    \qquad a\in\calA\,.
}
For every $a\in\calA$ we then have
\eq{
    \Theta_PW\pi_Q(a)\Omega_Q
    &=
    \Theta_P\pi_P(a)\widetilde{\Omega_Q}\\
    &=
    \pi_P(\theta a)\Theta_P\widetilde{\Omega_Q}\\
    &=
    \lambda\pi_P(\theta a)\widetilde{\Omega_Q}\\
    &=
    \lambda W\pi_Q(\theta a)\Omega_Q\\
    &=
    \lambda W\Theta_Q\pi_Q(a)\Omega_Q\,.
}
Since $\pi_Q(\calA)\Omega_Q$ is dense in $\calH_Q$, it follows that
\eql{
    \Theta_PW
    =
    \lambda W\Theta_Q
    =
    \calN(\omega_P,\omega_Q)W\Theta_Q\,.
}
Thus $W$ preserves the parity eigenspaces when
$\calN(\omega_P,\omega_Q)=+1$ and exchanges them when
$\calN(\omega_P,\omega_Q)=-1$.

By the parity-sector characterization of the relative index of two
basis projections
\cite[Theorem~1 (b) and (c)]{AzaMussnichReyesLega2022StateIndex}, the
sign specifying whether these parity eigenspaces are preserved or
exchanged is
\eql{
    \br{-1}^{\dim\im P\cap\ker Q}\,.
}
We conclude that
\eql{
    \calN(\omega_P,\omega_Q)
    =
    \br{-1}^{\dim\im P\cap\ker Q}\,.
}
Equivalently, under the identification
$\Set{\pm1}\cong\ZZ_2$, this is the parity
\eql{
    \dim\im P\cap\ker Q\mod 2\,.
}
\end{proof}

\section{1D class D superconductors and their Majorana index}
\label{sec:Majorana index}

We now specify to our physics application; to do so we start by setting up the context for superconductors in the interacting setting.

Let $\mathcal{H}$ be a separable Hilbert space and $\Xi:\mathcal{H}\to\mathcal{H}$
be anti-unitary such that $\Xi^{2}=\Id$. With this data we define the canonical
self-dual CAR algebra \cite{Araki1970QuasifreeCAR} 
\[
\mathcal{A}:=\text{CAR}_{\text{SD}}\left(\mathcal{H},\Xi\right)
\]
which is the universal C* algebra generated by the linear map
\[
B:\mathcal{H}\to\mathcal{A}
\]
such that 
\[
B\left(\Xi\psi\right)=B\left(\psi\right)^{\ast}
\]
and 
\[
\Set{B\left(\psi\right)^{\ast},B\left(\varphi\right)}=\left\langle \psi,\varphi\right\rangle 1_\calA\,.
\] We interpret $\calH$ as the one-particle Bogoliubov-de Gennes Hilbert space, $\Xi$ as particle-hole conjugation and $\calA$ as the algebra of Fermionic observables.

On $\mathcal{A}$ we have the parity automorphism $\theta:\mathcal{A}\to\mathcal{A}$
given by 
\[
\theta B\left(\psi\right)\equiv-B\left(\psi\right)\qquad\left(\psi\in\mathcal{H}\right)
\]
and extended linearly and continuously.

\begin{rem}The self-dual CAR algebra carries no canonical non-trivial
\emph{scalar} \(U(1)\) gauge action. Indeed, if
\eq{
    \rho_t\left(B(\psi)\right)
    :=
    \ee^{-\ii t}B(\psi),
}
then self-duality and preservation of the adjoint would respectively give
\eq{
    \rho_t\left(B(\psi)^*\right)
    &=
    \ee^{-\ii t}B(\Xi\psi),&
    \rho_t\left(B(\psi)\right)^*
    &=
    \ee^{\ii t}B(\Xi\psi).
}
Thus the transformation is compatible with the self-dual relation only
when \(\ee^{-\ii t}=\ee^{\ii t}\). The only non-trivial surviving
transformation is \(t=\pi\), which is precisely the parity automorphism
\(\theta\).

Non-scalar Bogoliubov \(U(1)\)-actions may nevertheless exist. A strongly
continuous unitary representation
\eq{
    \RR/2\pi\ZZ\ni t
    \longmapsto
    V_t\in\calU(\calH),
    \qquad
    V_t\Xi=\Xi V_t,
}
defines
\eq{
    \rho_t^V\left(B(\psi)\right)
    :=
    B(V_t\psi).
}
If \(V_t=\ee^{-\ii tQ}\), the compatibility condition is
\(\Xi Q\Xi=-Q\), so a non-zero compatible charge cannot be scalar. For
example, a basis projection \(P\), satisfying
\(\Xi P\Xi=P^\perp\), gives
\eq{
    V_t^P
    &:=
    \ee^{-\ii t}P+\ee^{\ii t}P^\perp,&
    \rho_t^P\left(B(\psi)\right)
    &:=
    B(V_t^P\psi),
}
with
\eq{
    V_t^P\Xi
    &=
    \Xi V_t^P,&
    V_\pi^P
    &=
    -\Id,&
    \rho_\pi^P
    &=
    \theta.
}
Such an action depends on the additional choice of \(P\), or more generally
of \(Q\), and is not intrinsic to the self-dual CAR algebra. Thus the
absence of \(U(1)\) symmetry means that no continuous charge symmetry is
specified, not that no circle action can be chosen.
\end{rem}

\begin{rem}\(\Xi\) is not an additional physical symmetry. The relation
\eq{
    B(\Xi\psi)
    =
    B(\psi)^*
}
identifies \(\Xi\) with the adjoint operation on the field generators:
injectivity of \(B\), involutivity of the adjoint, and the CAR relations
make \(\Xi\) uniquely determined, involutive, and anti-unitary. It therefore
records the particle--hole redundancy of the Nambu description rather than
an automorphism under which states must be invariant. The one-particle
condition \(\Xi P\Xi=P^\perp\) is the corresponding statement for a
covariance projection; in the many-body algebra this constraint is already
built into the self-dual CAR relations. See
\cref{sec:role of particle-hole symmetry} below.
\end{rem}

\subsection{Local states in one-dimension}

We now specify to one dimension by choosing $\mathcal{H}=\ell^{2}\left(\mathbb{Z}\right)\otimes\mathbb{C}^{2N}$
for some $N\in\mathbb{N}$. Hence we explicitly assume that the internal on-site space has even dimension; this would be automatically guaranteed if we had a basis projection $P$ obeying both $\Xi P \Xi=P^\perp$ and $[P,X]=0$ (assuming $[\Xi,X]=0$ too).

On $\mathcal{A}$ we have the so-called locality-automorphism. It is the Bogoliubov automorphism associated with the locality automorphism appearing in the classification scheme of Chung and Shapiro \cite{ChungShapiro2023}.
If $\Sigma:=\sgn\left(X\right)$ (with $X$ the position operator)
is a self-adjoint unitary on $\mathcal{H}$ given by 
\[
\sgn\left(X\right)\delta_{x}:=\begin{cases}
\delta_{x} & x>0\\
-\delta_{x} & x\leq0
\end{cases}\, ,
\]
then, assuming $\left[\Xi,X\right]=0$, the (outer automorphism) $\sigma:\mathcal{A}\to\mathcal{A}$
is given by the Bogoliubov automorphism 
\[
B\left(\psi\right)\mapsto B\left(\Sigma\psi\right)\qquad\left(\psi\in\mathcal{H}\right)\,.
\]

On $\mathcal{A}$ we consider pure states
which are $\sigma$-local, i.e.,
\begin{defn}[$\sigma$-local states]
 We say that $\omega\in\mathcal{P}\left(\mathcal{A}\right)$ is local
iff $\omega$ and $\omega\circ\sigma$ are locally-comparable, i.e.,
there exists some $u\in\mathcal{U}\left(\mathcal{A}\right)$ with
which 
\[
\omega\circ\sigma=\omega\circ\operatorname{Ad}_{u}\equiv\omega\left(u^{\ast}\cdot u\right)\,.
\]
\end{defn}

Let $\Lambda := \chi_\NN(X)$ be the projection onto the RHS of space on $\calH$. Then clearly $\Lambda=\frac12\br{\Sigma+\Id}$, thus $[A,\Lambda]\in\calK$ iff $[A,\Sigma]\in\calK$. Thus, recall that Chung-Shapiro defined a Fermi projection $P=P^2=P^\ast$  to be local in 1D iff $[P,\Lambda]\in\calK$ \cite{ChungShapiro2023}.

The first thing to note about the $\sigma$-local many-body definition is that it does \emph{not} capture the full generality of the single-body definition.
\begin{claim}
    There exist projections $P=P^2=P^\ast$ on $\calH=\ell^2(\ZZ)\otimes\CC^{2N}$ which have $[\Sigma,P]\in\calK(\calH)$ (and is hence $\Lambda$-local in the sense of Chung-Shapiro \cite{ChungShapiro2023}) but such that the associated quasi-free state $\omega_P$ is \emph{not} $\sigma$-local. It can moreover be chosen to be particle-hole symmetric in the sense that $\Xi P \Xi = P^\perp$.
\end{claim}
\begin{proof}
It suffices to take \(N=1\). Let \(\calH=\ell^2(\ZZ)\otimes\CC^2\), let \(\Xi\) be the anti-linear involution interchanging \(\delta_x\otimes e_+\) and \(\delta_x\otimes e_-\), and take \(\Sigma=-1\) on \(x\leq0\) and \(\Sigma=1\) on \(x\geq1\). If \(P_0\) denotes the projection onto \(E_0=\overline{\szpan}\{\delta_x\otimes e_+:x\in\ZZ\}\), then \(\Xi P_0\Xi=\Id-P_0\) and \([P_0,\Sigma]=0\). For \(n\geq1\), set

$$
p_n=\delta_n\otimes e_+,\qquad h_n=\delta_{1-n}\otimes e_-,\qquad
\theta_n=n^{-1/2},\quad c_n=\cos\theta_n,\quad s_n=\sin\theta_n,
$$

and define a unitary \(U\) by

$$
Up_n=c_np_n+s_nh_n,\qquad Uh_n=-s_np_n+c_nh_n,
$$

using the same real rotation on \(\szpan\{\Xi p_n,\Xi h_n\}\) and the identity on the remaining orthogonal complement. Then \(U\Xi=\Xi U\), so

$$
P:=UP_0U^*
\qquad\text{satisfies}\qquad
\Xi P\Xi=\Id-P
$$

and therefore determines a pure quasi-free state \(\omega_P\). On the planes \(\szpan\{p_n,h_n\}\) and \(\szpan\{\Xi p_n,\Xi h_n\}\), respectively,

$$
P-\Sigma P\Sigma
=\pm
\begin{pmatrix}
0&2c_ns_n\\
2c_ns_n&0
\end{pmatrix}\,.
$$

Thus its nonzero singular values are \(2|c_ns_n|\), with fixed multiplicity. Since \(c_ns_n\sim n^{-1/2}\), these singular values tend to zero, but their squares are not summable; hence

$$
[P,\Sigma]\in\calK(\calH),
\qquad
[P,\Sigma]\notin\mathcal J_2(\calH),
$$

where we used \(P-\Sigma P\Sigma=[P,\Sigma]\Sigma\). Now put \(Q=\Sigma P\Sigma\). The Bogoliubov automorphism \(\sigma\) sends \(\omega_P\) to \(\omega_Q\), since

$$
(\omega_P\circ\sigma)\bigl(B(\psi)^*B(\varphi)\bigr)
=\langle\Sigma\psi,P\Sigma\varphi\rangle
=\langle\psi,Q\varphi\rangle\,.
$$

But \(P-Q\notin\mathcal J_2(\calH)\), so the Shale--Stinespring--Araki criterion implies that \(\omega_P\) and \(\omega_Q\) are not quasi-equivalent. Consequently \(\omega_P\circ\sigma\) cannot equal \(\omega_P\circ\operatorname{Ad}_u\) for any unitary \(u\in\calA\), since inner conjugation preserves the GNS sector. Thus \(\omega_P\) is not \(\sigma\)-local although \([P,\Sigma]\) is compact.
\end{proof}

Next, there is nothing special about the origin. Indeed, we have the 
\begin{lem}
For any $x\in\ZZ$, $\omega$ is $\sigma$-local iff $\omega$ is $\sigma_x$ local, where $\sigma_x$ is the Bogoliubov automorphism associated with the unitary operator $\Sigma_x:=\sgn(X-x)$.
\end{lem}
\begin{proof}
Set \(W_x:=\Sigma_x\Sigma\). Since \(\Sigma_x\) and \(\Sigma\) are sign operators with cuts at \(x\) and \(0\), respectively, \(W_x=-\Id\) on the finitely many sites between the cuts and \(W_x=\Id\) elsewhere; hence \(W_x-\Id\in\calJ_1(\calH)\), and

$$
\calK_x:=\ker(W_x+\Id)
$$

is a finite-dimensional \(\Xi\)-invariant subspace of dimension \(2N|x|\). Choose a basis projection on \(\calK_x\), let \(c_1,\ldots,c_{N|x|}\) be the associated annihilation operators, and define

$$
v_x:=\prod_{j=1}^{N|x|}(\Id-2c_j^*c_j)\,.
$$

Then \(v_x=v_x^*\in\calU(\calA)\) and

$$
v_xB(\psi)v_x=B(W_x\psi)\, ,
$$

so \(\alpha_{W_x}=\Ad{v_x}\). Since \(\Sigma_x\Sigma=W_x\) and \(\Sigma^2=\Id\),

$$
\sigma_x\circ\sigma=\alpha_{W_x}=\Ad{v_x},
\qquad\text{hence}\qquad
\sigma_x=\Ad{v_x}\circ\sigma
       =\sigma\circ\Ad{\sigma(v_x)}\,.
$$

Consequently, for every state \(\omega\)

$$
\omega\circ\sigma_x
=(\omega\circ\sigma)\circ\Ad{\sigma(v_x)}\, ,
$$

so \(\omega\circ\sigma_x\) and \(\omega\circ\sigma\) are locally comparable. By symmetry and transitivity of local comparability,

$$
\omega\sim\omega\circ\sigma_x
\quad\Longleftrightarrow\quad
\omega\sim\omega\circ\sigma\, ,
$$

which proves that \(\omega\) is \(\sigma_x\)-local if and only if it is \(\sigma\)-local.
\end{proof}

\subsection{The role of particle-hole symmetry}
\label{sec:role of particle-hole symmetry}
In the single particle picture, a Fermi projection $P=P^{2}=P^{\ast}\in\mathcal{B}\left(\mathcal{H}\right)$
is called ``particle-hole-symmetric'' iff $\Xi P\Xi=P^{\perp}$.
This implies the following constraint on the quasi-free state $\omega_{P}$:
\begin{eqnarray*}
\omega_{\Xi P\Xi}\left(B\left(\psi\right)^{\ast}B\left(\varphi\right)\right) & = & \left\langle \psi,\Xi P\Xi\varphi\right\rangle \\
 & = & \left\langle \psi,P^{\perp}\varphi\right\rangle \\
 & = & \left\langle \psi,\varphi\right\rangle -\left\langle \psi,P\varphi\right\rangle \\
 & = & \left\langle \psi,\varphi\right\rangle -\omega_{P}\left(B\left(\psi\right)^{\ast}B\left(\varphi\right)\right)\, ,
\end{eqnarray*}
but we also have
\begin{eqnarray*}
\omega_{\Xi P\Xi}\left(B\left(\psi\right)^{\ast}B\left(\varphi\right)\right) & = & \left\langle \psi,\Xi P\Xi\varphi\right\rangle \\
 & = & \overline{\left\langle \Xi\psi,P\Xi\varphi\right\rangle }\\
 & = & \left\langle P\Xi\varphi,\Xi\psi\right\rangle \\
 & = & \left\langle \Xi\varphi,P\Xi\psi\right\rangle \\
 & = & \omega_{P}\left(B\left(\Xi\varphi\right)^{\ast}B\left(\Xi\psi\right)\right)\\
 & = & \omega_{P}\left(B\left(\varphi\right)B\left(\psi\right)^{\ast}\right)\,.
\end{eqnarray*}
We learn that 
\[
\omega_{P}\left(B\left(\psi\right)^{\ast}B\left(\varphi\right)\right)+\omega_{P}\left(B\left(\varphi\right)B\left(\psi\right)^{\ast}\right)=\left\langle \psi,\varphi\right\rangle \,.
\]
But thanks to linearity this constraint may thus be rewritten as 
\begin{eqnarray*}
\omega_{P}\left(\Set{B\left(\psi\right)^{\ast},B\left(\varphi\right)}-\left\langle \psi,\varphi\right\rangle 1_\calA\right) & = & 0\,.
\end{eqnarray*}
This condition is already true at the level of the algebra, for \emph{any
}state. Hence, the single-body particle-hole-symmetry constraint
is automatically fulfilled in the many-body self-dual setting. For that reason there is no need to "enforce" particle-hole symmetry for pure states on the self-dual CAR algebra $\calA$.

On the other hand, quasi-free states are always even and hence $\theta$-invariant,
and this shall remain a genuine constraint in the many-body setting
as well when we generalize from quasi-free pure to interacting pure states.

\subsection{The many-body Majorana number}
\begin{defn}[The many-body Majorana number]
 Let $\omega\in\mathcal{P}\left(\mathcal{A}\right)$ be a parity-invariant
$\sigma$-local pure state. Let $u\in\mathcal{U}\left(\mathcal{A}\right)$ be the
locality-unitary such that 
\begin{eqnarray*}
\omega\circ\sigma & = & \omega\circ\text{Ad}_{u}\,.
\end{eqnarray*}
Note that $\omega\circ\sigma\circ\theta=\omega\circ\sigma$ automatically.
We then define for any such pure state $\omega$ the following topological
index: 
\begin{eqnarray*}
\mathcal{N}_{\sigma}\left(\omega\right) & := & \calN(\omega,\omega\circ\sigma)\equiv\omega\left(u^{\ast}\theta\left(u\right)\right)\in\Set{\pm1}\,.
\end{eqnarray*}
\end{defn}

\begin{lem}
$\mathcal{N}_{\sigma}$ is well-defined and locally constant in the norm
topology. In particular, it is norm continuous.
\end{lem}

\begin{proof}
Well-definededness follows at once by construction, from \cref{thm:main theorem for properties of the Z2 index of a pair}.

Let $\omega$ be a parity-invariant pure state which is $\sigma$-local. Then for any $\widetilde{\omega}$ any other pure state which is also parity-invariant, and obeys $\norm{\omega-\widetilde{\omega}}<2$, $\widetilde{\omega}$ is also $\sigma$-local and $\mathcal{N}_{\sigma}(\omega)=\mathcal{N}_{\sigma}(\widetilde{\omega})$. Indeed, $\norm{\omega-\widetilde{\omega}}<2$ implies local-comparability so that we know 
\eq{
\calN(\omega,\widetilde{\omega}) = 1\,.
} The same can be said about $\omega\circ\sigma,\widetilde{\omega}\circ\sigma$ since the automorphism $\sigma$ preserves the norm distance. Hence 
\eq{
\calN(\omega\circ\sigma,\widetilde{\omega}\circ\sigma) = 1\,.
} Thus using the multiplicativity of the index we have
\eq{
\calN_\sigma(\widetilde{\omega}) &\equiv \calN(\widetilde{\omega},\widetilde{\omega}\circ\sigma)\\
&= \calN(\widetilde{\omega},\omega)\calN(\omega,\omega\circ\sigma)\calN(\omega\circ\sigma,\widetilde{\omega}\circ\sigma)\\
&= 1\cdot \calN(\omega,\omega\circ\sigma) \cdot 1\,.
}

\end{proof}
%
%Note that since there is no $U\left(1\right)$-symmetry in the context
%of the self-dual CAR algebra, there is no "additional" or "accidental" $\mathbb{Z}$-valued index
%in this setting.

\begin{example}[Non-trivial \(\calN_\sigma\) index: the Kitaev \(k\)-chain]\label{example:Kitaev chain}
Let \(N=1\), \(R\delta_x=\delta_{x+1}\), and \(k\geq1\). Set

$$
\Xi=-\calC(\Id\otimes\sigma_3),\qquad
H_k=\begin{bmatrix}0&R^{k*}\\ R^k&0\end{bmatrix},\qquad
P_k=\frac{\Id-H_k}{2}
=\frac12\begin{bmatrix}\Id&-R^{k*}\\-R^k&\Id\end{bmatrix}\, ,
$$ where $\calC$ is complex conjugation on $\ell^2(\ZZ)$: $\psi_x\mapsto\overline{\psi_x}$.

Then \(H_k=H_k^*\), \(H_k^2=\Id\), and \(\Xi H_k\Xi=-H_k\); hence \(\Xi P_k\Xi=P_k^\perp\), and \(P_k\) defines a parity-invariant pure quasi-free state \(\omega_{P_k}\). Writing

$$
\psi_n^-:=\frac1{\sqrt2}
\begin{bmatrix}\delta_n\\-\delta_{n+k}\end{bmatrix}\, ,
$$

we have \(\im(P_k)=\overline{\szpan}\{\psi_n^-:n\in\ZZ\}\). Moreover, a vector \(\bigl[\varphi,-R^k\varphi\bigr]^{\mathsf T}\in\im(P_k)\) belongs to \(\ker(P_k\Sigma P_k+P_k^\perp)\) precisely when

$$
(\Sigma+R^{k*}\Sigma R^k)\varphi=0\,.
$$

Since this diagonal operator vanishes exactly on \(\szpan\{\delta_{-k+1},\ldots,\delta_0\}\), it follows that

$$
\ker(P_k\Sigma P_k+P_k^\perp)
=\szpan\{\psi_{-k+1}^-,\ldots,\psi_0^-\},
\qquad
\dim\ker(P_k\Sigma P_k+P_k^\perp)=k\,.
$$

To construct a locality unitary, let

$$
Q_k=\sum_{\ell=1}^k\delta_\ell\otimes\delta_\ell^*,
\qquad D_k=\Id-2Q_k,\qquad
W_k=\begin{bmatrix}\Id&0\\0&D_k\end{bmatrix}\,.
$$

Then \(\rank(\Id-W_k)=k\), and the identities

$$
D_kR^k=\Sigma R^k\Sigma,\qquad
R^{k*}D_k=\Sigma R^{k*}\Sigma
$$

give \(W_k^*P_kW_k=\Sigma P_k\Sigma\). For \(f_\ell=[0,\delta_\ell]^{\mathsf T}\), define the Majorana unitaries \(\gamma_\ell=\sqrt2B(f_\ell)\) and \(w_k=\gamma_1\cdots\gamma_k\). Since \(\Xi f_\ell=f_\ell\), the self-dual CAR relations imply

$$
\Ad{w_k}(B(\psi))=B((-1)^kW_k\psi),
\qquad
\theta(w_k)=(-1)^kw_k\,.
$$

The factor \((-1)^k\) is implemented by \(\theta^k\), under which \(\omega_{P_k}\) is invariant; therefore

$$
\omega_{P_k}\circ\Ad{w_k}
=\omega_{W_k^*P_kW_k}
=\omega_{\Sigma P_k\Sigma}
=\omega_{P_k}\circ\sigma\,.
$$

Thus \(\omega_{P_k}\) is \(\sigma\)-local, with locality unitary \(w_k\), and

$$
\calN_\sigma(\omega_{P_k})
=\omega_{P_k}\bigl(w_k^*\theta(w_k)\bigr)
=(-1)^k\,.
$$

In particular, \(\calN_\sigma(\omega_{P_k})=-1\) for every odd \(k\).
\end{example}
\begin{claim}
\label{claim:symmetric local automorphism invariance}
Let \(\omega\) be a pure, \(\theta\)-invariant and \(\sigma\)-local
state on \(\calA\). Let
\(\alpha\in\operatorname{Aut}(\calA)\) commute with \(\theta\), and
suppose that
\eq{
    \alpha\circ\sigma\circ\alpha^{-1}\circ\sigma^{-1}
    =
    \Ad{v}
}
for some \(v\in\calU(\calA)\) satisfying \(\theta(v)=v\).
Then \(\omega\circ\alpha\) is pure, \(\theta\)-invariant and
\(\sigma\)-local, and
\eq{
    \calN_\sigma(\omega\circ\alpha)
    =
    \calN_\sigma(\omega)\,.
}
\end{claim}

\begin{proof}
Purity and \(\theta\)-invariance are immediate. To prove
\(\sigma\)-locality, choose \(w\in\calU(\calA)\) such that
\eq{
    \omega\circ\sigma
    =
    \omega\circ\Ad{w}\,.
}
The commutator identity implies
\eq{
    \alpha\circ\sigma\circ\alpha^{-1}
    =
    \Ad{v}\circ\sigma
}
and hence
\eq{
    \omega\circ\alpha\circ\sigma
    &=
    \omega\circ\Ad{v}\circ\sigma\circ\alpha
    \\
    &=
    \omega\circ\sigma\circ\Ad{\sigma(v)}\circ\alpha
    \\
    &=
    \omega\circ\Ad{w}\circ\Ad{\sigma(v)}\circ\alpha
    \\
    &=
    \omega\circ\alpha\circ
    \Ad{\alpha^{-1}\left(\sigma(v)w\right)}\,.
}
Thus \(\omega\circ\alpha\) is \(\sigma\)-local.

By automorphism invariance and multiplicativity of the relative index
from
\cref{thm:main theorem for properties of the Z2 index of a pair},
\eq{
    \calN_\sigma(\omega\circ\alpha)
    &=
    \calN\left(
        \omega,
        \omega\circ\Ad{v}\circ\sigma
    \right)
    \\
    &=
    \calN(\omega,\omega\circ\sigma)
    \calN\left(
        \omega\circ\sigma,
        \omega\circ\Ad{v}\circ\sigma
    \right)
    \\
    &=
    \calN_\sigma(\omega)
    \calN\left(
        \omega,
        \omega\circ\Ad{v}
    \right)
    \\
    &=
    \calN_\sigma(\omega)
    \omega\left(v^\ast\theta(v)\right)
    \\
    &=
    \calN_\sigma(\omega)\,.
}
\end{proof}

\subsection{The single-body Majorana number}
\label{subsec:single-body Majorana number}

Let $H=H^\ast\in\calB(\calH)$ be an invertible single-particle
Hamiltonian with particle-hole symmetry
\eq{
    \Xi H\Xi=-H
}
and let
\eql{\label{eq:Fermi projection and flattened Hamiltonian}
    P
    :=
    \chi_{(-\infty,0)}(H),
    \qquad
    \sgn(H)
    =
    \Id-2P\,.
}
It follows that
\eq{
    \Xi P\Xi=P^\perp\,.
}
Recall that
\eq{
    \Sigma
    =
    \sgn(X)
    =
    2\Lambda-\Id,
    \qquad
    \Lambda
    =
    \chi_{\NN}(X)\,.
}
We assume that $P$ is $\Sigma$-local, namely
\eq{
    [P,\Sigma]\in\calK(\calH)\,.
}
Equivalently, $[P,\Lambda]\in\calK(\calH)$. In particular, the
compressed operators appearing below are Fredholm

A standard real-space representative of the single-body Majorana
number is
\eql{\label{eq:single-body Majorana number projection formula}
    \calN_\sigma^{\rm non-int.}(P)
    :=
    \br{-1}^{\dim\ker\br{P\Sigma P+P^\perp}}
    \in
    \Set{\pm1}
    \cong
    \ZZ_2\, ,
}
see
\cite{GrossmannSchulzBaldes2016,KatsuraKoma2018,ChungShapiro2023}.
Equivalently, it may be expressed in terms of the half-space
compression of the flattened Hamiltonian as
\eql{\label{eq:single-body Majorana number Hamiltonian formula}
    \calN_\sigma^{\rm non-int.}(P)
    =
    \br{-1}^{\dim\ker\br{
        \Lambda\sgn(H)\Lambda+\Lambda^\perp
    }}\,.
}

We briefly verify the equivalence of
\cref{eq:single-body Majorana number projection formula,%
eq:single-body Majorana number Hamiltonian formula}.
Since $\Sigma=2\Lambda-\Id$, changing the sign of the compression on
$\im P$ does not change its kernel, and hence
\eq{
    \ker\br{P\Sigma P+P^\perp}
    =
    \ker\br{P(\Id-2\Lambda)P+P^\perp}\,.
}
On the other hand, by
\cref{eq:Fermi projection and flattened Hamiltonian},
\eq{
    \Lambda\sgn(H)\Lambda+\Lambda^\perp
    =
    \Lambda(\Id-2P)\Lambda+\Lambda^\perp\,.
}
The projection-pairing duality
\cite[Proposition~1]{GrossmannSchulzBaldes2016},
also recalled in
\cite[Proposition~24]{KatsuraKoma2018}, therefore gives
\eql{\label{eq:duality of single-body Majorana kernels}
    \dim\ker\br{P\Sigma P+P^\perp}
    =
    \dim\ker\br{
        \Lambda\sgn(H)\Lambda+\Lambda^\perp
    }\,.
}
More explicitly, the isomorphism between the two kernels is
$v\mapsto\Lambda v$, with inverse $w\mapsto 2Pw$.

We note that \cref{eq:duality of single-body Majorana kernels} is
purely projection-theoretic and does not use particle-hole symmetry.
Particle-hole symmetry is instead what makes the parity of the kernel
a homotopy-invariant $\ZZ_2$-index. The index is locally constant on
the space of $\Sigma$-local, particle-hole symmetric projections.
Moreover, \cite{ChungShapiro2023} proves that it is a complete
invariant: two such projections have the same index if and only if
they can be joined by a norm-continuous path of $\Sigma$-local,
$\Xi$-symmetric projections.

\subsection{Correspondence between the many-body and single-body indices}

\begin{thm}
Let $P=P^\ast=P^2\in\calB(\calH)$ be a particle-hole symmetric
projection, i.e.,
\eql{
    \Xi P\Xi=P^\perp\,.
}
Assume further that
\eql{
    [P,\Sigma]\in\calJ_2(\calH)\,.
}
Then $\omega_P$ is $\sigma$-local and the many-body index agrees with
the single-body index:
\eql{
    \calN_\sigma^{\rm non-int.}(P)
    =
    \calN_\sigma(\omega_P)\,.
}
Equivalently, if $u\in\calU(\calA)$ is any unitary such that
$\omega_P\circ\sigma=\omega_P\circ\Ad{u}$, then
\eql{\label{eq:single-body many-body correspondence}
    \br{-1}^{\dim\ker\br{P\Sigma P+P^\perp}}
    =
    \omega_P\br{u^\ast\theta(u)}\,.
}
\end{thm}

We remark that we expect the theorem to remain true, in an appropriately
generalized sense, when $[P,\Sigma]\in\calK(\calH)$. The single-body
index remains well-defined under this weaker hypothesis, but the
definition of $\sigma$-locality would have to be generalized.

\begin{proof}
Set
\eql{
    Q:=\Sigma P\Sigma\,.
}
Since $\Xi$ commutes with $\Sigma$ and $\Sigma^2=\Id$, we have
\eq{
    \Xi Q\Xi
    &=
    \Sigma\Xi P\Xi\Sigma\\
    &=
    \Sigma P^\perp\Sigma\\
    &=
    Q^\perp\,.
}
Thus $Q$ is also particle-hole symmetric. Moreover,
\eql{
    P-Q=[P,\Sigma]\Sigma\in\calJ_2(\calH)\,.
}
The Bogoliubov automorphism $\sigma$ transforms the quasi-free state
$\omega_P$ according to
\eql{
    \omega_P\circ\sigma=\omega_Q\,.
}
We may therefore apply
\cref{thm:Hilbert-space CAR correspondence} to the pair $P,Q$. It
follows that $\omega_P$ and $\omega_Q$ are locally-comparable, so
$\omega_P$ is $\sigma$-local, and
\eq{
    \calN_\sigma(\omega_P)
    &=
    \calN(\omega_P,\omega_P\circ\sigma)\\
    &=
    \calN(\omega_P,\omega_Q)\\
    &=
    \br{-1}^{\dim\im P\cap\ker Q}\,.
}

It remains only to identify the kernel appearing in the single-body
Majorana number. Since $\Sigma$ is unitary, we have
\eq{
    \ker\br{P\Sigma P+P^\perp}
    &=
    \Set{\psi\in\im P:P\Sigma\psi=0}\\
    &=
    \im P\cap\ker\br{\Sigma P\Sigma}\\
    &=
    \im P\cap\ker Q\,.
}
Consequently,
\eq{
    \calN_\sigma^{\rm non-int.}(P)
    &=
    \br{-1}^{\dim\ker\br{P\Sigma P+P^\perp}}\\
    &=
    \br{-1}^{\dim\im P\cap\ker Q}\\
    &=
    \calN_\sigma(\omega_P)\,.
}
Finally, if $u\in\calU(\calA)$ implements the locality relation
$\omega_P\circ\sigma=\omega_P\circ\Ad{u}$, then by definition
\eql{
    \calN_\sigma(\omega_P)
    =
    \omega_P\br{u^\ast\theta(u)}\, ,
}
which proves \cref{eq:single-body many-body correspondence}.
\end{proof}

\subsection{Which topology should one use for the index?}
Let
\eq{
\calP_{\theta,\sigma\mathrm{-loc}}(\calA)
:=
\Set{
\omega\in\calP(\calA):
\omega\circ\theta=\omega,\
\exists u\in\calU(\calA)\text{ such that }
\omega\circ\sigma=\omega\circ\Ad{u}
}.
}
By the continuity result above, the map
\eq{
\calP_{\theta,\sigma\mathrm{-loc}}(\calA)
\ni
\omega
\longmapsto
\calN_\sigma(\omega)
\in \Set{\pm1}\cong\ZZ_2
}
is locally constant for the norm topology on the state space, i.e. the topology induced by the norm of $\calA$. We now explain why this is the natural topology for the index, and why the weak-* topology is too coarse.

If one temporarily forgets the locality condition, then in the CAR/UHF setting considered here the pure-state space is weak-* path connected; in fact, for nonelementary simple separable real-rank-zero C* algebras it is weakly contractible \cite{SpiegelPflaum2025WeakContractibility}. Thus one might hope to view the index as measuring the obstruction to replacing an arbitrary weak-$*$ path by one that remains local in the appropriate sense. This is morally analogous to the one-particle picture in \cite{ChungShapiro2023}: without the half-space locality constraint, the relevant infinite-dimensional unitary group is norm-contractible by Kuiper's theorem \cite{Kuiper1965HomotopyType}; the topological obstruction appears only after imposing the compatibility, or ``stitching,'' condition across the spatial cut.

However, the topology in which this obstruction is stable is the norm topology, not the weak-$*$ topology. The norm topology remembers superselection sectors: if two pure states have norm distance strictly smaller than $2$, then their GNS representations are unitarily equivalent, and the states are inner-unitarily conjugate \cite[Corollaries~8--9]{GlimmKadison1960}. Consequently, a norm-continuous path of pure states cannot pass from one GNS sector to another. Indeed, by compactness of the interval, such a path can be subdivided into pieces of norm length strictly smaller than $2$, so all states along the path lie in the same sector. Already for quasi-free CAR states this gives many disconnected norm components: by the Araki--Powers--St{\o}rmer quasi-equivalence criterion, two pure quasi-free states defined by projections $P$ and $Q$ are in the same GNS sector only if $P-Q\in\calJ_2(\calH)$ \cite{Araki1970QuasifreeCAR,PowersStormer1970FreeStates}. In particular, diagonal product states whose defining projections differ on infinitely many sites lie in distinct sectors.

This sector sensitivity should not be confused with total disconnectedness of the pure-state space in norm. Within a fixed sector there are typically many norm-continuous paths, for instance
\eq{
t\longmapsto \omega\circ\Ad{\ee^{\ii t h}},
\qquad h=h^\ast\in\calA\,.
}
The point is rather that the norm topology is fine enough to see both superselection sectors and the $\ZZ_2$-valued relative obstruction defined above. The weak-* topology, by contrast, only tests against fixed local observables, and therefore may fail to detect index-carrying changes occurring at spatial infinity. 

Indeed, there exists a weak-$*$ continuous path of pure states within
$\calP_{\theta,\sigma\mathrm{-loc}}(\calA)$ whose endpoints have different
indices.

\begin{example}[Weak-$*$ discontinuity at a pure state]
We construct a weak-$*$ continuous path
\eq{
[0,1]\ni t\longmapsto\omega_t
\in\calP_{\theta,\sigma\mathrm{-loc}}(\calA)
}
such that
\eq{
\calN_\sigma(\omega_t)=-1
\quad (0\leq t<1),
\qquad
\calN_\sigma(\omega_1)=+1\,.
}

We work on
\eq{
\calH=\ell^2(\ZZ)\otimes\CC^2
}
with the particle-hole conjugation chosen as in the Kitaev-chain example, so
that the projections below obey $\Xi P\Xi=P^\perp$. Let
$\Lambda=\chi_\NN(X)$.

For each $n\geq1$, define a permutation $\tau_n:\ZZ\to\ZZ$ by
\eq{
\tau_n(x)=
\begin{cases}
x & -n\leq x\leq n\\
x+1 & x\geq n+1\\
x+1 & x\leq -n-2\\
n+1 & x=-n-1
\end{cases}\,.
}
Thus $\tau_n$ is the identity on $\Set{-n,\dots,n}$, while away from this
interval it shifts both tails one step to the right, moving one basis vector
from the far left tail to the far right tail. Let $V_n$ be the corresponding
real unitary on $\ell^2(\ZZ)$,
\eq{
V_n\delta_x:=\delta_{\tau_n(x)}\,.
}
Then
\eq{
V_n\stackrel{s}{\longrightarrow}\Id,
\qquad
V_n^*\stackrel{s}{\longrightarrow}\Id\,.
}

We next interpolate between $V_n$ and $V_{n+1}$. Set
\eq{
a_n:=-n-2,
\qquad
b_n:=-n-1,
\qquad
c_n:=n+1\,.
}
A direct calculation gives
\eq{
R_n:=V_n^*V_{n+1}\, ,
}
where $R_n$ is the identity outside
\eq{
E_n:=\szpan\Set{\delta_{a_n},\delta_{b_n},\delta_{c_n}}
}
and acts on these three basis vectors by
\eq{
\delta_{a_n}\longmapsto\delta_{c_n},
\qquad
\delta_{c_n}\longmapsto\delta_{b_n},
\qquad
\delta_{b_n}\longmapsto\delta_{a_n}\,.
}
This is an even three-cycle. In particular, its restriction to $E_n$ belongs
to $\SO(3)$. Choose a norm-continuous path of real unitaries
\eq{
[0,1]\ni s\longmapsto R_n(s)
}
which is the identity on $E_n^\perp$ and satisfies
\eq{
R_n(0)=\Id,
\qquad
R_n(1)=R_n\,.
}
For example, on $E_n$ one may take the rotation through angle $2\pi s/3$
around the axis spanned by
$\delta_{a_n}+\delta_{b_n}+\delta_{c_n}$, with the appropriate orientation.

Define
\eq{
V_{n,s}:=V_nR_n(s)\,.
}
Then $V_{n,s}$ is a real unitary and
\eq{
V_{n,0}=V_n,
\qquad
V_{n,1}=V_{n+1}\,.
}
Moreover,
\eq{
[V_{n,s},\Lambda]
}
is finite-rank for every $n$ and $s$.

For a real unitary $V$ on $\ell^2(\ZZ)$, set
\eq{
P(V)
:=
\onehalf
\begin{pmatrix}
\Id & -V^*\\
-V & \Id
\end{pmatrix}\,.
}
Then $P(V)$ is a projection satisfying
\eq{
\Xi P(V)\Xi=P(V)^\perp\,.
}
Consequently, $P(V)$ defines a pure, $\theta$-invariant quasi-free state
$\omega_{P(V)}$

For $n\geq1$ and $s\in[0,1]$, consider the compressed operator
\eq{
F_{n,s}:=\Lambda V_{n,s}\Lambda+\Lambda^\perp\,.
}
Since $[V_{n,s},\Lambda]$ is finite-rank, $F_{n,s}$ is Fredholm. The map
$s\mapsto F_{n,s}$ is norm-continuous, and hence its Fredholm index is
constant. At $s=0$,
\eq{
F_{n,0}=\Lambda V_n\Lambda+\Lambda^\perp\,.
}
On the right half-chain, $V_n$ fixes
$\delta_1,\dots,\delta_n$ and shifts
$\delta_{n+1},\delta_{n+2},\dots$ one step to the right. Therefore
\eq{
\ker(F_{n,0})=\Set{0},
\qquad
\dim\coker(F_{n,0})=1
}
and hence
\eq{
\findex(F_{n,s})=-1
\qquad
(n\geq1,\ s\in[0,1])\,.
}
Furthermore, $[P(V_{n,s}),\Sigma]$ is finite-rank and therefore belongs to
$\calJ_2(\calH)$. The correspondence theorem now gives
\eq{
\calN_\sigma\bigl(\omega_{P(V_{n,s})}\bigr)=-1
\qquad
(n\geq1,\ s\in[0,1])\,.
}

We concatenate these paths near $t=1$. Set
\eq{
t_n:=1-\frac1n,
\qquad
s_n(t):=\frac{t-t_n}{t_{n+1}-t_n}
\quad
\bigl(t\in[t_n,t_{n+1}]\bigr)
}
and define
\eq{
V(t):=
\begin{cases}
V_{n,s_n(t)} & t\in[t_n,t_{n+1}]\\
\Id & t=1
\end{cases}\,.
}
This is well-defined because
\eq{
V_{n,1}=V_{n+1}=V_{n+1,0}\,.
}
It is norm-continuous on $[0,1)$. Moreover, if
$t\in[t_n,t_{n+1}]$, then
\eq{
V(t)\delta_x=\delta_x
\qquad
(-n\leq x\leq n)\,.
}
It follows that
\eq{
V(t)\stackrel{s}{\longrightarrow}\Id,
\qquad
V(t)^*\stackrel{s}{\longrightarrow}\Id
\qquad
(t\longrightarrow1)\,.
}

Finally, define
\eq{
P(t):=P(V(t)),
\qquad
\omega_t:=\omega_{P(t)}\,.
}
Every $P(t)$ is a particle-hole-symmetric projection, so every $\omega_t$ is
pure and $\theta$-invariant. Furthermore,
\eq{
[P(t),\Sigma]\in\calJ_2(\calH)\, ,
}
so every $\omega_t$ is $\sigma$-local. Thus
\eq{
\omega_t\in\calP_{\theta,\sigma\mathrm{-loc}}(\calA)
\qquad
(t\in[0,1])\,.
}

The map $t\mapsto P(t)$ is weak-operator continuous, including at $t=1$.
Quasi-free correlation functions are Pfaffians of matrix coefficients of
$P(t)$, and hence depend continuously on $t$. By norm density of the local
CAR polynomials, it follows that
\eq{
[0,1]\ni t\longmapsto\omega_t
}
is weak-$*$ continuous.

For $t<1$, the construction and the correspondence theorem give
\eq{
\calN_\sigma(\omega_t)=-1\,.
}
At $t=1$, however,
\eq{
V(1)=\Id,
\qquad
P(1)
=
P_\infty
:=
\onehalf
\begin{bmatrix}
\Id & -\Id\\
-\Id & \Id
\end{bmatrix}\,.
}
The corresponding compressed operator is
\eq{
F_\infty=\Lambda\Id\Lambda+\Lambda^\perp=\Id
}
and therefore
\eq{
\calN_\sigma(\omega_1)
=
\calN_\sigma(\omega_{P_\infty})
=
+1\,.
}
We have thus constructed a weak-$*$ continuous path in
$\calP_{\theta,\sigma\mathrm{-loc}}(\calA)$ satisfying
\eq{
\calN_\sigma(\omega_0)=-1,
\qquad
\calN_\sigma(\omega_1)=+1\,.
}
In fact, the index is equal to $-1$ on the entire half-open interval
$[0,1)$ and jumps to $+1$ at the pure limiting state $\omega_1$.

Intuitively, the nontrivial index is carried by a unit of transport which is
pushed continuously toward spatial infinity. Every fixed local observable
eventually sees only the pure limiting state, while the half-space Fredholm
index continues to detect the transported unit at every finite time.
\end{example}

We conclude that it is probably \emph{inappropriate }to try to classify
the space of pure states w.r.t. the weak-* topology.
\begin{claim}
If $P_{n}\to P$ in WOT then $\omega_{P_{n}}\to\omega_{P}$ in the
weak-* topology.
\end{claim}

\begin{proof}
Let $P_{n}\to P$ in WOT. Let $a\in\mathcal{A}$. Then $a$ is the
norm-limit of some polynomial of $a\left(\psi_{1}\right)\cdots a\left(\psi_{m}\right)$,
possibly with stars on some of these. But we have
\[
\left|\omega_{P_{n}}\left(a\left(\psi\right)^{\ast}a\left(\varphi\right)\right)-\omega_{P}\left(a\left(\psi\right)^{\ast}a\left(\varphi\right)\right)\right|=\left|\left\langle \psi,P_{n}\varphi\right\rangle -\left\langle \psi,P\varphi\right\rangle \right|\to0\,.
\]

On the other hand, if $\omega_{P_{n}}\to\omega_{P}$ in weak-star
topology, then in particular 
\[
\left|\left\langle \psi,P_{n}\varphi\right\rangle -\left\langle \psi,P\varphi\right\rangle \right|\to0
\]
for all $\psi,\varphi\in\mathcal{H}$ so that $P_{n}\to P$ in WOT.
\end{proof}
\begin{lem}[Norm convergence of pure quasi-free states]\label{lem:HS-topology-qf-norm-convergence}
Let $P_n$ and $P$ be basis projections on the self-dual Hilbert space
$(\calH,\Xi)$, and let $\omega_{P_n}$ and $\omega_P$ be the associated pure
quasi-free states on $\calA=\mathrm{CAR}_{\mathrm{SD}}(\calH,\Xi)$. Then
\eql{\label{eq:HS-topology-qf-norm-convergence}
\norm{\omega_{P_n}-\omega_P}_{\calA^*}\longrightarrow 0
\qquad\Longleftrightarrow\qquad
\norm{P_n-P}_{\calJ_2(\calH)}\longrightarrow 0\,.
}
Thus, on the space of basis projections, the topology corresponding to the norm
topology on the associated pure quasi-free states is the Hilbert--Schmidt
topology.
\end{lem}
\begin{proof}
For pure states,

$$
\norm{\varphi-\psi}_{\calA^*}
=2\sqrt{1-T(\varphi,\psi)}\, ,
$$

while the Shale--Stinespring--Araki criterion and Araki’s determinant formula give, for basis projections \(P,Q\),

$$
T(\omega_P,\omega_Q)=
\begin{cases}
\det\!\bigl(\Id-(P-Q)^2\bigr)^\gamma
& P-Q\in\calJ_2(\calH)\\
0&P-Q\notin\calJ_2(\calH)
\end{cases}
\qquad \gamma>0\, ,
$$

where the value of \(\gamma\) depends only on conventions
\cite{Araki1970QuasifreeCAR,PowersStormer1970FreeStates}. If
\(\norm{P_n-P}_{\calJ_2}\to0\), then

$$
\norm{(P_n-P)^2}_{\calJ_1}
=\norm{P_n-P}_{\calJ_2}^2\longrightarrow0\, ,
$$

so continuity of the Fredholm determinant yields
\(\det(\Id-(P_n-P)^2)\to1\). Hence
\(T(\omega_{P_n},\omega_P)\to1\), and therefore
\(\norm{\omega_{P_n}-\omega_P}_{\calA^*}\to0\). Conversely, suppose the latter convergence holds. Then \(T(\omega_{P_n},\omega_P)\to1\), so for all sufficiently large \(n\), \(P_n-P\in\calJ_2(\calH)\) and

$$
\det(\Id-K_n)\longrightarrow1,
\qquad K_n:=(P_n-P)^2\in\calJ_1(\calH),\quad 0\leq K_n\leq\Id\,.
$$

For such \(n\), the determinant is positive, and \(-\log(1-t)\geq t\) on \([0,1)\) gives

$$
\norm{P_n-P}_{\calJ_2}^2
=\tr(K_n)
\leq\tr\!\bigl(-\log(\Id-K_n)\bigr)
=-\log\det(\Id-K_n)
\longrightarrow0\,.
$$

Thus \(\norm{P_n-P}_{\calJ_2}\to0\), proving \cref{eq:HS-topology-qf-norm-convergence}.
\end{proof}

\section{Automorphic-path-connectedness}
\label{sec:automorphism path connectedness}
A basic fact in Hilbert space is that, in operator norm topology, all unitaries are path-connected \cite{Kuiper1965HomotopyType}. This turned out to be a crucial point in the one-dimensional classification scheme of insulators \cite{ChungShapiro2023}. For one, this fact immediately implies that all non-trivial projections (with infinite range and kernel) are path-connected via \eq{
t\mapsto U_t^\ast P U_t\,.
} The idea then is that without locality or any symmetries, the topology is set up in such a way so as to make states path-connected.

What is the analog of that in the interacting setting? As we saw in the preceding section, the answer is not so simple: operator norm continuous paths cannot cross selection sectors, and while the weak-star topology does make all pure states path-connected, that topology washes away the index.

One possible idea out of this is by undoing the collapsing between unitaries and projections in Hilbert space: conjugation by a unitary in Hilbert space should be associated rather with automorphisms in the interacting setting: \eq{
P\mapsto U^\ast P U \qquad \longleftrightarrow\qquad \omega\mapsto\omega\circ \alpha\,.
} Hence instead of requiring directly that $\omega_1,\omega_2$ be path-connected, for any $\omega_1,\omega_2$, we seek a continuous path of automorphisms connecting them. This is generally referred to as homogeneity of pure states w.r.t. automorphisms \cite{KishimotoOzawaSakai2003Homogeneity} and we consider it the appropriate analog of \cite{Kuiper1965HomotopyType} in the interacting setting. In fact \cite{KishimotoOzawaSakai2003Homogeneity} prove that such automorphisms may be taken as asymptotically inner. Of course point-norm continuity of this path of automorphisms then just induces a weak-star continuous path of pure-states, so have we gained anything? Actually yes, because the path must arise within the space of automorphisms on $\calA$, and moreover, we shall put the locality constraints on the automorphism itself. This turns out to yield a path that does preserve the index. 

An alternative would have been to define a finer notion of connectivity on top of weak-star connected paths that bears additional constraints, so as to preserve the index. For example, one could say $\omega_1,\omega_2$ are "admissible" weak-star connected if there is a weak-star continuous path $\omega_1,\omega_2$ passing within $\calP_{\theta,\sigma\mathrm{-loc}}(\calA)$ such that, for any $t$, 
\eq{
\omega_t\circ\sigma = \omega_t\circ\Ad{u_t}
} for some unitary $u_t$ and moreover, $t\mapsto u_t$ is \emph{norm} continuous. However, the only proof we could find for completeness with respect to that notion of connectedness would pass through the automorphic path connected notion we are about to use anyway, so we find the latter the more natural notion. We describe it next.

For an automorphism \(\alpha\in\operatorname{Aut}(\calA)\), define its
multiplicative commutator with the locality automorphism by
\eq{
    \partial_\sigma(\alpha)
    &:=
    \alpha\circ\sigma\circ\alpha^{-1}\circ\sigma^{-1}\,.
}

\begin{defn}[Symmetric local automorphism path]
\label{def:symmetric local automorphism path}
A \emph{symmetric local automorphism path} consists of a point-norm
continuous path
\eq{
    [0,1]\ni t
    \longmapsto
    \alpha_t\in\Automorphisms{\calA}
}
and a norm-continuous path
\eq{
    [0,1]\ni t
    \longmapsto
    v_t\in\calU(\calA^\theta)
}
such that, for every \(t\in[0,1]\),
\eq{
    \alpha_0
    &=
    \mathrm{id},
    &
    v_0
    &=
    1_\calA\\
    \alpha_t\circ\theta
    &=
    \theta\circ\alpha_t,
    &
    \partial_\sigma(\alpha_t)
    &=
    \Ad{v_t}\,.
}
\end{defn}

\begin{defn}[automorphic-path-connectedness]
\label{def:automorphism path connectedness}
Let \(\omega_1,\omega_2\) be pure, \(\theta\)-invariant,
\(\sigma\)-local states on \(\calA\). We say that
\(\omega_1\) and \(\omega_2\) are \emph{automorphic-path-connected}
iff there exists a symmetric local automorphism path
\([0,1]\ni t\mapsto(\alpha_t,v_t)\) such that
\eql{
    \omega_2
    =
    \omega_1\circ\alpha_1\,.
}
\end{defn}

\begin{claim}
automorphic-path-connectedness is an equivalence relation on the
space of pure states that are \(\theta\)-invariant and
\(\sigma\)-local.
\end{claim}

\begin{proof}
Reflexivity follows from the constant path
\eq{
    \alpha_t
    &=
    \mathrm{id},&
    v_t
    &=
    1_\calA\,.
}

We next prove symmetry. Let
\(t\mapsto(\alpha_t,v_t)\) be a symmetric local automorphism path
satisfying
\eq{
    \omega_2
    =
    \omega_1\circ\alpha_1\,.
}
First note that, if
\eq{
    \partial_\sigma(\alpha)
    =
    \Ad{v}\, ,
}
then
\eql{\label{eq:locality commutator of inverse}
    \partial_\sigma(\alpha^{-1})
    =
    \Ad{\alpha^{-1}(v^*)}\,.
}
Indeed,
\eq{
    \partial_\sigma(\alpha^{-1})
    &=
    \alpha^{-1}\circ
    \partial_\sigma(\alpha)^{-1}
    \circ\alpha\\
    &=
    \alpha^{-1}\circ\Ad{v^*}\circ\alpha\\
    &=
    \Ad{\alpha^{-1}(v^*)}\,.
}

Define
\eq{
    \widetilde\alpha_t
    &:=
    \alpha_1^{-1}\circ\alpha_{1-t},\\
    \widetilde v_t
    &:=
    \alpha_1^{-1}
    \br{
        v_1^*v_{1-t}
    }\,.
}
By \cref{eq:locality commutator of inverse} and
\cref{lem:composition of symmetric local automorphisms},
\eq{
    \partial_\sigma(\widetilde\alpha_t)
    &=
    \Ad{
        \alpha_1^{-1}(v_1^*)
        \alpha_1^{-1}(v_{1-t})
    }\\
    &=
    \Ad{\widetilde v_t}\,.
}
Moreover,
\eq{
    \widetilde\alpha_0
    &=
    \mathrm{id},&
    \widetilde v_0
    &=
    1_\calA,&
    \widetilde\alpha_1
    &=
    \alpha_1^{-1}\,.
}
The maps \(t\mapsto\widetilde\alpha_t\) and
\(t\mapsto\widetilde v_t\) have the required continuity properties.
Since every \(\alpha_t\) commutes with \(\theta\) and every \(v_t\)
is \(\theta\)-even,
\eq{
    \widetilde\alpha_t\circ\theta
    &=
    \theta\circ\widetilde\alpha_t,&
    \theta(\widetilde v_t)
    &=
    \widetilde v_t\,.
}
Finally,
\eq{
    \omega_1
    =
    \omega_2\circ\alpha_1^{-1}
    =
    \omega_2\circ\widetilde\alpha_1\,.
}
Thus automorphic-path-connectedness is symmetric.

It remains to prove transitivity. Suppose that
\(t\mapsto(\alpha_t,v_t)\) connects \(\omega_1\) to \(\omega_2\), and
\(t\mapsto(\beta_t,w_t)\) connects \(\omega_2\) to \(\omega_3\).
Define
\eq{
    (\gamma_t,z_t)
    :=
    \begin{cases}
        \br{
            \alpha_{2t},
            v_{2t}
        }
        &
        0\leq t\leq\onehalf\\[2mm]
        \br{
            \alpha_1\circ\beta_{2t-1},
            v_1\alpha_1(w_{2t-1})
        }
        &
        \onehalf\leq t\leq1
    \end{cases}\,.
}
The two definitions agree at \(t=\onehalf\), since
\eq{
    \beta_0
    &=
    \mathrm{id},&
    w_0
    &=
    \Id\,.
}
By \cref{lem:composition of symmetric local automorphisms},
\eq{
    \partial_\sigma
    \br{
        \alpha_1\circ\beta_{2t-1}
    }
    =
    \Ad{
        v_1\alpha_1(w_{2t-1})
    }\,.
}
Hence \(t\mapsto(\gamma_t,z_t)\) is a symmetric local automorphism
path. Its endpoint satisfies
\eq{
    \omega_1\circ\gamma_1
    &=
    \omega_1\circ\alpha_1\circ\beta_1\\
    &=
    \omega_2\circ\beta_1\\
    &=
    \omega_3\,.
}
Therefore automorphic-path-connectedness is transitive, and hence is
an equivalence relation.
\end{proof}

The implementer \(v_t\) in
\cref{def:symmetric local automorphism path} is included as part of the
path data, since point-norm continuity of \(t\mapsto\Ad{v_t}\) does not
by itself guarantee a norm-continuous choice of implementers \(v_t\).

\begin{lem}[Composition of symmetric local automorphisms]
\label{lem:composition of symmetric local automorphisms}
Suppose that
\eq{
    \partial_\sigma(\alpha)
    &=
    \Ad{v},&
    \partial_\sigma(\beta)
    &=
    \Ad{w}\,.
}
Then
\eql{
    \partial_\sigma(\alpha\circ\beta)
    =
    \Ad{v\alpha(w)}\,.
}
In particular, if \(\alpha,\beta\) commute with \(\theta\) and \(v,w\)
are \(\theta\)-even, then \(v\alpha(w)\) is \(\theta\)-even
\end{lem}

\begin{proof}
A direct calculation gives
\eq{
    \partial_\sigma(\alpha\circ\beta)
    &=
    \alpha\circ\beta\circ\sigma\circ
    \beta^{-1}\circ\alpha^{-1}\circ\sigma^{-1}
    \\
    &=
    \alpha\circ
    \left(
        \beta\circ\sigma\circ\beta^{-1}\circ\sigma^{-1}
    \right)
    \circ\alpha^{-1}
    \circ
    \left(
        \alpha\circ\sigma\circ\alpha^{-1}\circ\sigma^{-1}
    \right)
    \\
    &=
    \Ad{\alpha(w)}\circ\Ad{v}
    \\
    &=
    \Ad{v\alpha(w)}\,.
}
If \(\theta(v)=v\), \(\theta(w)=w\), and
\(\alpha\circ\theta=\theta\circ\alpha\), then
\eq{
    \theta\left(v\alpha(w)\right)
    =
    v\alpha(w)\,.
}
\end{proof}

By \cref{claim:symmetric local automorphism invariance}, symmetric local
automorphism paths preserve the space of pure, \(\theta\)-invariant,
\(\sigma\)-local states. Moreover, if
\(t\mapsto(\alpha_t,v_t)\) is a symmetric local automorphism path and
\(\omega\) belongs to this space, then
\eql{\label{eq:index invariance along automorphism paths}
    \calN_\sigma(\omega\circ\alpha_t)
    =
    \calN_\sigma(\omega)
    \qquad
    (t\in[0,1])\,.
}
By \cref{claim:symmetric local automorphism invariance}, the conclusion
\cref{eq:index invariance along automorphism paths} requires only the
corresponding conditions on a fixed automorphism: if \(\alpha\) commutes
with \(\theta\) and
\(\partial_\sigma(\alpha)=\Ad{v}\) for some
\(v\in\calU(\calA^\theta)\), then
\(\calN_\sigma(\omega\circ\alpha)=\calN_\sigma(\omega)\). Thus the
existence of a symmetric local automorphism path is stronger than what
is needed for index invariance. The point of the remainder of this
section is to prove the converse in this stronger sense: equality of the
indices implies the existence of such a path, not merely of a single
symmetric local automorphism.

\subsection{Equivariant homogeneity of invariant pure states}
The following \cref{lem:equivariant homogeneity} is an equivariant version of the pure-state
 homogeneity theorem of Kishimoto, Ozawa, and Sakai
 \cite{KishimotoOzawaSakai2003Homogeneity}. Its proof largely follows
 their local-adjustment and approximate-intertwining argument, with the
 correcting unitaries chosen in the fixed-point algebra. A closely
 related equivariant adaptation of the same argument, in the setting of
 finite-group-symmetric quantum spin chains and using the notion of
 split states, appears in
 \cite[Section~4]{Ogata2021PureStates}.

The proof of equivariant homogeneity has two inputs. The first is the
following fixed-point version of the weak local-adjustment property of
\cite[Lemma~2.2]{KishimotoOzawaSakai2003Homogeneity}. The second is the
approximate-intertwining argument in the proof of
\cite[Theorem~2.1]{KishimotoOzawaSakai2003Homogeneity}. The point is that
the first input may be implemented by unitaries in the fixed-point
algebra, after which the second preserves equivariance at every stage.

\begin{lem}[Equivariant local adjustment]
\label{lem:equivariant local adjustment}
Let \(\calB\) be a simple, separable, unital, infinite-dimensional
C*-algebra, and let
\eq{
    \gamma:G\longrightarrow\operatorname{Aut}(\calB)
}
be a pointwise outer action of a finite group \(G\). Let \(\varphi\) be
a pure \(\gamma\)-invariant state on \(\calB\). For every finite set
\(\calF\subset\calB\) and every \(\varepsilon>0\), there are a finite
set \(\calG_0\subset\calB^\gamma\) and \(\delta>0\) with the following
property.

Suppose that \(\psi\) is a pure \(\gamma\)-invariant state satisfying
\eq{
    \left|
        \psi(d)-\varphi(d)
    \right|
    <
    \frac{\delta}{2}
    \qquad
    (d\in\calG_0)\,.
}
Then, for every finite set \(\calF'\subset\calB\) and every
\(\varepsilon'>0\), there exists a norm-continuous path
\eq{
    [0,1]\ni t
    \longmapsto
    u_t\in\calU(\calB^\gamma)
}
such that
\eq{
    u_0
    &=
    1_{\calB^\gamma},\\
    \norm{
        \Ad{u_t}(a)-a
    }
    &<
    \varepsilon
    \qquad
    (a\in\calF,\ t\in[0,1]),\\
    \left|
        \psi(a)
        -
        \br{
            \varphi\circ\Ad{u_1}
        }(a)
    \right|
    &<
    \varepsilon'
    \qquad
    (a\in\calF')\,.
}
\end{lem}

The proof is given in
\cref{app:equivariant local adjustment}.

\begin{lem}[Equivariant homogeneity]
\label{lem:equivariant homogeneity}
Let \(\calB\) be a simple, separable, unital, infinite-dimensional
C*-algebra, and let
\eq{
    \gamma:G\longrightarrow\operatorname{Aut}(\calB)
}
be a pointwise outer action of a finite group \(G\). Let
\(\varphi_0,\varphi_1\) be pure \(\gamma\)-invariant states on
\(\calB\). Then there exists a point-norm continuous path
\eq{
    [0,1]\ni t
    \longmapsto
    \beta_t\in\operatorname{Aut}(\calB)
}
such that
\eq{
    \beta_0
    &=
    \mathrm{id},\\
    \beta_t\circ\gamma_g
    &=
    \gamma_g\circ\beta_t
    \qquad
    (g\in G,\ t\in[0,1])\\
    \varphi_1
    &=
    \varphi_0\circ\beta_1\,.
}
Moreover, \(\beta_1\) is asymptotically inner through
\(\gamma\)-invariant unitaries.
\end{lem}

\begin{proof}
Since \(\gamma\) is pointwise outer and \(\calB\) is simple, the
fixed-point algebra \(\calB^\gamma\) is simple and
infinite-dimensional. The fixed-sector argument in
\cref{app:equivariant local adjustment} also shows that the
restriction of a pure \(\gamma\)-invariant state to
\(\calB^\gamma\) is pure. Thus the ordinary homogeneity theorem
\cite[Theorem~2.1]{KishimotoOzawaSakai2003Homogeneity}, applied to
\(\calB^\gamma\), supplies the initial approximation by a path in
\(\calU(\calB^\gamma)\). Here we use that a \(\gamma\)-invariant
state is determined by its restriction to \(\calB^\gamma\), since
\eq{
    \varphi(a)
    =
    \varphi(E_\gamma(a))\,.
}

We now apply the induction in the proof of
\cite[Theorem~2.1]{KishimotoOzawaSakai2003Homogeneity}, replacing
their Lemma~2.2 by
\cref{lem:equivariant local adjustment}. All states occurring in the
induction remain pure and \(\gamma\)-invariant. At the \(n\)-th stage,
the finite set on which the correcting path is required to be nearly
central is enlarged to contain the first \(n\) elements of a dense
sequence in the unit ball of \(\calB\), the required inverse images
under the accumulated partial products, and all their \(G\)-translates.
The finite set on which the two states are compared is enlarged
independently, as allowed by
\cref{lem:equivariant local adjustment}. The errors are chosen
summably. This is precisely the alternating construction following
\cite[Lemma~2.2]{KishimotoOzawaSakai2003Homogeneity}; the only change
is that every correcting path belongs to \(\calU(\calB^\gamma)\).

The resulting forward and inverse partial products converge
point-norm. Hence there are norm-continuous paths
\eq{
    [0,\infty)\ni s
    \longmapsto
    v_s^{(j)}
    \in
    \calU(\calB^\gamma),
    \qquad
    v_0^{(j)}
    =
    \Id,
    \qquad
    j\in\Set{0,1},
}
such that
\eq{
    \beta^{(j)}(a)
    &:=
    \lim_{s\to\infty}
    \Ad{v_s^{(j)}}(a)
    \qquad
    a\in\calB
}
defines an automorphism \(\beta^{(j)}\) of \(\calB\), with
\eq{
    \br{
        \beta^{(j)}
    }^{-1}(a)
    &=
    \lim_{s\to\infty}
    \Ad{\br{v_s^{(j)}}^*}(a)
    \qquad
    a\in\calB
}
and
\eq{
    \varphi_0\circ\beta^{(0)}
    =
    \varphi_1\circ\beta^{(1)}\,.
}
Since every \(v_s^{(j)}\) belongs to \(\calB^\gamma\), both limiting
automorphisms commute with \(\gamma\).

Set
\eq{
    \beta
    :=
    \beta^{(0)}
    \circ
    \br{
        \beta^{(1)}
    }^{-1}\,.
}
Then
\eq{
    \varphi_1
    =
    \varphi_0\circ\beta\, ,
}
and \(\beta\) commutes with \(\gamma\). Moreover, with the convention
\(\Ad{u}(a)=u^*au\), set
\eq{
    z_s
    :=
    \br{
        v_s^{(1)}
    }^*
    v_s^{(0)}
    \in
    \calU(\calB^\gamma)\,.
}
Then
\eq{
    \Ad{z_s}
    =
    \Ad{v_s^{(0)}}
    \circ
    \Ad{\br{v_s^{(1)}}^*}\, ,
}
and the convergence of the forward and inverse partial products gives
\eq{
    \beta(a)
    =
    \lim_{s\to\infty}
    \Ad{z_s}(a),
    \qquad
    a\in\calB\,.
}
Thus \(\beta\) is asymptotically inner through
\(\gamma\)-invariant unitaries.

Finally, define
\eq{
    \beta_t
    :=
    \begin{cases}
        \Ad{z_{t/(1-t)}}
        &
        0\leq t<1\\
        \beta
        &
        t=1
    \end{cases}\,.
}
This is a point-norm continuous path in
\(\operatorname{Aut}(\calB)\). It satisfies
\eq{
    \beta_0
    &=
    \mathrm{id},\\
    \beta_t\circ\gamma_g
    &=
    \gamma_g\circ\beta_t
    \qquad
    (g\in G,\ t\in[0,1])\\
    \varphi_1
    &=
    \varphi_0\circ\beta_1\,.
}
\end{proof}
\subsection{Normalization of a \(\sigma\)-local state}

Choose once and for all a right-half Majorana unitary
\(m\in\calU(\calA)\) satisfying
\eql{
    m^*
    &=
    m,&
    m^2
    &=
    \Id,&
    \theta(m)
    &=
    -m,&
    \sigma(m)
    &=
    m
    \label{eq:fixed right Majorana}\,.
}

Define
\eql{
    \tau
    :=
    \sigma\circ\Ad{m}
    \label{eq:twisted locality involution}\,.
}
By \cref{eq:fixed right Majorana},
\eq{
    \tau^2
    &=
    \Id,&
    \tau\circ\theta
    &=
    \theta\circ\tau\,.
}

\begin{lem}[Normalization in the two index sectors]
\label{lem:normalization in the two index sectors}
Let \(\omega\) be a pure, \(\theta\)-invariant, \(\sigma\)-local state,
and set
\eq{
    \varepsilon
    :=
    \calN_\sigma(\omega)\,.
}
Then there exists a unitary
\eq{
    w\in\calU(\calA^\theta)
}
such that the pure state
\eq{
    \rho
    :=
    \omega\circ\Ad{w}
}
has the following invariance property:
\begin{enumerate}
    \item if \(\varepsilon=+1\), then
    \eq{
        \rho\circ\theta
        &=
        \rho,&
        \rho\circ\sigma
        &=
        \rho\,.
    }
    \item if \(\varepsilon=-1\), then
    \eq{
        \rho\circ\theta
        &=
        \rho,&
        \rho\circ\tau
        &=
        \rho\,.
    }
\end{enumerate}
\end{lem}

\begin{proof}
Let \((\pi,\calH,\Omega)\) be the GNS representation of \(\omega\).
Let \(\Theta\) be the canonical parity implementer:
\eq{
    \Theta\pi(a)\Theta
    &=
    \pi(\theta(a)),&
    \Theta\Omega
    &=
    \Omega\,.
}
Since \(\omega\) is \(\sigma\)-local, \(\pi\) and
\(\pi\circ\sigma\) are unitarily equivalent. Hence there exists a
unitary \(S\) such that
\eq{
    S\pi(a)S^*
    =
    \pi(\sigma(a))\,.
}
Since \(\sigma^2=\Id\) and \(\pi\) is irreducible, \(S^2\) is scalar.
Multiplying \(S\) by a phase, we may assume that
\eq{
    S^*
    &=
    S,&
    S^2
    &=
    \Id\,.
}

If \(u\in\calU(\calA)\) satisfies
\eq{
    \omega\circ\sigma
    =
    \omega\circ\Ad{u},
}
then \(S\Omega\) and \(\pi(u)\Omega\) implement the same pure state.
They therefore agree up to a phase. By the definition of the index,
\eq{
    \Theta\pi(u)\Omega
    =
    \varepsilon\pi(u)\Omega\,.
}
It follows that
\eql{
    \Theta S
    =
    \varepsilon S\Theta
    \label{eq:Theta S commutation sign}\,.
}

Suppose first that \(\varepsilon=+1\). Then \(\Theta\) and \(S\)
commute. Choose \(s\in\Set{\pm1}\) such that
\eq{
    \eta
    :=
    \frac{(\Id+sS)\Omega}
    {\norm{(\Id+sS)\Omega}}
}
is well-defined. Then
\eq{
    \Theta\eta
    &=
    \eta,&
    S\eta
    &=
    s\eta\,.
}
The vector state represented by \(\eta\) is therefore invariant under
both \(\theta\) and \(\sigma\).

The representation of \(\calA^\theta\) on the \(+1\)-eigenspace of
\(\Theta\) is irreducible. Indeed,
\eq{
    \pi(\calA^\theta)''
    =
    \Set{\Theta}'\, ,
}
and its compression to the \(+1\)-eigenspace is the full algebra of
bounded operators on that eigenspace. Kadison transitivity therefore
gives
\eq{
    w\in\calU(\calA^\theta)
}
such that
\eq{
    \pi(w)\Omega=\eta\,.
}
Thus \(\rho=\omega\circ\Ad{w}\) is invariant under both
\(\theta\) and \(\sigma\).

Suppose now that \(\varepsilon=-1\). Set
\eq{
    M
    &:=
    \pi(m),&
    R
    &:=
    SM\,.
}
Since \(\sigma(m)=m\),
\eq{
    SM=MS\,.
}
Hence \(R\) is a self-adjoint unitary, and
\eq{
    R\pi(a)R
    =
    \pi(\tau(a))\,.
}
By \cref{eq:Theta S commutation sign,eq:fixed right Majorana},
both \(S\) and \(M\) anticommute with \(\Theta\). Consequently,
\eq{
    R\Theta=\Theta R\,.
}
Choose \(r\in\Set{\pm1}\) such that
\eq{
    \eta
    :=
    \frac{(\Id+rR)\Omega}
    {\norm{(\Id+rR)\Omega}}
}
is well-defined. Then
\eq{
    \Theta\eta
    &=
    \eta,&
    R\eta
    &=
    r\eta\,.
}
The vector state represented by \(\eta\) is invariant under
\(\theta\) and \(\tau\). As in the first case, Kadison transitivity for
\(\calA^\theta\) gives \(w\in\calU(\calA^\theta)\) satisfying
\eq{
    \pi(w)\Omega=\eta\,.
}
Thus \(\rho=\omega\circ\Ad{w}\) is invariant under
\(\theta\) and \(\tau\)
\end{proof}

\subsection{Classification of automorphic-path-components}

\begin{thm}[automorphic-path classification]
\label{thm:automorphism path classification}
Let \(\omega_1,\omega_2\) be pure, \(\theta\)-invariant,
\(\sigma\)-local states on the one-dimensional self-dual CAR algebra.
If
\eql{
    \calN_\sigma(\omega_1)
    =
    \calN_\sigma(\omega_2),
    \label{eq:equal sigma indices}
}
then \(\omega_1\) and \(\omega_2\) are
automorphic-path-connected in the sense of
\cref{def:automorphism path connectedness}.
\end{thm}

\begin{proof}
Set
\eq{
    \varepsilon
    :=
    \calN_\sigma(\omega_1)
    =
    \calN_\sigma(\omega_2)\,.
}
By \cref{lem:normalization in the two index sectors}, there exist
\eq{
    w_1,w_2\in\calU(\calA^\theta)
}
such that
\eq{
    \rho_i
    :=
    \omega_i\circ\Ad{w_i}
    \qquad
    (i=1,2)
}
have the following common invariance property:
\begin{enumerate}
    \item if \(\varepsilon=+1\), both states are invariant under the
    action generated by \(\theta\) and \(\sigma\);
    \item if \(\varepsilon=-1\), both states are invariant under the
    action generated by \(\theta\) and
    \(\tau=\sigma\circ\Ad{m}\).
\end{enumerate}

In the first case, define the action
\eq{
    \gamma^+:
    \ZZ_2^2
    \longrightarrow
    \operatorname{Aut}(\calA)
}
by
\eq{
    \gamma^+_{(1,0)}
    &:=
    \theta,&
    \gamma^+_{(0,1)}
    &:=
    \sigma\,.
}
Its three nontrivial automorphisms are
\eq{
    \theta,\qquad
    \sigma,\qquad
    \theta\circ\sigma\,.
}
These are respectively parity on the full chain, the left half-chain,
and the right half-chain, and are outer.

In the second case, define
\eq{
    \gamma^-:
    \ZZ_2^2
    \longrightarrow
    \operatorname{Aut}(\calA)
}
by
\eq{
    \gamma^-_{(1,0)}
    &:=
    \theta,&
    \gamma^-_{(0,1)}
    &:=
    \tau\,.
}
Since \(\tau\) is an inner perturbation of \(\sigma\), and
\(\theta\circ\tau\) is an inner perturbation of
\(\theta\circ\sigma\), this action is also pointwise outer

Apply \cref{lem:equivariant homogeneity} to the appropriate action.
There exists a point-norm continuous path
\eq{
    [0,1]\ni t
    \longmapsto
    \beta_t\in\operatorname{Aut}(\calA)
}
such that
\eq{
    \beta_0
    &=
    \Id,&
    \rho_2
    &=
    \rho_1\circ\beta_1
}
and every \(\beta_t\) commutes with the corresponding
\(\ZZ_2^2\)-action. In particular,
\eq{
    \beta_t\circ\theta
    =
    \theta\circ\beta_t\,.
}

If \(\varepsilon=+1\), then \(\beta_t\) commutes with \(\sigma\), so
\eq{
    \partial_\sigma(\beta_t)
    =
    \Id
    =
    \Ad{\Id}\,.
}

Suppose that \(\varepsilon=-1\). Since \(\beta_t\) commutes with
\(\tau=\sigma\circ\Ad{m}\),
\eq{
    \beta_t\circ\sigma\circ\Ad{m}\circ\beta_t^{-1}
    =
    \sigma\circ\Ad{m}\,.
}
Using
\eq{
    \beta_t\circ\Ad{m}\circ\beta_t^{-1}
    =
    \Ad{\beta_t(m)}\, ,
}
we obtain
\eq{
    \partial_\sigma(\beta_t)
    =
    \Ad{q_t}\, ,
}
where
\eql{
    q_t
    :=
    \sigma\left(\beta_t(m)\right)m
    \label{eq:beta locality cocycle}\,.
}
The map \(t\mapsto q_t\) is norm-continuous and
\eq{
    q_0
    =
    \sigma(m)m
    =
    \Id\,.
}
Moreover,
\eq{
    \theta(q_t)
    &=
    \theta\left(
        \sigma(\beta_t(m))m
    \right)\\
    &=
    \sigma\left(
        \beta_t(\theta(m))
    \right)\theta(m)\\
    &=
    \sigma\left(
        \beta_t(-m)
    \right)(-m)\\
    &=
    q_t\,.
}
Thus \(t\mapsto(\beta_t,q_t)\) is a symmetric local automorphism path
in the nontrivial index sector as well.

It remains to undo the normalizations. From
\eq{
    \rho_2
    =
    \rho_1\circ\beta_1
}
we obtain
\eq{
    \omega_2\circ\Ad{w_2}
    =
    \omega_1\circ\Ad{w_1}\circ\beta_1
}
and hence
\eql{
    \omega_2
    =
    \omega_1\circ\alpha,
    \qquad
    \alpha
    :=
    \Ad{w_1}\circ\beta_1\circ\Ad{w_2^*}
    \label{eq:final automorphism}\,.
}

The even CAR algebra \(\calA^\theta\) is an AF algebra and its unitary
group is path-connected. Choose norm-continuous paths
\eq{
    [0,1]\ni t
    &\longmapsto
    w_{i,t}\in\calU(\calA^\theta),&
    w_{i,0}
    &=
    \Id,&
    w_{i,1}
    &=
    w_i.
}
For every such path,
\eq{
    \partial_\sigma\left(\Ad{w_{i,t}}\right)
    =
    \Ad{c_{i,t}}\, ,
}
where
\eql{
    c_{i,t}
    :=
    \sigma(w_{i,t})^*w_{i,t}
    \label{eq:inner locality cocycle}\,.
}
The maps \(t\mapsto c_{i,t}\) are norm-continuous,
\(c_{i,0}=\Id\), and
\eq{
    \theta(c_{i,t})
    =
    c_{i,t}\,.
}

Concatenate the symmetric local automorphism path
\eq{
    t\longmapsto\Ad{w_{1,t}}\, ,
}
the path obtained from \(t\mapsto\beta_t\) by left composition with
\(\Ad{w_1}\), and the path obtained from
\(t\mapsto\Ad{w_{2,t}^*}\) by left composition with
\(\Ad{w_1}\circ\beta_1\). By
\cref{lem:composition of symmetric local automorphisms}, the
corresponding locality implementers are obtained by multiplying the
implementers from
\cref{eq:beta locality cocycle,eq:inner locality cocycle}. They are
norm-continuous and \(\theta\)-even.

The resulting path is a symmetric local automorphism path from
\(\Id\) to the automorphism \(\alpha\) in
\cref{eq:final automorphism}. Since
\eq{
    \omega_2
    =
    \omega_1\circ\alpha\, ,
}
the states are automorphic-path-connected.
\end{proof}

Combining \cref{eq:index invariance along automorphism paths} with
\cref{thm:automorphism path classification} gives the classification
of the automorphic-path-components.

\begin{cor}
\label{cor:index classifies automorphism path components}
Let \(\omega_1,\omega_2\) be pure, \(\theta\)-invariant,
\(\sigma\)-local states. Then
\eql{
    \omega_1
    \text{ and }
    \omega_2
    \text{ are }
    \text{automorphic-path-connected}
    \quad\Longleftrightarrow\quad
    \calN_\sigma(\omega_1)
    =
    \calN_\sigma(\omega_2)\,.
}
\end{cor}

\begin{example}[A Majorana flip inside an automorphic-path-component]
\label{ex:majorana flip within automorphism path component}
Let \(\omega_0\) be the \(k=0\) Kitaev state, and let \(m\) be the
right-half Majorana unitary from
\cref{eq:fixed right Majorana}, chosen with support away from the cut.
Set
\eq{
    \omega_1
    :=
    \omega_0\circ\Ad{m}\,.
}
Since \(\Ad{m}\) commutes with both \(\theta\) and \(\sigma\), the
state \(\omega_1\) is pure, \(\theta\)-invariant and
\(\sigma\)-invariant. Hence
\eq{
    \calN_\sigma(\omega_0)
    =
    \calN_\sigma(\omega_1)
    =
    +1\,.
}
On the other hand,
\eq{
    \calN(\omega_0,\omega_1)
    &=
    \omega_0\br{m^*\theta(m)}\\
    &=
    -1\,.
}
Nevertheless, by
\cref{cor:index classifies automorphism path components},
\(\omega_0\) and \(\omega_1\) are
automorphic-path-connected.
\end{example}

\appendix
\crefalias{section}{appendix}

\section{What about the other $\ZZ_2$ indices?}\label{sec:time reversal pair index}
It is natural to wonder whether the $\ZZ_2$ index which recently appeared in \cite{BachmannBolsRahnama2024ManyBodyFKM} can be formulated in terms of our abstract $\ZZ_2$ index. The answer is of course no: the two appear in different symmetry contexts. The situation is somewhat analogous to the two different (but topologically equivalent) indices appearing in \cite[Appendix A]{ChungShapiro2023}. To phrase the Fu-Kane-Mele index one needs on top of the parity symmetry $\theta$ also Fermionic time-reversal. We next explain how that would work at the present level of abstraction. In terms of physics, the index discussed in \cref{sec:Majorana index} corresponds to one-dimensional superconductors of class D whereas the one in the present section would fit one-dimensional superconductors in class DIII or two-dimensional class AII insulators.

\subsection{The $\ZZ_2$-index of a time-reversal pair}
\label{sec:abstract time reversal pair index}

Let $\calA$ be a unital C*-algebra, let $\theta:\calA\to\calA$ be an
involutive automorphism, and let $\tau:\calA\to\calA$ be a
conjugate-linear $*$-automorphism satisfying
\eq{
    \tau^2=\theta\,.
}
In particular, $\tau$ commutes with $\theta$. For a state $\omega$ on
$\calA$, define the state $\omega^\tau$ by
\eql{
    \omega^\tau(a)
    :=
    \overline{\omega(\tau(a))},
    \qquad a\in\calA
    \label{eq:time reversal action on states}\,.
}
Then $(\omega^\tau)^\tau=\omega\circ\theta$.

\begin{defn}[Time-reversal pair]
    Let $\omega_1,\omega_2$ be pure, $\theta$-invariant states on
    $\calA$. We call $(\omega_1,\omega_2)$ an admissible
    \emph{time-reversal pair} iff
    \eq{
        \omega_2=\omega_1^\tau
    }
    and there is a $\theta$-even unitary $u\in\calU(\calA)$ such that
    \eq{
        \omega_2=\omega_1\circ\Ad{u}\,.
    }
\end{defn}

The condition is symmetric in the two states, since
$(\omega_1^\tau)^\tau=\omega_1$. The evenness of $u$ is essential: the
ordinary parity index of the pair is then trivial, while the following
time-reversal index need not be.

\begin{defn}[Time-reversal pair index]
    For an admissible time-reversal pair, define
    \eql{
        \calN_{\theta,\tau}(\omega_1,\omega_2)
        :=
        \omega_1\br{\tau(u)u}
        \label{eq:time reversal pair index}\,.
    }
\end{defn}

\begin{thm}
\label{thm:properties of time reversal pair index}
The index in \cref{eq:time reversal pair index} has the following
properties.
\begin{enumerate}
    \item It is well-defined and
    $\calN_{\theta,\tau}(\omega_1,\omega_2)\in\Set{\pm1}$.

    \item It is symmetric:
    \eq{
        \calN_{\theta,\tau}(\omega_1,\omega_2)
        =
        \calN_{\theta,\tau}(\omega_2,\omega_1)\,.
    }

    \item If $\norm{\omega_1-\omega_2}<2$, then
    $\calN_{\theta,\tau}(\omega_1,\omega_2)=+1$

    \item If $\alpha$ is an automorphism commuting with both $\theta$
    and $\tau$, then
    \eq{
        \calN_{\theta,\tau}(\omega_1\circ\alpha,
        \omega_2\circ\alpha)
        =
        \calN_{\theta,\tau}(\omega_1,\omega_2)\,.
    }
    The index is multiplicative under graded stacking.
\end{enumerate}
\end{thm}

\begin{proof}
Let $(\calH_1,\pi_1,\Omega_1)$ be the GNS representation of
$\omega_1$ and set $\Omega_2:=\pi_1(u)\Omega_1$. Define on the dense
subspace $\pi_1(\calA)\Omega_1$ the anti-linear map
\eql{
    T\pi_1(a)\Omega_1
    :=
    \pi_1(\tau(a))\Omega_2,
    \qquad a\in\calA
    \label{eq:GNS time reversal operator}\,.
}
Using $\omega_2=\omega_1^\tau$ and the $\theta$-invariance of
$\omega_1$, one checks that $T$ extends to an anti-unitary on
$\calH_1$ and
\eq{
    T\pi_1(a)T^*=\pi_1(\tau(a))\,.
}
Let $\Gamma$ be the self-adjoint unitary implementing $\theta$ in this
GNS representation, normalized by $\Gamma\Omega_1=\Omega_1$. Since
$u$ is even, also $\Gamma\Omega_2=\Omega_2$, and consequently
$T\Gamma=\Gamma T$.

Since $\tau^{-1}(u)=\tau(u)$ for even $u$, applying $\tau$ to the
state relation gives
\eq{
    \omega_1
    =
    \omega_1\circ\Ad{\tau(u)u}\,.
}
Purity therefore implies that $\pi_1(\tau(u)u)\Omega_1$ is a scalar
multiple of $\Omega_1$. Moreover,
\eq{
    T^2\Omega_1
    =
    \pi_1(\tau(u)u)\Omega_1\,.
}
Both $T^2$ and $\Gamma$ implement $\theta$. Since $\pi_1$ is
irreducible, there is $\kappa\in\mathbb{S}^1$ such that
\eql{
    T^2=\kappa\Gamma,
    \qquad
    \kappa
    =
    \omega_1\br{\tau(u)u}
    =
    \calN_{\theta,\tau}(\omega_1,\omega_2)
    \label{eq:time reversal square and pair index}\,.
}
Conjugating $T^2=\kappa\Gamma$ by the anti-unitary $T$ and using
$T\Gamma=\Gamma T$ gives $\kappa=\overline{\kappa}$. Hence
$\kappa\in\Set{\pm1}$.

If $v$ is another even implementing unitary, then
$\pi_1(v)\Omega_1=\lambda\pi_1(u)\Omega_1$ for some
$\lambda\in\mathbb{S}^1$. The anti-unitary obtained from $v$ is
$\lambda T$, whose square equals $T^2$. This proves well-definedness.
For the reversed pair the corresponding anti-unitary is $\kappa T$,
and its square is again $T^2$, proving symmetry.

If $\kappa=-1$, then Kramers orthogonality gives
$\ip{\Omega_1}{\Omega_2}=\ip{\Omega_1}{T\Omega_1}=0$. The pure-state
distance formula then gives $\norm{\omega_1-\omega_2}=2$; see
\cite[Lemma~2.4]{PowersStormer1970FreeStates}. This proves the third
claim. Covariance follows by replacing $u$ with $\alpha^{-1}(u)$.
Multiplicativity follows from the implementing unitary
$u\hat\otimes\widetilde{u}$ and
\cref{eq:time reversal pair index}.
\end{proof}

\subsubsection{Relation to the Katsura--Koma $\ZZ_2$ index}
\label{sec:time reversal pair projection correspondence}

Let $\calH$ be a separable Hilbert space, let $J$ be an anti-unitary on
$\calH$ satisfying $J^2=-\Id$, and let $P,Q$ be orthogonal projections
such that
\eq{
    Q=JPJ^*\,.
}
If $P-Q$ is compact, the Katsura--Koma index of the pair is
\eql{
    \findex_2(P,Q)
    :=
    \dim\ker\br{P-Q-\Id}\Mod{2}
    =
    \dim\br{\im P\cap\ker Q}\Mod{2}
    \label{eq:Katsura Koma pair index}\,.
}
See \cite{KatsuraKoma2016}.

Let $\calA:=\operatorname{CAR}(\calH)$, let $\theta$ be Fermionic
parity, and let $\tau_J$ be the conjugate-linear $*$-automorphism
determined by
\eql{
    \tau_J\br{a^*(f)}
    :=
    a^*(Jf)
    \qquad f\in\calH
    \label{eq:antiunitary induced CAR automorphism}\,.
}
Then $\tau_J^2=\theta$. If $\omega_P,\omega_Q$ are the pure
gauge-invariant quasi-free states associated with $P,Q$, then
$\omega_Q=\omega_P^{\tau_J}$.

\begin{thm}[Katsura--Koma correspondence]
\label{thm:Katsura Koma pure state correspondence}
Assume that $P-Q\in\calJ_2(\calH)$. Then
$(\omega_P,\omega_Q)$ is an admissible time-reversal pair and
\eql{
    \calN_{\theta,\tau_J}(\omega_P,\omega_Q)
    =
    \br{-1}^{\findex_2(P,Q)}
    =
    \br{-1}^{\dim\im P\cap\ker Q}
    \label{eq:Katsura Koma pure state equality}\,.
}
\end{thm}

\begin{proof}
Set
\eql{
    R
    :=
    \chi_{(0,\infty)}\br{P+Q-\Id}
    \label{eq:Katsura Koma reference projection}\,.
}
Since $J(P+Q-\Id)J^*=P+Q-\Id$, we have $JRJ^*=R$. We claim that
\eql{
    P-R,Q-R\in\calJ_2(\calH),
    \qquad
    \operatorname{index}(P,R)
    =
    \dim\im P\cap\ker Q
    \label{eq:reference projection properties}\,.
}
To see this, use the Halmos decomposition of the pair $P,Q$. On its
generic part one has, for commuting positive operators $C,S$ with
$C^2+S^2=\Id$,
\eq{
    P_{\rm g}
    &=
    \begin{pmatrix}\Id&0\\0&0\end{pmatrix},
    &
    Q_{\rm g}
    &=
    \begin{pmatrix}C^2&CS\\CS&S^2\end{pmatrix},
    &
    R_{\rm g}
    &=
    \frac12
    \begin{pmatrix}\Id+C&S\\S&\Id-C\end{pmatrix}\,.
}
The assumption $P-Q\in\calJ_2$ implies $S\in\calJ_2$, and hence
$P_{\rm g}-R_{\rm g}\in\calJ_2$. On the four defect subspaces, $R$
agrees with $P$ except that it vanishes on $\im P\cap\ker Q$. This
proves \cref{eq:reference projection properties}; the claim for $Q$
then follows from $Q-P\in\calJ_2$.

Let $(\calH_R,\pi_R,\Omega_R)$ be the GNS representation of
$\omega_R$, and let $\Gamma_R$ implement $\theta$, with
$\Gamma_R\Omega_R=\Omega_R$. Since $R$ is $J$-invariant, the formula
\eq{
    T_R\pi_R(a)\Omega_R
    :=
    \pi_R(\tau_J(a))\Omega_R
}
defines an anti-unitary satisfying
\eql{
    T_R^2=\Gamma_R
    \label{eq:reference time reversal square}\,.
}

By $P-R\in\calJ_2$, the Shale--Stinespring criterion gives a vector
$\Omega_P\in\calH_R$ implementing $\omega_P$. It may be chosen with
parity
\eql{
    \Gamma_R\Omega_P
    =
    \br{-1}^{\operatorname{index}(P,R)}\Omega_P
    \label{eq:relative quasi free parity}\,.
}
For completeness, this is the usual relative-parity formula: on the
defect subspaces a representative of $\omega_P$ is obtained from
$\Omega_R$ by
$\dim(\im P\cap\ker R)+\dim(\ker P\cap\im R)$ odd generators, while
the generic Halmos part is implemented by even quadratic rotations.
The parity is therefore
$(-1)^{\operatorname{index}(P,R)}$

The vector $T_R\Omega_P$ implements $\omega_Q$. Moreover, it has the
same parity as $\Omega_P$, since $T_R$ commutes with $\Gamma_R$. The
restriction of $\pi_R(\calA^\theta)$ to either parity sector is
irreducible. Kadison transitivity therefore gives
$u\in\calU(\calA^\theta)$ such that, after a phase choice,
\eq{
    \pi_R(u)\Omega_P=T_R\Omega_P\,.
}
Thus $(\omega_P,\omega_Q)$ is an admissible time-reversal pair. By
\cref{eq:reference time reversal square,eq:relative quasi free parity},
the GNS characterization in
\cref{eq:time reversal square and pair index}, and
\cref{eq:reference projection properties},
\eq{
    \calN_{\theta,\tau_J}(\omega_P,\omega_Q)
    =
    \br{-1}^{\operatorname{index}(P,R)}
    =
    \br{-1}^{\dim\im P\cap\ker Q}\,.
}
This proves \cref{eq:Katsura Koma pure state equality}.
\end{proof}

\subsection{The $\ZZ_2$ index of 1D class DIII superconductors}
\label{sec:DIII index}

The preceding construction also gives the one-dimensional class DIII
$\ZZ_2$ index. Return to the self-dual CAR algebra and the locality
automorphism $\sigma$ of \cref{sec:Majorana index}. Let $\tau$ be
Fermionic time reversal, so that
\eq{
    \tau^2
    =
    \theta,
    \qquad
    \tau\circ\sigma
    =
    \sigma\circ\tau\,.
}
Let $\omega$ be a $\tau$-invariant, $\sigma$-local pure state. It is
automatically $\theta$-invariant.

We first note that time reversal forces
\eq{
    \calN_\sigma(\omega)
    =
    +1\,.
}
Indeed, in the GNS representation of $\omega$, let $\Gamma$ implement
$\theta$, let $T$ be the anti-unitary implementing $\tau$, normalized
so that $T\Omega=\Omega$, and let $S$ be a self-adjoint unitary
implementing $\sigma$. Thus
\eq{
    T^2
    =
    \Gamma,
    \qquad
    S^2
    =
    \Id\,.
}
Since $\tau$ and $\sigma$ commute and the representation is
irreducible, there is $\kappa\in\Set{\pm1}$ such that
\eq{
    TST^*
    =
    \kappa S\,.
}
Conjugating this relation once more by $T$ gives
\eq{
    \Gamma S\Gamma
    =
    S\,.
}
The commutation sign of $\Gamma$ and $S$ is
$\calN_\sigma(\omega)$, as in
\cref{eq:Theta S commutation sign}, proving the claim. Consequently,
the locality relation may be implemented by an even unitary
$u\in\calU(\calA^\theta)$:
\eq{
    \omega\circ\sigma
    =
    \omega\circ\Ad{u}\,.
}

Now $\tau\circ\sigma$ is a conjugate-linear $*$-automorphism satisfying
\eq{
    \br{\tau\circ\sigma}^2
    =
    \theta,
    \qquad
    \omega^{\tau\circ\sigma}
    =
    \omega\circ\sigma\,.
}
Hence $(\omega,\omega\circ\sigma)$ is an admissible time-reversal pair
for $\tau\circ\sigma$, and we define
\eql{
    \calN_{\rm DIII}(\omega)
    :=
    \calN_{\theta,\tau\circ\sigma}
    \br{\omega,\omega\circ\sigma}
    =
    \omega\br{
        \br{\tau\circ\sigma}(u)u
    }
    \in
    \Set{\pm1}
    \label{eq:class DIII index}\,.
}
By \cref{eq:time reversal square and pair index},
\eq{
    \calN_{\rm DIII}(\omega)
    =
    \kappa\,.
}
Thus the nontrivial class DIII phase is characterized by
$TST^*=-S$. Well-definedness and multiplicativity under stacking
follow from
\cref{thm:properties of time reversal pair index}. The same abstract
pair index has a distinct two-dimensional application, to which we
now turn.

\subsection{The Fu--Kane--Mele index of 2D class AII interacting insulators}
\label{sec:Fu Kane Mele pair index}

We now take
\eql{
    \calH
    :=
    \ell^2(\ZZ^2)\otimes\CC^{2N},
    \qquad
    \calA
    :=
    \operatorname{CAR}(\calH)
    \label{eq:FKM CAR algebra}\,.
}
Let $\Theta$ be the odd time-reversal anti-unitary on $\calH$, so that
$\Theta^2=-\Id$. Its action on the CAR algebra is the conjugate-linear
$*$-automorphism $\tau$ defined by
\eql{
    \tau\br{a^*(f)}
    :=
    a^*(\Theta f),
    \qquad f\in\calH
    \label{eq:physical time reversal CAR automorphism}\,.
}
Thus $\Theta$ is an operator on the one-particle Hilbert space, whereas
$\tau$ is the induced time-reversal automorphism of the many-body
algebra. In particular, $\tau^2=\theta$ on $\calA$, while $\tau^2$ is
the identity on the even observable algebra. If $\rho_t$ denotes the
$U(1)$ gauge action, then
\eq{
    \tau\circ\rho_t
    =
    \rho_{-t}\circ\tau\,.
}

\begin{rem}[Fu--Kane--Mele flux pair]
Let $\omega$ be a charge- and time-reversal-invariant SRE pure state
with a symmetric parent Hamiltonian. The flux-insertion construction
of \cite{BachmannShapiroTauber2026} may be run along the two branches
$[0,\pi]$ and $[0,-\pi]$, using half-plane gauge transformations,
quasi-adiabatic flow, and restriction of the generator to a half-line.
Let $\gamma_{\pi}$ and $\gamma_{-\pi}$ denote the two endpoint
automorphisms and set
\eql{
    \omega_+
    :=
    \omega\circ\gamma_{\pi},
    \qquad
    \omega_-
    :=
    \omega\circ\gamma_{-\pi}
    \label{eq:plus minus pi flux states}\,.
}
Time-reversal covariance exchanges the two states. At the endpoints
their fluxes differ by $2\pi$; the half-line localization and the SRE
disentangler give an even almost-local unitary $u\in\calU(\calA)$ such
that
\eql{
    \omega_- = \omega_+^\tau,
    \qquad
    \omega_- = \omega_+\circ\Ad{u},
    \qquad
    \theta(u)=u
    \label{eq:FKM admissible time reversal pair}\,.
}
These are precisely the conclusions established in the fluxon
construction of \cite{BachmannBolsRahnama2024ManyBodyFKM}. Hence its
many-body Fu--Kane--Mele index, in sign convention, is
\eql{
    \calN_{\rm FKM}(\omega)
    :=
    \calN_{\theta,\tau}(\omega_+,\omega_-)
    \label{eq:many body FKM as pair index}\,.
}
Indeed, by \cref{eq:time reversal square and pair index}, the value is
$-1$ precisely when the $\pi$-flux state belongs to a Kramers pair.
Independence of the flux-insertion choices, multiplicativity under
stacking, and invariance under symmetry-preserving locally generated
automorphisms are then the corresponding results of
\cite{BachmannBolsRahnama2024ManyBodyFKM}.
\end{rem}

\section{A child's garden of locality}
\label{sec:a childs garden of locality}
One may ask what about other notions of locality in the many-body setting? The most interesting proposal for us is Matsui's split property \cite{Matsui2013Split}. To define it we remark that whenever $P=P^2=P^\ast$ is a projection on $\calH$ commuting with $\Xi$, then there is the associated self-dual CAR algebra $\calA_P$ which may be considered a sub-algebra of $\calA$, in the sense that \eq{
\calA \cong \calA_P \hat{\otimes}\calA_{P^\perp}
} and there is the natural injection $\iota_P:\calA_P\hookrightarrow\calA$, which induces a truncated state $\left.\omega\right|_{\calA_P}\equiv\omega\circ\iota_{\calA_P}$. We say that an element $a\in\calA$ has \emph{finite rank support} iff there exists some finite rank $P$ commuting with $\Xi$ such that $a\in\calA_P$. Moreover, $a\in\calA$ has \emph{finite spatial support} iff there exists some \emph{finite} $\Lambda\subseteq\ZZ$ such that $a\in\calA_{\Lambda}$ where we abuse the notation $\calA_\Lambda\equiv\calA_{\chi_{\Lambda}(X)\otimes\Id_{\CC^{2N}}}$ and similarly $\omega_\Lambda\equiv\omega\circ\iota_{\calA_{\Lambda}}$ here and in the sequel.

In particular for $\NN\subset\ZZ$, we may consider the decomposition \eql{\label{eq:left right decomposition for the algebra}
\calA = \calA_L \hat{\otimes}\calA_R\,.
} In general, of course for arbitrary $\omega\circ\theta=\omega$,
\eq{
\omega\neq\omega_L\hat{\otimes}\omega_R\,.
} We note however that normalization and positivity are preserved by such restrictions, so at least the RHS is indeed a state. However, purity is \emph{not} preserved.
\begin{defn}[Matsui's split states]
    A state $\omega$ which is $\theta$-invariant is called \emph{split} iff $\omega$ is GNS quasi-equivalent to $\omega_L\hat{\otimes}\omega_R$. That is, iff there is a normal *-isomorphism $\Phi$ of von Neumann algebras of the double-commutants of the respective GNS representations $\pi,\pi_{\rm L\hat{\otimes} R}$of the states $\omega,\omega_L\hat{\otimes}\omega_R$ such that \eq{
    \Phi(\pi(a)) = \pi_{L\hat{\otimes} R}(a)\qquad(a\in\calA)\,.
    }
\end{defn}
\begin{claim}
    If $\omega$ is a pure parity-invariant state on $\calA$ which is split in the sense of Matsui then it is $\sigma$-local. The converse need not hold. 
\end{claim}
See \cite[Remark 3.4]{Matsui2020SplitFermionicString}, which this claim settles.

\begin{proof}
    Under the decomposition \cref{eq:left right decomposition for the algebra}, we may write 
    \eq{
    \sigma = \theta_L \hat{\otimes}\mathrm{id}_R\,.
    } Now, since $\omega_L\circ\theta_L = \omega\circ\theta\circ\iota_L$ where $\iota_L:\calA_L\to\calA$ is the inclusion map, then using parity invariance we find that \eq{
    \omega_L \circ \theta_L = \omega_L\,.
    }

    This implies that
    \eq{
    \omega_L\hat{\otimes}\omega_R\circ \sigma = \omega_L\hat{\otimes}\omega_R\,.
    } We learn that
 \eq{
    \omega &\sim_{q.e.} \omega_L\hat{\otimes}\omega_R
      &&\text{(by hypothesis of being split)}\\
    &=\quad\,\, \omega_L\hat{\otimes}\omega_R\circ\sigma&&\text{(by earlier calculation)}\\
    &\sim_{q.e.} \omega\circ\sigma&&\text{(automorphism precomposition preserves quasiequivalence relation)}\,.
} However, $\omega,\omega\circ\sigma$ are both pure, so quasi-equivalence implies local-comparability.

For the converse direction, we give an explicit
counterexample. This settles \cite{Matsui2020SplitFermionicString} Remark 3.4 We work with
\eq{
    \calH=\ell^2(\ZZ)\otimes\CC^2\,.
}
Let $e_+,e_-$ be the standard basis of $\CC^2$, and let the particle-hole
conjugation be
\eq{
    \Xi(\delta_x\otimes e_+)=\delta_x\otimes e_-,
    \qquad
    \Xi(\delta_x\otimes e_-)=\delta_x\otimes e_+\, ,
}
extended anti-linearly. We write
\eq{
    a_x:=B(\delta_x\otimes e_+),
    \qquad
    a_x^*:=B(\delta_x\otimes e_-)\,.
}
Thus $a_x,a_x^*$ satisfy the usual CAR relations.

For $n\in\NN$, define the two-site blocks
\eq{
    L_n:=\Set{-2n+1,-2n+2},
    \qquad
    R_n:=\Set{2n-1,2n}\,.
}
The sets $L_n$ decompose the left half-chain and the sets $R_n$ decompose the
right half-chain. Let $\calA_n\subset\calA$ be the finite-dimensional local
algebra generated by
\eq{
    a_{-2n+1},\quad a_{-2n+2},\quad a_{2n-1},\quad a_{2n}\,.
}
Then $\calA_n\cong M_{16}(\CC)$

For each $x\in\ZZ$, set
\eq{
    N_x:=a_x^*a_x,
    \qquad
    Q_x:=a_xa_x^*=\Id-N_x\,.
}
Define projections in $\calA_n$ by
\eq{
    E^{(0)}_{L,n}:=Q_{-2n+1}Q_{-2n+2},
    \qquad
    E^{(2)}_{L,n}:=N_{-2n+1}N_{-2n+2}
}
and
\eq{
    E^{(0)}_{R,n}:=Q_{2n-1}Q_{2n},
    \qquad
    E^{(2)}_{R,n}:=N_{2n-1}N_{2n}\,.
}
Now put
\eq{
    e_n:=E^{(0)}_{L,n}E^{(0)}_{R,n},
    \qquad
    f_n:=E^{(2)}_{L,n}E^{(2)}_{R,n}
}
and define
\eq{
    V_n:=a_{-2n+1}^*a_{-2n+2}^*a_{2n-1}^*a_{2n}^*\,.
}
Then $V_n$ is a partial isometry with
\eq{
    V_n^*V_n=e_n,
    \qquad
    V_nV_n^*=f_n\,.
}
Consequently
\eql{
    p_n:=\onehalf\br{e_n+f_n+V_n+V_n^*}
}
is a minimal projection in $\calA_n$

Let $\tr_n$ denote the non-normalized matrix trace on $\calA_n$, normalized so
that minimal projections have trace $1$. Define a pure state $\varphi_n$ on
$\calA_n$ by
\eq{
    \varphi_n(A):=\tr_n(p_nA),
    \qquad
    A\in\calA_n\,.
}
Using the block decomposition
\eq{
    \ZZ=\bigsqcup_{n\in\NN}\br{L_n\cup R_n},
}
we define the infinite tensor product state
\eq{
    \omega:=\bigotimes_{n\in\NN}\varphi_n\,.
}
Since each $\varphi_n$ is pure, the product state $\omega$ is pure.

We next show that $\omega$ is $\sigma$-local. Recall that $\sigma$ acts as
parity on the left half-chain and trivially on the right half-chain. On the
$n$-th block, the projections $e_n$ and $f_n$ are fixed by this left parity
action. Moreover, $V_n$ contains exactly two creation operators supported on
the left half-chain, and is therefore also fixed by left parity. Hence
\eq{
    \sigma(p_n)=p_n\,.
}
Since $\tr_n$ is invariant under automorphisms, it follows that
\eq{
    \varphi_n\circ\sigma=\varphi_n
}
on each block. Therefore
\eq{
    \omega\circ\sigma=\omega\,.
}
Thus $\omega$ is $\sigma$-local, with implementing unitary $u=\Id$.

It remains to show that $\omega$ is not split. Let
\eq{
    \psi
    :=
    \left.\omega\right|_{\calA_L}
    \hat{\otimes}
    \left.\omega\right|_{\calA_R}\,.
}
On the $n$-th block, write
\eq{
    \psi_n
    :=
    \left.\varphi_n\right|_{\calA_{L,n}}
    \hat{\otimes}
    \left.\varphi_n\right|_{\calA_{R,n}}\,.
}
A direct computation gives
\eq{
    \varphi_n(p_n)=1,
    \qquad
    \psi_n(p_n)=\frac{1}{4}\,.
}
Indeed, the left and right restrictions of $\varphi_n$ each give weight
$\onehalf$ to the empty two-site configuration and weight $\onehalf$ to the
fully occupied two-site configuration, while the off-diagonal terms $V_n$ and
$V_n^*$ vanish under the product of the restrictions.

For $N\in\NN$, define
\eq{
    P_N:=p_1p_2\cdots p_N\,.
}
Since the projections $p_n$ belong to mutually disjoint even local algebras,
they commute, and $P_N$ is a projection. Moreover,
\eq{
    \omega(P_N)=1
}
for every $N$, whereas
\eq{
    \psi(P_N)
    =
    \prod_{n=1}^N\psi_n(p_n)
    =
    \br{\frac{1}{4}}^N
    \xrightarrow[N\to\infty]{}
    0\,.
}
Equivalently,
\eq{
    \prod_{n=1}^{\infty}\psi_n(p_n)
    =
    \prod_{n=1}^{\infty}\frac{1}{4}
    =
    0\,.
}
By the standard infinite tensor product criterion for product states, $\omega$
and $\psi$ are disjoint. In particular, they are not GNS quasi-equivalent.
Therefore $\omega$ is not split in the sense of Matsui.

We have constructed a pure state $\omega$ satisfying
\eq{
    \omega\circ\sigma=\omega\, ,
}
so that $\omega$ is $\sigma$-local, but $\omega$ is not split. Hence the
converse implication fails.
\end{proof}
We next recall the notion of a locally generated automorphism, and then use it
to define short-range-entangled states. We write $\calA_{\rm loc}$ for the set
of elements of finite spatial support, so that
\eq{
    \calA
    =
    \overline{\calA_{\rm loc}}\,.
}

\begin{defn}[Locally generated automorphisms]
Let
\eq{
    I_r(x)
    :=
    \Set{
        y\in\ZZ:
        |x-y|\leq r
    }.
}
A time-dependent zero-chain is a family
\eq{
    \Phi
    =
    \Set{
        \Phi_t(x):
        t\in[0,1],\
        x\in\ZZ
    }
}
such that:
\begin{enumerate}
    \item every $\Phi_t(x)$ is self-adjoint and parity-invariant,
    \eq{
        \Phi_t(x)^*
        &=
        \Phi_t(x),&
        \theta\br{\Phi_t(x)}
        &=
        \Phi_t(x)\,.
    }
    \item the map
    \eq{
        [0,1]\ni t
        \longmapsto
        \Phi_t(x)
    }
    is norm-continuous for every $x\in\ZZ$;
    \item the zero-chain is uniformly bounded:
    \eq{
        \sup_{t\in[0,1]}
        \sup_{x\in\ZZ}
        \norm{\Phi_t(x)}
        <
        \infty\,.
    }
\end{enumerate}

We say that $\Phi$ is exponentially localized if there are constants
$C,\mu>0$ such that, for every $r\in\NN$, there are self-adjoint,
parity-invariant elements
\eq{
    \Phi_t^{(r)}(x)
    \in
    \calA_{I_r(x)}
}
satisfying
\eq{
    \sup_{t\in[0,1]}
    \sup_{x\in\ZZ}
    \norm{
        \Phi_t(x)-\Phi_t^{(r)}(x)
    }
    \leq
    C\ee^{-\mu r}\,.
}
% More generally, $\Phi$ is fast-decaying if, for every $m\in\NN$, there
% is a constant $C_m<\infty$ such that
% \eq{
%     \sup_{t\in[0,1]}
%     \sup_{x\in\ZZ}
%     \norm{
%         \Phi_t(x)-\Phi_t^{(r)}(x)
%     }
%     \leq
%     C_m\br{1+r}^{-m}\,.
% }

Such a zero-chain defines a time-dependent derivation on local observables
by
\eql{\label{eq:zero-chain derivation}
    \delta_t(a)
    :=
    \ii
    \sum_{x\in\ZZ}
    [\Phi_t(x),a],
    \qquad
    a\in\calA_{\rm loc}.
}
The series in \cref{eq:zero-chain derivation} converges in norm. Indeed,
if $a$ is supported in a finite set $\Lambda$ and $x$ is far from
$\Lambda$, then a sufficiently localized, parity-even approximation to
$\Phi_t(x)$ has support disjoint from $\Lambda$ and hence commutes with
$a$. The remaining commutator is bounded by the localization error, which
is summable in $x$.

The associated locally generated automorphism, or LGA, is the unique
point-norm continuous family of automorphisms
\eq{
    [0,1]\ni t
    \longmapsto
    \alpha_t\in\Automorphisms{\calA}
}
satisfying
\eq{
    \frac{\dif}{\dif t}\alpha_t(a)
    &=
    \alpha_t\br{\delta_t(a)},\\
    \alpha_0(a)
    &=
    a,
    \qquad
    a\in\calA_{\rm loc}.
}
Since every term of the zero-chain is parity-invariant, the derivations
commute with $\theta$, and therefore
\eq{
    \alpha_t\circ\theta
    =
    \theta\circ\alpha_t
    \qquad
    \br{
        t\in[0,1]
    }.
}
\end{defn}

The parity assumption on the zero-chain is essential in the Fermionic
setting. Without it, summands localized far from an odd observable need
not commute with that observable, and the series in
\cref{eq:zero-chain derivation} need not converge.

\begin{defn}[Product states]
Let $\omega$ be a pure, $\theta$-invariant state. We say that $\omega$ is
a \emph{product state} iff, for every subset $\Lambda\subseteq\ZZ$,
\eq{
    \omega
    =
    \left.\omega\right|_{\calA_\Lambda}
    \hat{\otimes}
    \left.\omega\right|_{\calA_{\Lambda^c}}.
}
\end{defn}

\begin{defn}[SRE states]
A pure state $\omega$ is called exponentially short-range-entangled,
abbreviated $\mathrm{SRE}_{\exp}$, if there are a pure,
$\theta$-invariant product state $\omega_{\rm product}$ and an
exponentially localized LGA $t\mapsto\alpha_t$ such that
\eq{
    \omega
    =
    \omega_{\rm product}\circ\alpha_1.
}
% If the LGA is only assumed to be fast-decaying, we call $\omega$ an
% $\mathrm{SRE}_{\mathrm{fast}}$ state.
\end{defn}

Every state appearing in the preceding definition is automatically
parity-invariant. Indeed, $\omega_{\rm product}$ is parity-invariant and
the LGA commutes with $\theta$.

\begin{defn}[Area law]
Let $\omega$ be a pure state. For a finite interval
$\Lambda\Subset\ZZ$, let $\rho_\Lambda$ be the density matrix of
$\left.\omega\right|_{\calA_\Lambda}$ with respect to the
non-normalized matrix trace $\operatorname{Tr}_\Lambda$, so that
\eq{
    \left.\omega\right|_{\calA_\Lambda}(A)
    &=
    \operatorname{Tr}_\Lambda
    \br{
        \rho_\Lambda A
    },\\
    \operatorname{Tr}_\Lambda
    \br{
        \rho_\Lambda
    }
    &=
    1.
}
The entanglement entropy of $\omega$ on $\Lambda$ is
\eq{
    S_\Lambda(\omega)
    :=
    -
    \operatorname{Tr}_\Lambda
    \br{
        \rho_\Lambda\log\rho_\Lambda
    }.
}
We say that $\omega$ satisfies an area law if
\eql{\label{eq:area law definition}
    \sup_{
        \Lambda\Subset\ZZ\
        \mathrm{interval}
    }
    S_\Lambda(\omega)
    <
    \infty.
}
\end{defn}

\begin{claim}
\label{claim:SRE implies area law}
Every pure $\mathrm{SRE}_{\exp}$ state satisfies an area law.
\end{claim}

\begin{proof}
Write
\eq{
    \omega
    =
    \omega_{\rm product}\circ\alpha_1,
}
where $\omega_{\rm product}$ is a pure product state and
$t\mapsto\alpha_t$ is generated by an exponentially localized
zero-chain $\Phi$.

We first decompose the zero-chain into strictly local terms. Let
\eq{
    \mathbb E_{I_r(x)}
    :
    \calA
    \longrightarrow
    \calA_{I_r(x)}
}
denote the trace-preserving conditional expectation. Set
\eq{
    \Psi_t(x,0)
    &:=
    \mathbb E_{I_0(x)}
    \br{
        \Phi_t(x)
    },\\
    \Psi_t(x,r)
    &:=
    \mathbb E_{I_r(x)}
    \br{
        \Phi_t(x)
    }
    -
    \mathbb E_{I_{r-1}(x)}
    \br{
        \Phi_t(x)
    },
    \qquad
    r\geq1.
}
Each $\Psi_t(x,r)$ is self-adjoint, parity-invariant, and supported in
$I_r(x)$. Exponential localization and the contractivity of the
conditional expectations imply, after changing the constants, that
\eq{
    \sup_{t\in[0,1]}
    \sup_{x\in\ZZ}
    \norm{
        \Psi_t(x,r)
    }
    \leq
    C_1\ee^{-\mu_1r}
}
for some $C_1,\mu_1>0$. Moreover,
\eql{\label{eq:local decomposition zero chain}
    \Phi_t(x)
    =
    \sum_{r=0}^{\infty}
    \Psi_t(x,r),
}
with norm convergence uniform in $t$ and $x$.

Fix a finite interval $\Lambda$. We first work in finite volume and
consider the entropy of the time-evolved pure state across the
bipartition
\eq{
    \Lambda
    \sqcup
    \Lambda^c.
}
Terms in \cref{eq:local decomposition zero chain} supported entirely in
$\Lambda$ or entirely in $\Lambda^c$ generate unitaries on only one side
of the bipartition and therefore do not change the entanglement entropy.
Only terms whose support intersects both $\Lambda$ and $\Lambda^c$
contribute.

We use the small-incremental-entangling bound
\cite[Theorem~3]{MarienAudenaertVanAcoleyenVerstraete2016AreaLaw}.
For a Hamiltonian term $h$ supported on a finite set meeting both sides
of a bipartition, this bound gives
\eq{
    \left|
        \frac{\dif}{\dif t}S_\Lambda
    \right|
    \leq
    C_{\rm SIE}
    \norm{h}
    \log d_h,
}
where $d_h$ is the smaller of the dimensions of the two interacting
subsystems. The estimate is independent of the dimensions of all
remaining degrees of freedom.

For a term $\Psi_t(x,r)$, the logarithm of the interacting dimension is
bounded by
\eq{
    \log d_h
    \leq
    C_2\br{r+1},
}
where $C_2$ depends only on the on-site dimension. For each fixed $r$,
the number of intervals $I_r(x)$ meeting both $\Lambda$ and
$\Lambda^c$ is at most
\eq{
    C_3\br{r+1},
}
because $\Lambda$ has two boundary points. Therefore
\eq{
    \left|
        \frac{\dif}{\dif t}
        S_\Lambda
        \br{
            \omega_{\rm product}\circ\alpha_t
        }
    \right|
    &\leq
    C_4
    \sum_{r=0}^{\infty}
    \br{r+1}^2
    \sup_{x\in\ZZ}
    \norm{
        \Psi_t(x,r)
    }\\
    &\leq
    C_5
    \sum_{r=0}^{\infty}
    \br{r+1}^2
    \ee^{-\mu_1r}\\
    &\leq
    C_6,
}
where $C_6<\infty$ is independent of $\Lambda$ and of the finite volume.

Since $\omega_{\rm product}$ is a pure product state,
\eq{
    S_\Lambda
    \br{
        \omega_{\rm product}
    }
    =
    0.
}
Integrating the preceding estimate over $t\in[0,1]$ gives
\eq{
    S_\Lambda
    \br{
        \omega_{\rm product}\circ\alpha_1
    }
    \leq
    C_6.
}
The finite-volume evolutions converge locally to $\alpha_t$. Hence, for
each fixed $\Lambda$, the corresponding density matrices converge in
trace norm. Entropy continuity in the finite-dimensional algebra
$\calA_\Lambda$ allows us to pass to the infinite-volume limit without
changing the bound. Consequently,
\eq{
    S_\Lambda(\omega)
    \leq
    C_6
}
for every finite interval $\Lambda$. This proves
\cref{eq:area law definition}.
\end{proof}

\begin{claim}
\label{claim:area law implies split}
Let $\omega$ be a pure, parity-invariant state satisfying an area law.
Then $\omega$ is split.
\end{claim}

\begin{proof}
By the area-law hypothesis, the entropies of the restrictions of
$\omega$ to finite intervals are uniformly bounded. Matsui's
Fermionic entropy criterion
\cite[Theorem~1.8]{Matsui2013Split} then implies that the half-chain von
Neumann algebras in the GNS representation of $\omega$ are of type I.

For a pure, parity-invariant CAR state, the type-I half-chain condition
is equivalent to quasi-equivalence with the graded product of the two
half-chain restrictions
\cite[Theorem~1.3]{Matsui2020SplitFermionicString}. Hence
\eq{
    \omega
    \sim_{\mathrm{q.e.}}
    \left.\omega\right|_{\calA_L}
    \hat{\otimes}
    \left.\omega\right|_{\calA_R}.
}
Thus $\omega$ is split in the sense defined above.
\end{proof}

Combining
\cref{claim:SRE implies area law,claim:area law implies split}
with the split-to-$\sigma$-local implication proved above gives the
following hierarchy. Within the class of pure, parity-invariant states,
\emph{one has}
\eql{\label{eq:hierarchy of locality modes}
\boxed{
\mathrm{SRE}_{\exp}
\subsetneq
\text{area law}
\subsetneq
\text{split}
\subsetneq
\sigma\text{-local}.
}
}
The strictness of the final inclusion is given by the preceding
counterexample. 

We next show that the first two inclusions in
\cref{eq:hierarchy of locality modes} are also strict. Both examples
already occur for one complex Fermionic mode per site. We write
\eq{
    a_x
    :=
    B(\delta_x\otimes e_+),
    \qquad
    a_x^*
    :=
    B(\delta_x\otimes e_-).
}

\begin{example}[An area-law state which is not
\(\mathrm{SRE}_{\exp}\)]
\label{ex:area law not SRE}
For $n\in\NN$, set
\eq{
    x_n
    &:=
    4^n,\\
    y_n
    &:=
    4^n+2^n,\\
    \varepsilon_n
    &:=
    n^{-4}.
}
The pairs
\eq{
    K_n
    :=
    \Set{x_n,y_n}
}
are mutually disjoint, and
\eq{
    |x_n-y_n|
    =
    2^n.
}

On the two-site Fock space over $K_n$, consider the unit vector
\eq{
    \xi_n
    :=
    \sqrt{1-\varepsilon_n}
    \ket{00}
    +
    \sqrt{\varepsilon_n}
    \ket{11},
}
where
\eq{
    \ket{11}
    :=
    a_{x_n}^*a_{y_n}^*\ket{00}.
}
Let $\varphi_n$ be its vector state. Since $\xi_n$ has even Fermionic
parity, $\varphi_n$ is parity-invariant. Put every site which does not
belong to one of the pairs $K_n$ in the vacuum state and define the
graded product state
\eql{\label{eq:sparse pair state}
    \omega
    :=
    \hat{\bigotimes}_{n\in\NN}
    \varphi_n
    \hat{\otimes}
    \omega_{\rm vacuum}^{\rm unpaired}.
}
The state $\omega$ is pure and parity-invariant.

Let $\Lambda\Subset\ZZ$ be a finite interval. A pair $K_n$ contributes
to $S_\Lambda(\omega)$ only when exactly one of its two sites belongs to
$\Lambda$. In that case, the reduced density matrix has eigenvalues
$1-\varepsilon_n$ and $\varepsilon_n$. Consequently,
\eq{
    S_\Lambda(\omega)
    =
    \sum_{
        n:
        |K_n\cap\Lambda|=1
    }
    h(\varepsilon_n),
}
where
\eq{
    h(p)
    :=
    -p\log p
    -
    (1-p)\log(1-p)
}
is the binary entropy. Therefore
\eq{
    S_\Lambda(\omega)
    &\leq
    \sum_{n=1}^{\infty}
    h(n^{-4})\\
    &<
    \infty.
}
The upper bound is independent of $\Lambda$, so $\omega$ satisfies an
area law.

On the other hand, parity invariance gives
\eq{
    \omega(a_{x_n})
    &=
    0,\\
    \omega(a_{y_n})
    &=
    0,
}
whereas a direct two-site calculation gives
\eq{
    \left|
        \omega(a_{x_n}a_{y_n})
    \right|
    =
    \sqrt{
        \varepsilon_n
        \br{
            1-\varepsilon_n
        }
    }
    \asymp
    n^{-2}.
}
Thus
\eq{
    \left|
        \omega(a_{x_n}a_{y_n})
        -
        \omega(a_{x_n})
        \omega(a_{y_n})
    \right|
    \asymp
    n^{-2},
}
even though
\eq{
    \dist
    \br{
        \Set{x_n},
        \Set{y_n}
    }
    =
    2^n.
}
There cannot exist constants $C,\mu>0$ such that
\eq{
    n^{-2}
    \leq
    C\ee^{-\mu2^n}
}
for every $n$. Hence $\omega$ is not exponentially clustering, even for
single-site observables. By
\cref{claim:SRE implies bounded support clustering}, every
$\mathrm{SRE}_{\exp}$ state is exponentially clustering. Therefore
\eq{
    \omega
    \notin
    \mathrm{SRE}_{\exp}.
}
This proves that
\eq{
    \mathrm{SRE}_{\exp}
    \subsetneq
    \text{area law}.
}
\end{example}

\begin{example}[A split state which does not satisfy an area law]
\label{ex:split not area law}
Partition the right half-chain into consecutive intervals
\eq{
    \ZZ_{\rm R}
    =
    \bigsqcup_{n\in\NN}
    I_n
}
with
\eq{
    |I_n|
    =
    2n.
}
More explicitly, set
\eq{
    s_n
    &:=
    1+n(n-1),\\
    I_n
    &:=
    \Set{
        s_n,
        s_n+1,
        \ldots,
        s_n+2n-1
    }.
}
Write
\eq{
    J_n
    &:=
    \Set{
        s_n,
        \ldots,
        s_n+n-1
    },\\
    J_n'
    &:=
    \Set{
        s_n+n,
        \ldots,
        s_n+2n-1
    }.
}
For $j\in\Set{1,\ldots,n}$, pair the sites
\eq{
    x_{n,j}
    &:=
    s_n+j-1,\\
    y_{n,j}
    &:=
    s_n+n+j-1.
}
On each pair $\Set{x_{n,j},y_{n,j}}$, take the even Bell vector
\eq{
    \xi_{n,j}
    :=
    \frac{1}{\sqrt{2}}
    \br{
        \ket{00}
        +
        \ket{11}
    },
}
where
\eq{
    \ket{11}
    :=
    a_{x_{n,j}}^*
    a_{y_{n,j}}^*
    \ket{00}.
}
Let $\varphi_n$ be the graded product of the corresponding $n$ pure
two-site states on $\calA_{I_n}$. Then $\varphi_n$ is pure and
parity-invariant.

Put the left half-chain in the vacuum state and define
\eql{\label{eq:split volume entangled blocks}
    \omega
    :=
    \omega_{{\rm vacuum},L}
    \hat{\otimes}
    \left(
        \hat{\bigotimes}_{n\in\NN}
        \varphi_n
    \right).
}
The state $\omega$ is pure and parity-invariant. Moreover, it is exactly
a graded product across the cut:
\eq{
    \omega
    =
    \left.\omega\right|_{\calA_L}
    \hat{\otimes}
    \left.\omega\right|_{\calA_R}.
}
In particular, $\omega$ is split.

However, the interval $J_n$ contains exactly one site from each of the
$n$ Bell pairs in $I_n$. Its reduced density matrix is therefore the
maximally mixed density matrix on $n$ complex Fermionic modes. Hence
\eq{
    S_{J_n}(\omega)
    =
    n\log2.
}
Since
\eq{
    S_{J_n}(\omega)
    \xrightarrow[n\to\infty]{}
    \infty,
}
the state $\omega$ does not satisfy an area law. Thus
\eq{
    \text{area law}
    \subsetneq
    \text{split}.
}
\end{example}

\subsubsection*{Exponential clustering}

We end this subsection with a discussion of exponential clustering. It
is not needed for the hierarchy in
\cref{eq:hierarchy of locality modes}, but it is useful to distinguish
two different notions which are sometimes both called exponential
decay of correlations.

\begin{defn}[Bounded-support exponential clustering]
A state $\omega$ is called exponentially clustering if, for every
$R<\infty$, there are constants $C_R,\mu_R>0$ such that, for all local
observables $a,b\in\calA_{\rm loc}$ satisfying
\eq{
    \diam\br{\supp(a)}
    &\leq
    R,\\
    \diam\br{\supp(b)}
    &\leq
    R,
}
one has
\eql{\label{eq:bounded-support exponential clustering}
    \left|
        \omega(ab)-\omega(a)\omega(b)
    \right|
    \leq
    C_R
    \norm{a}
    \norm{b}
    \ee^{
        -\mu_R
        \dist\br{
            \supp(a),
            \supp(b)
        }
    }.
}
\end{defn}

The dependence of both constants on $R$ is important. This condition
controls correlations detectable by observables of any fixed size, but
it provides no estimate uniform in the sizes of the two regions.

\begin{claim}
\label{claim:SRE implies bounded support clustering}
Every pure $\mathrm{SRE}_{\exp}$ state is exponentially clustering in
the sense of
\cref{eq:bounded-support exponential clustering}.
\end{claim}

\begin{proof}
Write
\eq{
    \omega
    =
    \omega_{\rm product}\circ\alpha_1
}
as in the proof of
\cref{claim:SRE implies area law}. Exponential localization of the
generator implies exponential quasi-locality of $\alpha_1$. Namely, for
every $R<\infty$, there are constants $C_R,\mu_R>0$ such that, whenever
$a$ is supported in a set of diameter at most $R$, there is an
observable $a^{(\ell)}$ supported in the $\ell$-neighborhood of
$\supp(a)$ satisfying
\eq{
    \norm{
        \alpha_1(a)-a^{(\ell)}
    }
    \leq
    C_R
    \norm{a}
    \ee^{-\mu_R\ell}.
}

Let $a,b$ satisfy the hypotheses of
\cref{eq:bounded-support exponential clustering}, and set
\eq{
    d
    :=
    \dist\br{
        \supp(a),
        \supp(b)
    }.
}
Choose $\ell$ to be a fixed fraction of $d$, so that the supports of
$a^{(\ell)}$ and $b^{(\ell)}$ remain disjoint. Since
$\omega_{\rm product}$ is a parity-invariant product state, it
factorizes on these disjoint local algebras:
\eq{
    \omega_{\rm product}
    \br{
        a^{(\ell)}b^{(\ell)}
    }
    =
    \omega_{\rm product}
    \br{
        a^{(\ell)}
    }
    \omega_{\rm product}
    \br{
        b^{(\ell)}
    }.
}
Using
\eq{
    \omega(ab)
    =
    \omega_{\rm product}
    \br{
        \alpha_1(a)\alpha_1(b)
    }
}
and the two quasi-local approximation estimates gives
\eq{
    \left|
        \omega(ab)-\omega(a)\omega(b)
    \right|
    \leq
    C_R'
    \norm{a}
    \norm{b}
    \ee^{-\mu_R'd}
}
for suitable constants $C_R',\mu_R'>0$. This is
\cref{eq:bounded-support exponential clustering}.
\end{proof}

The converse relation between exponential clustering and the area law
fails for the bounded-support definition above.

\begin{claim}
\label{claim:bounded support clustering does not imply area law}
There exists a pure, parity-invariant state which is exponentially
clustering in the sense of
\cref{eq:bounded-support exponential clustering}, but which does not
satisfy an area law.
\end{claim}

\begin{proof}
We use a block-product construction. Partition $\ZZ$ into consecutive
finite intervals
\eq{
    \ZZ
    =
    \bigsqcup_{n\in\ZZ}
    I_n
}
whose lengths
\eq{
    L_n
    :=
    |I_n|
}
satisfy
\eq{
    L_n
    \xrightarrow[|n|\to\infty]{}
    \infty.
}
Choose integers $R_n\to\infty$ sufficiently slowly, with
\eq{
    2R_n+2
    \leq
    \frac{L_n}{8}.
}

For each $n$, consider the even-parity subspace of the finite-volume
Fermionic Fock space over $I_n$. Standard concentration of measure for
a Haar-random unit vector in this subspace implies the existence of a
pure, parity-invariant block state $\varphi_n$ with the following two
properties. First, if $J_n\subset I_n$ is a central interval of length
$\floor{L_n/8}$, then
\eql{\label{eq:random block volume entropy}
    S_{J_n}(\varphi_n)
    \geq
    c_0L_n
}
for some constant $c_0>0$ independent of $n$. Second, for any
contractions $a,b\in\calA_{I_n}$ whose supports have diameter at most
$R_n$,
\eql{\label{eq:random block small correlations}
    \left|
        \varphi_n(ab)
        -
        \varphi_n(a)\varphi_n(b)
    \right|
    \leq
    \ee^{-c_1L_n}
}
for some $c_1>0$ independent of $n$.

For completeness, these two properties follow by requiring that every
reduced density matrix on at most $L_n/8$ sites be exponentially close
in trace norm to the normalized trace. A union bound and concentration
of measure show that this holds with probability tending to one. It
implies
\cref{eq:random block small correlations} by taking the union of the
two supports, and
\cref{eq:random block volume entropy} by entropy continuity. This is
the same random-state estimate used in
\cite[Appendix~D]{BrandaoHorodecki2015AreaLaw}; restricting the Haar
measure to the even-parity subspace changes only the inessential
constants.

Define the graded block-product state
\eql{\label{eq:random block product state}
    \omega
    :=
    \hat{\bigotimes}_{n\in\ZZ}
    \varphi_n.
}
Since all the block states are pure and parity-invariant, $\omega$ is
also pure and parity-invariant.

We verify bounded-support exponential clustering. Fix $R<\infty$. For
all but finitely many $n$, we have
\eq{
    R
    &\leq
    R_n,&
    L_n
    &\geq
    4R.
}
If two observables of diameter at most $R$ involve disjoint collections
of blocks, their correlation vanishes exactly by the block-product
property. If they involve a common sufficiently large block $I_n$,
expanding the observables in matrix bases for the finitely many block
factors they meet and applying
\cref{eq:random block small correlations} gives
\eq{
    \left|
        \omega(ab)-\omega(a)\omega(b)
    \right|
    \leq
    K_R
    \norm{a}
    \norm{b}
    \ee^{-c_1L_n},
}
where $K_R<\infty$ depends only on $R$ and the on-site dimension. If
\eq{
    d
    :=
    \dist\br{
        \supp(a),
        \supp(b)
    },
}
then observables which meet the same block satisfy
\eq{
    d
    \leq
    L_n+2R.
}
Consequently,
\eq{
    \ee^{-c_1L_n}
    \leq
    \ee^{2c_1R}
    \ee^{-c_1d}.
}
The finitely many exceptional blocks may be absorbed into the constant
$C_R$. Thus there are constants $C_R,\mu_R>0$ such that
\eq{
    \left|
        \omega(ab)-\omega(a)\omega(b)
    \right|
    \leq
    C_R
    \norm{a}
    \norm{b}
    \ee^{-\mu_Rd}.
}
Hence $\omega$ is exponentially clustering in the sense of
\cref{eq:bounded-support exponential clustering}.

On the other hand, by
\cref{eq:random block volume entropy,eq:random block product state},
\eq{
    S_{J_n}(\omega)
    =
    S_{J_n}(\varphi_n)
    \geq
    c_0L_n.
}
Since $L_n\to\infty$, we obtain
\eq{
    \sup_{
        \Lambda\Subset\ZZ\
        \mathrm{interval}
    }
    S_\Lambda(\omega)
    =
    \infty.
}
Therefore $\omega$ does not satisfy an area law.
\end{proof}

Brand{\~a}o and Horodecki use a stronger, uniform notion of exponential
clustering. For disjoint finite regions $X,Y\Subset\ZZ$, define
\eq{
    \operatorname{Cor}_\omega(X:Y)
    :=
    \sup
    \Set{
        \left|
            \omega(ab)-\omega(a)\omega(b)
        \right|:
        \begin{array}{c}
        a\in\calA_X,\ b\in\calA_Y,\\
        \norm{a}\leq1,\ \norm{b}\leq1
        \end{array}
    }.
}
The uniform Brand{\~a}o--Horodecki condition asks for constants
$\xi,l_0<\infty$, independent of the sizes of $X$ and $Y$, such that
\eql{\label{eq:uniform BH clustering}
    \operatorname{Cor}_\omega(X:Y)
    \leq
    2^{
        -\dist(X,Y)/\xi
    }
}
whenever
\eq{
    \dist(X,Y)
    \geq
    l_0.
}
This is much stronger than
\cref{eq:bounded-support exponential clustering}: in
\cref{eq:uniform BH clustering}, the observables may occupy regions of
arbitrarily large diameter while the constants remain fixed.

Brand{\~a}o and Horodecki prove that the uniform condition implies an
area law for pure states on finite one-dimensional chains
\cite{BrandaoHorodecki2015AreaLaw}. Their theorem does not directly
apply to a pure state of the infinite CAR chain, because the
restriction of an infinite pure state to a finite interval is mixed.
They explicitly identify the extension to infinitely many particles
and the corresponding von Neumann-algebraic setting as an open
problem.

The natural infinite-volume Fermionic analogue of their result is
therefore the conjectural implication
\eq{
    \left.
    \begin{array}{c}
    \omega\text{ pure and parity-invariant},\\
    \omega\text{ satisfies }
    \cref{eq:uniform BH clustering}
    \end{array}
    \right\}
    \quad
    \Longrightarrow
    \quad
    \sup_{
        \Lambda\Subset\ZZ\
        \mathrm{interval}
    }
    S_\Lambda(\omega)
    <
    \infty.
}
The block-product state in
\cref{claim:bounded support clustering does not imply area law} does
not satisfy the uniform condition: sufficiently large observables
within a random block recover the correlations hidden from
bounded-support measurements. Thus it does not provide a counterexample
to the uniform infinite-volume conjecture.
\section{LGAs, SREs and the Majorana number}
\label{sec:LGAs SREs}

\begin{claim}
    Let $[0,1] \ni t \mapsto \alpha_t$ be an LGA with a $\theta$-even generator. Then $\alpha_1\circ\sigma\circ\alpha_1^{-1}\circ\sigma^{-1}=\Ad{u}$ for some $u\in\calU(\calA)$ such that $\theta u = u$.
\end{claim}

\begin{proof}
Set
\eq{
\ZZ_{\rm L}
&:=
\Set{x\in\ZZ:x\leq0},\\
\ZZ_{\rm R}
&:=
\Set{x\in\ZZ:x\geq1}\\
I_r(x) &:= \Set{y\in\ZZ : |x-y|\leq r}\,.
}
The locality automorphism restricts to parity on the left half-chain and
to the identity on the right half-chain:
\eq{
\sigma|_{\calA_{\ZZ_{\rm L}}}
&=
\theta|_{\calA_{\ZZ_{\rm L}}},\\
\sigma|_{\calA_{\ZZ_{\rm R}}}
&=
\Id|_{\calA_{\ZZ_{\rm R}}}.
}

Let $\Phi$ be the exponentially localized zero-chain generating
$t\mapsto\alpha_t$, so that
\eq{
\delta_t(a)
=
\ii\sum_{x\in\ZZ}
[\Phi_t(x),a],
\qquad
\theta(\Phi_t(x))
=
\Phi_t(x).
}
The local approximants to $\Phi_t(x)$ may be chosen $\theta$-even. Indeed,
we may replace $\Phi_t^{(r)}(x)$ by
\eq{
\widetilde{\Phi}_t^{(r)}(x)
:=
\frac{1}{2}
\br{
\Phi_t^{(r)}(x)
+
\theta\br{\Phi_t^{(r)}(x)}
}.
}
Since $\Phi_t(x)$ is $\theta$-even, this replacement does not increase the
approximation error:
\eq{
\norm{
\Phi_t(x)
-
\widetilde{\Phi}_t^{(r)}(x)
}
\leq
C\ee^{-\mu r}.
}

For $x\leq-1$, take $r=-x$. Then
\eq{
I_{-x}(x)
\subseteq
\ZZ_{\rm L},
}
and hence
\eq{
\sigma\br{
\widetilde{\Phi}_t^{(-x)}(x)
}
=
\theta\br{
\widetilde{\Phi}_t^{(-x)}(x)
}
=
\widetilde{\Phi}_t^{(-x)}(x).
}
It follows that
\eq{
\norm{
\Phi_t(x)-\sigma\br{\Phi_t(x)}
}
&\leq
\norm{
\Phi_t(x)-\widetilde{\Phi}_t^{(-x)}(x)
}\\
&\quad+
\norm{
\sigma\br{
\Phi_t(x)-\widetilde{\Phi}_t^{(-x)}(x)
}
}\\
&\leq
2C\ee^{-\mu(-x)}.
}
Similarly, if $x\geq2$, then
\eq{
I_{x-1}(x)
\subseteq
\ZZ_{\rm R},
}
so
\eq{
\sigma\br{
\widetilde{\Phi}_t^{(x-1)}(x)
}
=
\widetilde{\Phi}_t^{(x-1)}(x),
}
and therefore
\eq{
\norm{
\Phi_t(x)-\sigma\br{\Phi_t(x)}
}
\leq
2C\ee^{-\mu(x-1)}.
}
The two remaining terms, corresponding to $x=0,1$, are uniformly bounded.
Consequently,
\eq{
\sum_{x\in\ZZ}
\sup_{t\in[0,1]}
\norm{
\Phi_t(x)-\sigma\br{\Phi_t(x)}
}
<
\infty.
}
Thus the series
\eq{
K_t
:=
\sum_{x\in\ZZ}
\br{
\Phi_t(x)-\sigma\br{\Phi_t(x)}
}
}
converges in norm, uniformly in $t$. In particular,
$t\mapsto K_t$ is norm-continuous. Moreover, $K_t$ is self-adjoint and,
since $\sigma$ commutes with $\theta$,
\eq{
\theta(K_t)
=
K_t.
}

Define
\eq{
\beta_t
:=
\sigma\circ\alpha_t\circ\sigma^{-1}.
}
This LGA is generated by
\eq{
\delta_t^\sigma
:=
\sigma\circ\delta_t\circ\sigma^{-1},
}
where
\eq{
\delta_t^\sigma(a)
=
\ii\sum_{x\in\ZZ}
\left[
\sigma\br{\Phi_t(x)},
a
\right].
}
It follows from the norm convergence defining $K_t$ that
\eq{
\br{
\delta_t-\delta_t^\sigma
}(a)
=
\ii[K_t,a].
}

Consider the relative evolution
\eq{
\gamma_t
:=
\alpha_t\circ\beta_t^{-1}.
}
For local $a\in\calA_{\rm loc}$, differentiation gives
\eq{
\frac{\dif}{\dif t}\gamma_t(a)
&=
\alpha_t\br{
\br{
\delta_t-\delta_t^\sigma
}
\br{
\beta_t^{-1}(a)
}
}\\
&=
\ii\left[
\alpha_t(K_t),
\gamma_t(a)
\right].
}
Set
\eq{
h_t
:=
\alpha_t(K_t).
}
Since the generator of $\alpha_t$ is $\theta$-even, $\alpha_t$ commutes
with $\theta$. Therefore
\eq{
\theta(h_t)
=
h_t.
}

Let $u_t$ be the norm-continuous unitary solution of
\eq{
\frac{\dif}{\dif t}u_t
=
-\ii u_t h_t,
\qquad
u_0
=
\Id.
}
With the convention
\eq{
\Ad{u_t}(a)
=
u_t^\ast a u_t,
}
we have
\eq{
\frac{\dif}{\dif t}
\Ad{u_t}(a)
=
\ii\left[
h_t,
\Ad{u_t}(a)
\right].
}
Thus $\gamma_t$ and $\Ad{u_t}$ satisfy the same evolution equation and
have the same initial condition. Uniqueness of the evolution, first on
$\calA_{\rm loc}$ and then by continuity on $\calA$, gives
\eq{
\gamma_t
=
\Ad{u_t}.
}

It remains to verify the parity of $u_t$. Since $\theta(h_t)=h_t$, both
$t\mapsto u_t$ and $t\mapsto\theta(u_t)$ solve
\eq{
\frac{\dif}{\dif t}v_t
=
-\ii v_t h_t,
\qquad
v_0
=
\Id.
}
Uniqueness therefore implies
\eq{
\theta(u_t)
=
u_t
\qquad
\br{
t\in[0,1]
}.
}

Finally,
\eq{
\gamma_1
&=
\alpha_1\circ
\br{
\sigma\circ\alpha_1\circ\sigma^{-1}
}^{-1}\\
&=
\alpha_1\circ\sigma\circ
\alpha_1^{-1}\circ\sigma^{-1}.
}
Taking $u:=u_1$, we conclude that
\eq{
\alpha_1\circ\sigma\circ
\alpha_1^{-1}\circ\sigma^{-1}
=
\Ad{u},
\qquad
\theta(u)
=
u.
}
\end{proof}

\begin{cor}
    If $\omega$ is a symmetric SRE state, i.e., $\omega=\omega_0\circ\alpha_1$ with $\omega_0$ a pure product state and $\alpha_1$ an LGA with a symmetric generator, then its index is trivial, \eq{
    \calN_\sigma(\omega) = 1\,.
    }
\end{cor}

\begin{claim}
    Let $k\in\ZZ$ be given. Then the $k$-Kitaev state and the $k+2$-Kitaev state are path-connected by a symmetric LGA.
\end{claim}

\begin{proof}
Write
\eq{
P_k
=
\frac{1}{2}
\begin{bmatrix}
\Id & -R^{k\ast}\\
-R^k & \Id
\end{bmatrix},
\qquad
\omega_k
:=
\omega_{P_k}.
}
Consider the one-particle unitary
\eq{
V
:=
\begin{bmatrix}
R & 0\\
0 & R^\ast
\end{bmatrix}.
}
A direct calculation gives
\eq{
V^\ast P_kV
&=
\frac{1}{2}
\begin{bmatrix}
\Id
&
-R^\ast R^{k\ast}R^\ast
\\
-RR^kR
&
\Id
\end{bmatrix}\\
&=
\frac{1}{2}
\begin{bmatrix}
\Id
&
-R^{(k+2)\ast}
\\
-R^{k+2}
&
\Id
\end{bmatrix}\\
&=
P_{k+2}.
}
It therefore remains to show that the Bogoliubov automorphism induced by
$V$ is the endpoint of an LGA with a $\theta$-even generator.

Let $e_1,e_2$ denote the standard basis of $\CC^2$ and set
\eq{
f_{x,1}
&:=
\ii\delta_x\otimes e_1,\\
f_{x,2}
&:=
\delta_x\otimes e_2.
}
Since
\eq{
\Xi
=
-\calC\br{\Id\otimes\sigma_3},
}
we have
\eq{
\Xi f_{x,j}
=
f_{x,j}
\qquad
\br{
x\in\ZZ,\ 
j\in\Set{1,2}
}.
}
Consequently, the operators
\eq{
\gamma_{x,j}
:=
\sqrt{2}B(f_{x,j})
}
are Majorana unitaries:
\eq{
\gamma_{x,j}^\ast
&=
\gamma_{x,j},\\
\gamma_{x,j}^2
&=
1_{\calA},\\
\theta\br{\gamma_{x,j}}
&=
-\gamma_{x,j}.
}

Define two translation-invariant, finite-range zero-chains by
\eq{
\Phi^{(0)}(x)
&:=
-\frac{\ii}{2}
\gamma_{x,1}\gamma_{x,2},\\
\Phi^{(1)}(x)
&:=
\frac{\ii}{2}
\gamma_{x,1}\gamma_{x+1,2}.
}
Each term is self-adjoint and $\theta$-even:
\eq{
\Phi^{(j)}(x)^\ast
&=
\Phi^{(j)}(x),\\
\theta\br{\Phi^{(j)}(x)}
&=
\Phi^{(j)}(x)
\qquad
\br{
j\in\Set{0,1}
}.
}

Choose continuous functions
\eq{
\lambda_0,\lambda_1
:
[0,1]
\longrightarrow
\RR
}
such that
\eq{
\supp(\lambda_0)
&\subseteq
\br{
0,\frac{1}{2}
},\\
\supp(\lambda_1)
&\subseteq
\br{
\frac{1}{2},1
},
}
with disjoint supports, and
\eq{
\int_0^1\lambda_0(t)\dif t
=
\int_0^1\lambda_1(t)\dif t
=
\frac{\pi}{2}.
}
Set
\eq{
\Phi_t(x)
:=
\lambda_0(t)\Phi^{(0)}(x)
+
\lambda_1(t)\Phi^{(1)}(x).
}
Then $\Phi$ is a continuous, strictly finite-range zero-chain satisfying
\eq{
\theta\br{\Phi_t(x)}
=
\Phi_t(x).
}
Let $t\mapsto\alpha_t$ be the corresponding LGA.

For the first part of the evolution, the CAR relations give
\eq{
\ii
\left[
\Phi^{(0)}(x),
\gamma_{x,1}
\right]
&=
-\gamma_{x,2},\\
\ii
\left[
\Phi^{(0)}(x),
\gamma_{x,2}
\right]
&=
\gamma_{x,1}.
}
Thus evolution through an angle $s$ induces the one-particle
transformation
\eq{
J_s
:=
\begin{bmatrix}
\cos(s)\Id
&
\sin(s)\Id
\\
-\sin(s)\Id
&
\cos(s)\Id
\end{bmatrix}.
}
In particular,
\eq{
J_{\pi/2}
=
\begin{bmatrix}
0 & \Id\\
-\Id & 0
\end{bmatrix}.
}

For the second part of the evolution, one similarly has
\eq{
\ii
\left[
\Phi^{(1)}(x),
\gamma_{x,1}
\right]
&=
\gamma_{x+1,2},\\
\ii
\left[
\Phi^{(1)}(x),
\gamma_{x+1,2}
\right]
&=
-\gamma_{x,1}.
}
The corresponding one-particle transformation through an angle $s$ is
therefore
\eq{
K_s
:=
\begin{bmatrix}
\cos(s)\Id
&
-\sin(s)R^\ast
\\
\sin(s)R
&
\cos(s)\Id
\end{bmatrix}.
}
At $s=\pi/2$,
\eq{
K_{\pi/2}
=
\begin{bmatrix}
0 & -R^\ast\\
R & 0
\end{bmatrix}.
}
Since the two evolutions occur in this order, their endpoint
one-particle transformation is
\eq{
J_{\pi/2}K_{\pi/2}
&=
\begin{bmatrix}
0 & \Id\\
-\Id & 0
\end{bmatrix}
\begin{bmatrix}
0 & -R^\ast\\
R & 0
\end{bmatrix}\\
&=
\begin{bmatrix}
R & 0\\
0 & R^\ast
\end{bmatrix}\\
&=
V.
}
It follows that
\eq{
\alpha_1\br{B(\psi)}
=
B(V\psi)
\qquad
\br{
\psi\in\calH
}.
}

Composition of a quasi-free state with this Bogoliubov automorphism
transforms its covariance projection according to
\eq{
\br{
\omega_{P_k}\circ\alpha_1
}
\br{
B(\psi)^\ast B(\varphi)
}
&=
\omega_{P_k}
\br{
B(V\psi)^\ast B(V\varphi)
}\\
&=
\langle
V\psi,
P_kV\varphi
\rangle\\
&=
\langle
\psi,
V^\ast P_kV\varphi
\rangle.
}
Hence
\eq{
\omega_{P_k}\circ\alpha_1
&=
\omega_{V^\ast P_kV}\\
&=
\omega_{P_{k+2}}.
}
Thus the $k$-Kitaev state and the $(k+2)$-Kitaev state are connected by
an LGA with a $\theta$-even generator.

Finally, the inverse Bogoliubov automorphism is induced by $V^\ast$ and
is likewise the endpoint of an LGA with a $\theta$-even generator.
Moreover,
\eq{
VP_kV^\ast
&=
\frac{1}{2}
\begin{bmatrix}
\Id
&
-RR^{k\ast}R
\\
-R^\ast R^kR^\ast
&
\Id
\end{bmatrix}\\
&=
\frac{1}{2}
\begin{bmatrix}
\Id
&
-R^{(k-2)\ast}
\\
-R^{k-2}
&
\Id
\end{bmatrix}\\
&=
P_{k-2}.
}
Therefore the $k$-Kitaev state and the $(k-2)$-Kitaev state are also
connected by a symmetric LGA.
\end{proof}

\begin{thm}[LGA-injectivity of the index] Let $\omega=\omega_{\rm p}\circ\beta_1$ be a $\theta$-invariant SRE state with a $\theta$-symmetric generator for its LGA $t\mapsto\beta_t$, such that $\calN(\omega_{\rm p},\omega_0)=+1$ where $\omega_0$ is the $k=0$ Kitaev state.

Then there exists an LGA $t\mapsto \alpha_t$ with a $\theta$-symmetric generator such that 
\eql{\omega = \omega_0\circ\alpha_1\,.} 
\end{thm}

\begin{proof}
Since $P_0$ has no matrix elements coupling distinct sites, the $k=0$
Kitaev state $\omega_0$ is a pure product state.

For $j\in\Set{{\rm p},0}$, write
\eq{
\omega_j
=
\hat{\bigotimes}_{x\in\ZZ}
\omega_{j,x}.
}
Let $\Gamma_x\in\calA_{\Set{x}}$ denote the one-site parity unitary, so
that
\eq{
\Gamma_x^\ast
&=
\Gamma_x,\\
\Gamma_x^2
&=
\Id,\\
\theta(a)
&=
\Gamma_xa\Gamma_x
\qquad
\br{
a\in\calA_{\Set{x}}
}.
}
Since $\omega_{j,x}$ is pure and $\theta$-invariant, it is represented by
a unit vector $\xi_{j,x}$ of definite local parity. Thus there is some
$\ve_{j,x}\in\Set{\pm1}$ such that
\eq{
\Gamma_x\xi_{j,x}
=
\ve_{j,x}\xi_{j,x}.
}
Set
\eq{
D
:=
\Set{
x\in\ZZ:
\ve_{{\rm p},x}\neq\ve_{0,x}
}.
}

We first show that $D$ is finite. Since
$\calN(\omega_{\rm p},\omega_0)$ is defined, there is a unitary
$u\in\calU(\calA)$ such that
\eq{
\omega_0
=
\omega_{\rm p}\circ\Ad{u}.
}
Choose a finite set $\Lambda\subseteq\ZZ$ and a local unitary
$v\in\calU(\calA_\Lambda)$ satisfying
\eq{
\norm{u-v}
<
1.
}
If $x\notin\Lambda$, then $\Gamma_x$ is even and has support disjoint
from that of $v$, so
\eq{
v^\ast\Gamma_xv
=
\Gamma_x.
}
Consequently,
\eq{
\left|
\ve_{0,x}-\ve_{{\rm p},x}
\right|
&=
\left|
\omega_{\rm p}
\br{
u^\ast\Gamma_xu
}
-
\omega_{\rm p}
\br{
v^\ast\Gamma_xv
}
\right|\\
&\leq
\norm{
u^\ast\Gamma_xu
-
v^\ast\Gamma_xv
}\\
&\leq
2\norm{u-v}\\
&<
2.
}
Since $\ve_{0,x}$ and $\ve_{{\rm p},x}$ both belong to $\Set{\pm1}$, this
implies
\eq{
\ve_{0,x}
=
\ve_{{\rm p},x}
\qquad
\br{
x\notin\Lambda
}.
}
Thus
\eq{
D
\subseteq
\Lambda,
}
and in particular $D$ is finite.

We next relate the parity of $\lvert D\rvert$ to the relative index.
Let
\eq{
\br{
\pi_{\rm p},
\calH_{\rm p},
\Omega_{\rm p}
}
}
be the GNS representation of $\omega_{\rm p}$, and let
$\Theta_{\rm p}$ be the unitary implementing $\theta$, normalized by
\eq{
\Theta_{\rm p}\Omega_{\rm p}
=
\Omega_{\rm p}.
}
In the product realization of this representation,
\eq{
\Omega_{\rm p}
=
\hat{\bigotimes}_{x\in\ZZ}
\xi_{{\rm p},x},
}
and
\eq{
\Theta_{\rm p}
=
\hat{\bigotimes}_{x\in\ZZ}
\br{
\ve_{{\rm p},x}\Gamma_x
}.
}
Indeed, each local factor $\ve_{{\rm p},x}\Gamma_x$ fixes
$\xi_{{\rm p},x}$ and implements the local parity automorphism.

Set
\eq{
\Omega_0
:=
\pi_{\rm p}(u)\Omega_{\rm p}.
}
Then $\Omega_0$ represents $\omega_0$ in the same GNS representation.
Since $\omega_0$ is a pure product state, the standard incomplete
tensor-product description gives, up to an overall phase,
\eq{
\Omega_0
=
\hat{\bigotimes}_{x\in\ZZ}
\xi_{0,x}.
}
It follows that
\eq{
\Theta_{\rm p}\Omega_0
&=
\hat{\bigotimes}_{x\in\ZZ}
\br{
\ve_{{\rm p},x}\Gamma_x\xi_{0,x}
}\\
&=
\br{
\prod_{x\in\ZZ}
\ve_{{\rm p},x}\ve_{0,x}
}
\Omega_0\\
&=
\br{
-1
}^{\lvert D\rvert}
\Omega_0.
}
On the other hand,
\eq{
\calN(\omega_{\rm p},\omega_0)
&=
\omega_{\rm p}
\br{
u^\ast\theta(u)
}\\
&=
\ip{
\pi_{\rm p}(u)\Omega_{\rm p}
}{
\Theta_{\rm p}\pi_{\rm p}(u)\Omega_{\rm p}
}\\
&=
\ip{
\Omega_0
}{
\Theta_{\rm p}\Omega_0
}.
}
Therefore,
\eql{
\calN(\omega_{\rm p},\omega_0)
=
\br{
-1
}^{\lvert D\rvert}.
}
The hypothesis
\eq{
\calN(\omega_{\rm p},\omega_0)
=
1
}
thus implies that $\lvert D\rvert$ is even.

We now construct a symmetric LGA connecting $\omega_0$ to
$\omega_{\rm p}$. For every $x\notin D$, the vectors $\xi_{0,x}$ and
$\xi_{{\rm p},x}$ belong to the same parity eigenspace. Hence there is
an even one-site unitary
\eq{
U_x
\in
\calU\br{
\calA_{\Set{x}}
}
}
such that
\eq{
U_x\xi_{0,x}
&=
\xi_{{\rm p},x},\\
\theta(U_x)
&=
U_x.
}
Since the even unitary group of the finite-dimensional algebra
$\calA_{\Set{x}}$ is connected, we may write
\eq{
U_x
=
\ee^{-\ii h_x},
}
where
\eq{
h_x^\ast
&=
h_x,\\
\theta(h_x)
&=
h_x,\\
\norm{h_x}
&\leq
\pi.
}
Define the time-independent zero-chain
\eq{
\Phi^{(1)}_t(x)
:=
\begin{cases}
h_x,
&
x\notin D,\\
0,
&
x\in D.
\end{cases}\,.
}
This zero-chain is uniformly bounded, strictly on-site, and
$\theta$-even. Let $t\mapsto\tau_t$ be the corresponding LGA. Since the
terms at distinct sites are even and have disjoint supports, they
commute, and hence
\eq{
\omega_0\circ\tau_1
=
\br{
\hat{\bigotimes}_{x\notin D}
\omega_{{\rm p},x}
}
\hat{\otimes}
\br{
\hat{\bigotimes}_{x\in D}
\omega_{0,x}
}.
}

It remains to change the finitely many factors indexed by $D$. Since
$\lvert D\rvert$ is even, choose a pairing
\eq{
D
=
\Set{
x_1,y_1,\ldots,x_m,y_m
}.
}
For each $r\in\Set{1,\ldots,m}$, we have
\eq{
\ve_{{\rm p},x_r}
&=
-\ve_{0,x_r},\\
\ve_{{\rm p},y_r}
&=
-\ve_{0,y_r},
}
and therefore
\eq{
\ve_{{\rm p},x_r}\ve_{{\rm p},y_r}
=
\ve_{0,x_r}\ve_{0,y_r}.
}
Thus the vectors
\eq{
\xi_{0,x_r}
\hat{\otimes}
\xi_{0,y_r}
}
and
\eq{
\xi_{{\rm p},x_r}
\hat{\otimes}
\xi_{{\rm p},y_r}
}
belong to the same total-parity eigenspace of the two-site algebra
$\calA_{\Set{x_r,y_r}}$. Consequently, there is an even unitary
\eq{
V_r
\in
\calU\br{
\calA_{\Set{x_r,y_r}}
}
}
such that
\eq{
V_r
\br{
\xi_{0,x_r}
\hat{\otimes}
\xi_{0,y_r}
}
=
\xi_{{\rm p},x_r}
\hat{\otimes}
\xi_{{\rm p},y_r}.
}
As before, we may write
\eq{
V_r
=
\ee^{-\ii k_r},
}
where
\eq{
k_r^\ast
&=
k_r,\\
\theta(k_r)
&=
k_r,\\
\norm{k_r}
&\leq
\pi.
}

Define a second time-independent zero-chain by
\eq{
\Phi^{(2)}_t(x)
:=
\begin{cases}
k_r,
&
x=x_r
\text{ for some }
r\in\Set{1,\ldots,m},\\
0,
&
\text{otherwise}.
\end{cases}\,.
}
Since only finitely many terms occur, this is a finite-range,
$\theta$-even zero-chain. Let $t\mapsto\rho_t$ be its LGA. The pairs are
disjoint, and all the terms are even, so
\eq{
\br{
\omega_0\circ\tau_1
}
\circ\rho_1
=
\omega_{\rm p}.
}

After a continuous time reparametrization, we may concatenate
$t\mapsto\tau_t$ and $t\mapsto\rho_t$ to obtain an LGA
$t\mapsto\gamma_t$ with a $\theta$-even generator and endpoint
\eq{
\gamma_1
=
\tau_1\circ\rho_1.
}
It satisfies
\eql{
\omega_{\rm p}
=
\omega_0\circ\gamma_1.
}

Finally, concatenate $t\mapsto\gamma_t$ with the given LGA
$t\mapsto\beta_t$. Since both generators are $\theta$-even, the
concatenated evolution $t\mapsto\alpha_t$ also has a $\theta$-even
generator, and its endpoint is
\eq{
\alpha_1
=
\gamma_1\circ\beta_1.
}
Therefore,
\eq{
\omega_0\circ\alpha_1
&=
\omega_0\circ\gamma_1\circ\beta_1\\
&=
\omega_{\rm p}\circ\beta_1\\
&=
\omega.
}
This proves the claim.
\end{proof}
\section{Equivariant local adjustment}
\label{app:equivariant local adjustment}

\begin{proof}[Proof of \cref{lem:equivariant local adjustment}]
The fixed-point algebra is
\eq{
    \calB^\gamma
    :=
    \Set{
        b\in\calB:
        \gamma_g(b)=b
        \text{ for every }g\in G
    }.
}
Since \(\gamma\) is pointwise outer and \(\calB\) is simple,
\(\calB^\gamma\) is simple. It is also infinite-dimensional. Indeed,
the canonical conditional expectation
\eq{
    E_\gamma(a)
    :=
    \frac{1}{|G|}
    \sum_{g\in G}
    \gamma_g(a)
}
has finite index in this situation, so finite-dimensionality of
\(\calB^\gamma\) would imply finite-dimensionality of \(\calB\).

Let \((\pi,\calH,\Omega)\) be the GNS representation of \(\varphi\).
Denote by \(U_g\) the canonical implementation of \(\gamma_g\):
\eq{
    U_g\pi(a)U_g^*
    &=
    \pi(\gamma_g(a)),&
    U_g\Omega
    &=
    \Omega.
}
Let
\eq{
    \calH^G
    :=
    \Set{
        \xi\in\calH:
        U_g\xi=\xi
        \text{ for every }g\in G
    }.
}
If
\eq{
    P_G
    :=
    \frac{1}{|G|}
    \sum_{g\in G}U_g,
}
then
\eq{
    P_G\pi(a)\Omega
    =
    \pi(E_\gamma(a))\Omega.
}
It follows that
\eq{
    \overline{\pi(\calB^\gamma)\Omega}
    =
    \calH^G.
}
Moreover, Kaplansky density followed by averaging gives
\eq{
    \pi(\calB^\gamma)''
    =
    \Set{
        U_g:g\in G
    }'.
}
Compressing to the trivial isotypic subspace therefore gives
\eq{
    \left.
    \pi(\calB^\gamma)''
    \right|_{\calH^G}
    =
    \calB(\calH^G).
}
Thus the representation
\eq{
    \rho(d)
    :=
    \left.\pi(d)\right|_{\calH^G},
    \qquad
    d\in\calB^\gamma,
}
is irreducible, and \(\varphi|_{\calB^\gamma}\) is pure. Since
\(\calB^\gamma\) is simple, \(\rho\) is faithful. Furthermore,
\eq{
    \rho(\calB^\gamma)
    \cap
    \Compacts{\calH^G}
    =
    \Set{0}.
}
Indeed, otherwise simplicity would imply
\(\rho(\calB^\gamma)\subseteq\Compacts{\calH^G}\). Since
\(\calB^\gamma\) is unital, this would make \(\calH^G\)
finite-dimensional and hence make the faithful algebra
\(\rho(\calB^\gamma)\) finite-dimensional, a contradiction.

We first prove the corresponding fixed-vector statement. We may
assume that \(\calF\) is contained in the unit ball. Replacing it by
\eq{
    \bigcup_{g\in G}
    \gamma_g(\calF)\, ,
}
we may also assume that \(\calF\) is \(G\)-invariant. We shall choose a
finite set \(\calG_0\subset\calB^\gamma\) and \(\delta>0\) such that,
whenever \(\eta\in\calH^G\) is a unit vector satisfying
\eq{
    \left|
        \ip{\eta}{\rho(d)\eta}
        -
        \ip{\Omega}{\rho(d)\Omega}
    \right|
    <
    \delta
    \qquad
    (d\in\calG_0),
}
there exists a norm-continuous path
\eq{
    [0,1]\ni t
    \longmapsto
    u_t\in\calU(\calB^\gamma)
}
such that
\eq{
    u_0
    &=
    1_{\calB^\gamma},&
    \pi(u_1)\Omega
    &=
    \eta,
}
and
\eql{\label{eq:fixed sector local adjustment}
    \norm{
        \Ad{u_t}(a)-a
    }
    <
    \varepsilon
    \qquad
    (a\in\calF,\ t\in[0,1]).
}

Choose \(\kappa,\nu>0\) such that
\eq{
    2\pi\kappa+2\nu
    <
    \varepsilon.
}
By the norm-controlled form of Kadison transitivity
\cite[Theorem~3.4]{SpiegelEtAl2022Kadison} for the irreducible
representation \(\rho\), there exists \(s>0\) such that, whenever
\(\xi_0,\xi_1\in\calH^G\) are unit vectors satisfying
\eq{
    \norm{\xi_0-\xi_1}
    <
    s,
}
there is a self-adjoint \(q\in\calB^\gamma\) satisfying
\eq{
    \ee^{\ii\rho(q)}\xi_0
    &=
    \xi_1,&
    \norm{q}
    &<
    \nu.
}
Indeed, one first chooses a self-adjoint operator of small norm on
\(\operatorname{span}\Set{\xi_0,\xi_1}\) whose exponential sends
\(\xi_0\) to \(\xi_1\), and then applies Kadison transitivity with
norm control. Fix \(r>0\) such that
\eq{
    2\pi r
    <
    s.
}

Apply the row property
\cite[Property~1.3 and Section~4]{KishimotoOzawaSakai2003Homogeneity}
to the irreducible representation \(\pi\), the projection onto
\(\CC\Omega\), and the finite set \(\calF\). We obtain a row
\eq{
    x
    =
    (x_1,\ldots,x_n)
    \in
    M_{1,n}(\calB)
}
such that
\eq{
    \norm{xx^*}
    &\leq
    1,&
    \pi(xx^*)\Omega
    &=
    \Omega,
}
and
\eq{
    \norm{
        \operatorname{ad}_a
        \circ
        \operatorname{Ad}_x
    }
    <
    \kappa
    \qquad
    (a\in\calF),
}
where
\eq{
    \operatorname{Ad}_x(b)
    :=
    \sum_{j=1}^n
    x_jbx_j^*.
}
Set
\eq{
    d_{ij}
    &:=
    E_\gamma(x_ix_j^*)
    \in
    \calB^\gamma,\\
    \calG_0
    &:=
    \Set{
        d_{ij}:
        1\leq i,j\leq n
    }.
}
Choose \(\delta>0\) so small that the finite-dimensional perturbation
in the proof of
\cite[Theorem~3.1]{KishimotoOzawaSakai2003Homogeneity}, with tolerance
\(r\), applies to any two of the families below whose Gram matrices
differ by at most \(2\delta\).

Suppose that \(\eta\in\calH^G\) satisfies the preceding
\(\delta\)-estimate. Since \(\rho(\calB^\gamma)\) contains no nonzero
compact operators, the Glimm argument in the proof of
\cite[Theorem~3.1]{KishimotoOzawaSakai2003Homogeneity} gives a unit
vector \(\zeta\in\calH^G\) such that
\eq{
    \left|
        \ip{\zeta}{\rho(d)\zeta}
        -
        \ip{\Omega}{\rho(d)\Omega}
    \right|
    <
    \delta
    \qquad
    (d\in\calG_0),
}
and
\eq{
    \zeta
    \perp
    \operatorname{span}
    \Set{
        \rho(d_{ij})\Omega,
        \rho(d_{ij})\eta:
        1\leq i,j\leq n
    }.
}
Here exact orthogonality is obtained by projecting a sufficiently far
vector in the weakly-null Glimm sequence away from this
finite-dimensional subspace and normalizing; this changes the
preceding matrix coefficients arbitrarily little.

For fixed vectors \(\xi\in\calH^G\), we have
\eq{
    \ip{\xi}{\pi(x_ix_j^*)\xi}
    =
    \ip{\xi}{\rho(d_{ij})\xi}.
}
Moreover,
\eq{
    \ip{\pi(x_i^*)\zeta}
       {\pi(x_j^*)\Omega}
    &=
    \ip{\zeta}{\pi(x_ix_j^*)\Omega}\\
    &=
    \ip{\zeta}{\rho(d_{ij})\Omega}\\
    &=
    0,
}
and similarly
\eq{
    \ip{\pi(x_i^*)\zeta}
       {\pi(x_j^*)\eta}
    =
    0.
}
Consequently,
\eq{
    \operatorname{span}
    \Set{
        \pi(x_i^*)\zeta
    }_i
    \perp
    \operatorname{span}
    \Set{
        \pi(x_i^*)\Omega,
        \pi(x_i^*)\eta
    }_i.
}
The Gram matrices of the families associated with
\((\Omega,\zeta)\) differ by less than \(\delta\), while those
associated with \((\eta,\zeta)\) differ by less than \(2\delta\).
The finite-dimensional perturbation just cited therefore gives
positive contractions
\eq{
    \overline h_0,\overline h_1
    \in
    \calB
}
such that
\eq{
    \norm{
        [a,\overline h_j]
    }
    <
    \kappa
    \qquad
    (a\in\calF,\ j\in\Set{0,1}),
}
and
\eq{
    \norm{
        \pi(\overline h_0)(\Omega+\zeta)
    }
    &<
    r,\\
    \norm{
        \br{
            \Id-\pi(\overline h_0)
        }
        (\Omega-\zeta)
    }
    &<
    r,
}
as well as
\eq{
    \norm{
        \pi(\overline h_1)(\eta+\zeta)
    }
    &<
    r,\\
    \norm{
        \br{
            \Id-\pi(\overline h_1)
        }
        (\eta-\zeta)
    }
    &<
    r.
}
Concretely, one perturbs the map between the two families to an
isometry between their orthogonal spans, takes the positive
contraction associated with the self-adjoint flip, lifts it by
Kadison transitivity, and applies \(\operatorname{Ad}_x\). The
commutator estimates are then exactly the row estimate above.

The elements \(\overline h_j\) need not belong to
\(\calB^\gamma\). We therefore average the completed KOS generators
and set
\eq{
    k_j
    :=
    E_\gamma(\overline h_j)
    \in
    \calB^\gamma,
    \qquad
    j\in\Set{0,1}.
}
Since \(\calF\) is \(G\)-invariant,
\eq{
    \norm{
        [a,k_j]
    }
    &\leq
    \frac{1}{|G|}
    \sum_{g\in G}
    \norm{
        [
            \gamma_{g^{-1}}(a),
            \overline h_j
        ]
    }\\
    &<
    \kappa
    \qquad
    (a\in\calF).
}
Averaging also preserves the vector estimates because
\(\Omega,\eta,\zeta\in\calH^G\). For example,
\eq{
    \pi(k_0)(\Omega+\zeta)
    =
    \frac{1}{|G|}
    \sum_{g\in G}
    U_g\pi(\overline h_0)(\Omega+\zeta),
}
and hence
\eq{
    \norm{
        \pi(k_0)(\Omega+\zeta)
    }
    <
    r.
}
In the same way,
\eq{
    \norm{
        \br{
            \Id-\pi(k_0)
        }
        (\Omega-\zeta)
    }
    &<
    r,\\
    \norm{
        \pi(k_1)(\eta+\zeta)
    }
    &<
    r,\\
    \norm{
        \br{
            \Id-\pi(k_1)
        }
        (\eta-\zeta)
    }
    &<
    r.
}

For \(0\leq T\leq\Id\), functional calculus gives
\eq{
    \norm{
        \br{
            \ee^{\ii\pi T}-\Id
        }\xi
    }
    &\leq
    \pi\norm{T\xi},\\
    \norm{
        \br{
            \ee^{\ii\pi T}+\Id
        }\xi
    }
    &\leq
    \pi\norm{
        \br{
            \Id-T
        }\xi
    }.
}
Using the decompositions
\eq{
    \Omega
    &=
    \onehalf
    \br{
        \Omega+\zeta
    }
    +
    \onehalf
    \br{
        \Omega-\zeta
    },\\
    \zeta
    &=
    \onehalf
    \br{
        \Omega+\zeta
    }
    -
    \onehalf
    \br{
        \Omega-\zeta
    },
}
we obtain
\eq{
    \norm{
        \ee^{\ii\pi\pi(k_0)}\Omega-\zeta
    }
    <
    \pi r.
}
Similarly,
\eq{
    \norm{
        \ee^{\ii\pi\pi(k_1)}\eta-\zeta
    }
    <
    \pi r.
}
Define
\eq{
    \widetilde\Omega
    &:=
    \ee^{\ii\pi\pi(k_0)}\Omega,\\
    \widetilde\eta
    &:=
    \ee^{\ii\pi\pi(k_1)}\eta.
}
Then
\eq{
    \norm{
        \widetilde\Omega-\widetilde\eta
    }
    <
    2\pi r
    <
    s.
}
Both vectors belong to \(\calH^G\). By the choice of \(s\), there is
a self-adjoint \(q\in\calB^\gamma\) such that
\eq{
    \ee^{\ii\rho(q)}
    \widetilde\Omega
    &=
    \widetilde\eta,&
    \norm{q}
    &<
    \nu.
}
Consequently, the unitary
\eq{
    v
    :=
    \ee^{-\ii\pi k_1}
    \ee^{\ii q}
    \ee^{\ii\pi k_0}
    \in
    \calU(\calB^\gamma)
}
satisfies
\eq{
    \pi(v)\Omega
    =
    \eta.
}

Concatenate the three paths
\eq{
    t
    &\longmapsto
    \ee^{\ii\pi t k_0},\\
    t
    &\longmapsto
    \ee^{\ii t q}
    \ee^{\ii\pi k_0},\\
    t
    &\longmapsto
    \ee^{-\ii\pi t k_1}
    \ee^{\ii q}
    \ee^{\ii\pi k_0},
    \qquad
    t\in[0,1].
}
After reparametrization, this gives a norm-continuous path
\(t\mapsto u_t\) in \(\calU(\calB^\gamma)\) from
\(1_{\calB^\gamma}\) to \(v\). Furthermore,
\eq{
    \norm{
        \Ad{\ee^{\ii\pi t k_j}}(a)-a
    }
    &\leq
    \pi\norm{
        [k_j,a]
    }
    <
    \pi\kappa,\\
    \norm{
        \Ad{\ee^{\ii t q}}(a)-a
    }
    &\leq
    2\norm{q}
    <
    2\nu.
}
It follows that
\eq{
    \norm{
        \Ad{u_t}(a)-a
    }
    <
    2\pi\kappa+2\nu
    <
    \varepsilon
}
for \(a\in\calF\) and \(t\in[0,1]\). This proves
\cref{eq:fixed sector local adjustment}.

It remains to pass from fixed vectors to invariant pure states. Let
\(\psi\) be as in
\cref{lem:equivariant local adjustment}, and let
\(\calF'\subset\calB\) be finite and \(\varepsilon'>0\).
The same fixed-point argument shows that
\(\psi|_{\calB^\gamma}\) is pure. Since \(\calB^\gamma\) is simple,
the GNS representations of \(\varphi|_{\calB^\gamma}\) and
\(\psi|_{\calB^\gamma}\) have the same kernel. The ordinary
homogeneity theorem
\cite[Theorem~2.1]{KishimotoOzawaSakai2003Homogeneity}, applied to
\(\calB^\gamma\), therefore shows that
\(\psi|_{\calB^\gamma}\) is the image of
\(\varphi|_{\calB^\gamma}\) under an automorphism which is
asymptotically inner through a path in \(\calU(\calB^\gamma)\).

It follows that there is a unitary \(w\in\calU(\calB^\gamma)\) such
that, for
\eq{
    \eta
    :=
    \pi(w)\Omega
    \in
    \calH^G,
}
the associated vector state
\eq{
    \varphi_\eta(a)
    :=
    \ip{\eta}{\pi(a)\eta}
}
satisfies
\eq{
    \left|
        \varphi_\eta(d)-\psi(d)
    \right|
    &<
    \frac{\delta}{2}
    \qquad
    (d\in\calG_0),\\
    \left|
        \varphi_\eta(a)-\psi(a)
    \right|
    &<
    \varepsilon'
    \qquad
    (a\in\calF').
}
For the second estimate, we use that both states are
\(\gamma\)-invariant and hence
\eq{
    \varphi_\eta(a)
    &=
    \varphi_\eta(E_\gamma(a)),&
    \psi(a)
    &=
    \psi(E_\gamma(a)).
}
Combining the first estimate with the assumed closeness of
\(\psi\) and \(\varphi\) on \(\calG_0\) gives
\eq{
    \left|
        \varphi_\eta(d)-\varphi(d)
    \right|
    <
    \delta
    \qquad
    (d\in\calG_0).
}
The fixed-vector statement therefore gives a path
\(t\mapsto u_t\) in \(\calU(\calB^\gamma)\), with
\(u_0=1_{\calB^\gamma}\), such that
\eq{
    \varphi_\eta
    =
    \varphi\circ\Ad{u_1},
}
and
\eq{
    \norm{
        \Ad{u_t}(a)-a
    }
    <
    \varepsilon
    \qquad
    (a\in\calF,\ t\in[0,1]).
}
Consequently,
\eq{
    \left|
        \psi(a)
        -
        \br{
            \varphi\circ\Ad{u_1}
        }(a)
    \right|
    <
    \varepsilon'
    \qquad
    (a\in\calF').
}
This proves
\cref{lem:equivariant local adjustment}.
\end{proof}
%%%%%%%%%%%%%%%%%%%%%%%%%%%%%%%%%%%%%%%%%%%%%%%%%%%%%%%%%%%%%%%%%%%%%%%%%%%%%%%%%%%%
\newpage
\begingroup
\let\itshape\upshape
\printbibliography
\endgroup
\end{document}